\documentclass[11pt,reqno]{amsart}
\usepackage{amssymb} 
\usepackage{hyperref}
\usepackage{enumerate}
\usepackage{color}
\usepackage[export]{adjustbox}
\numberwithin{equation}{section}

\newtheorem{theorem}{Theorem}[section]
\newtheorem{definition}[theorem]{Definition}
\newtheorem{proposition}[theorem]{Proposition}
\newtheorem{corollary}[theorem]{Corollary}
\newtheorem{lemma}[theorem]{Lemma}
\newtheorem{remark}[theorem]{Remark}
\newtheorem{example}[theorem]{Example}

\DeclareMathOperator{\trace}{Tr}

\DeclareMathOperator{\id}{id}
\DeclareMathOperator{\Wg}{Wg}
\DeclareMathOperator{\Mob}{Mob}

\DeclareMathOperator*{\Ex}{\mathbb{E}}

\newcommand{\eq}[1]{\begin{align}#1\end{align}}

\def\be{\begin{eqnarray}}
\def\ee{\end{eqnarray}} 

\begin{document}

\date{\today}

\title[Additivity rates and PPT property for random quantum channels]{Additivity rates and PPT property for\\random quantum channels}

\author{Motohisa Fukuda} 
\address{Zentrum Mathematik, M5\\ Technische Universit\"at M\"unchen\\ Boltzmannstrasse 3, 85748 Garching\\ Germany}
\email{m.fukuda@tum.de} 

\author{Ion Nechita} 
\address{Zentrum Mathematik, M5\\ Technische Universit\"at M\"unchen\\ Boltzmannstrasse 3, 85748 Garching\\ Germany; CNRS, Laboratoire de Physique Th\'{e}orique, IRSAMC\\ Universit\'{e} de Toulouse, UPS\\ F-31062 Toulouse, France}
\email{nechita@irsamc.ups-tlse.fr} 

\subjclass[2000]{}
\keywords{}

\begin{abstract}
Inspired by Montanaro \cite{mon}, we introduce the concept of \emph{additivity rates} of a quantum channel $L$, which give the first order (linear) term of the minimum output $p$-R\'enyi entropies of $L^{\otimes r}$ as functions of $r$. We lower bound the additivity rates of arbitrary quantum channels using the operator norms of several interesting matrices including  partially transposed Choi matrices.
As a direct consequence, we obtain upper bounds for the classical capacity of the channels. We study these matrices for random quantum channels defined by random subspaces of a bipartite tensor product space. A detailed spectral analysis of the relevant random matrix models is performed, and strong convergence towards free probabilistic limits is shown. As a corollary, we compute the threshold for random quantum channels to have the positive partial transpose (PPT) property. We then show that a class of random PPT channels violate generically additivity of the $p$-R\'enyi entropy for all $p\geq30.95$.
\end{abstract}

\maketitle

\tableofcontents

\section{Introduction}
In this paper, we focus on three questions related to additivity properties of quantum channels. First, we introduce the concept of additivity rates by which we can bound additivity violations for tensor powers of channels. Then, we use these results to upper bound the classical capacity of quantum channels. Finally, we prove the existence of PPT quantum channels violating the additivity of minimum output $p$-R\'enyi entropy. In the following three subsections, we introduce the above  questions and present our main results. 

\subsection{Additivity rates of quantum channels} \label{sec:additivity-violation}
One of the most important conjectures in quantum information theory had been the additivity of $p$-R\'enyi entropy:
for any (quantum) channels $ L_{1,2}$, 
\be
H^{\min}_p ( L_1  \otimes L_2) = H^{\min}_p ( L_1 ) + H^{\min}_p (L_2) 
\ee
for $1\leq p\leq\infty$.
Here, $H_p^{\min}(\cdot)$ is defined for a quantum channel $ L $ by 
\be
H_p^{\min}( L ) = \min_X H_p ( L (X)) 
\ee
where $X$ runs over all the quantum states, and
the $p$-R\'enyi entropy $H_p(\cdot)$ is defined by
\be
H_p (X) = \frac 1 {1-p} \log \left(\trace X^p\right).  
\ee
Note that $H_p(\cdot)$ becomes von Neumann entropy as $p\to1$. 
This conjecture was made first for $p=1$ in \cite{KR2001}, 
and then for $p \in (1,\infty)$ in \cite{AHW2000}.
More detailed explanations about additivity questions can be found in \cite{Holevo2006}.

These conjectures were disproved by Hayden and Winter for $1<p<\infty$ \cite{hwi} and by Hastings for $p=1$ \cite{has}. The case $p=0$, and $p$ close to $0$, was disproved in \cite{chl}.
Importantly, the violation in the case $p=1$ implies that we can increase the classical capacity of some quantum channels by using entangled inputs \cite{Shor}. 
Then, an important question comes to our mind: how much one can increase the classical capacity by using entanglement over as many quantum states as possible. Although this question on classical capacity should be most important, it is difficult to treat it directly.  
On the other hand, approaches via R\'enyi entropy only involve eigenvalues of matrices and this fact enables us to use random matrix and free probability to investigate the generic behavior of random quantum channels on this issue. 
In this paper, for natural class of random quantum channels,
 we bound additivity violation of R\'enyi entropy. 

In fact, Montanaro \cite{mon} investigated on the limit of additivity violation for $p=\infty$ and  extended the result to $1\leq p < \infty$ by using the monotonicity of the Schatten $p$-norms in $p$. 
In our paper, we study this problem first for $p=2$ and then extend it to $0 \leq p < 2$.
His paper and ours both depend on estimate of norm of random matrices. 
However, those two random matrices are different and give different estimates. 
Detailed discussions on this matter are made in Section \ref{sec:compare} and Section \ref{sec:Montanaro}. 

Finally, our informal main theorem on limitation of additivity violation can be stated as follows. 
\begin{theorem} Consider a sequence of random quantum channels $L_n:\mathcal M_d(\mathbb C) \to \mathcal M_k(\mathbb C)$, defined via random embeddings of $\mathbb C^d$ into $\mathbb C^n \otimes \mathbb C^k$, where $k$ is a fixed parameter and $d \sim tnk$ for a fixed $t \in (0,1)$.
Then, almost surely as $n \to \infty$, for all $p\in[0,2]$, there exist constants  $\alpha_p \in [0,1]$ such that, for all $r \geq 1$,
\begin{equation}
\frac 1 r H_p^{\min} (L_n^{\otimes r}) \geq  \alpha_p H_p^{\min}(L_n).
\end{equation}
The constants $\alpha_p$ satisfy the following relations
\begin{enumerate}
\item When $0<t<1/2$ is a constant,
\be \alpha_p = o(1) + \frac{p-1}{2p} \left [ 1 + \frac {2 \log 2 + \log (1-t)}{\log t} \right] \cdot \mathbf{1}_{(1,2]}(p).
\ee
\item When $k$ is large and $t \asymp k^{-\tau}$ with $\tau >0$,
\be \alpha^\Gamma_{p,k,t} = o(1) +  \begin{cases}
 \frac{p-1}{2p} &\quad \text{ if }  0< \tau \leq 1-1/p\\
 \tau  /2 &\quad \text{ if } 1-1/p \leq \tau \leq 2\\
1 &\quad \text{ if }  \tau \geq 2.
\end{cases} \ee
\end{enumerate}
The above statements hold for the complementary channels $L^C: \mathcal M_d(\mathbb C) \to \mathcal M_n(\mathbb C)$,
where  the roles of $\mathbb C^k$ and $\mathbb C^n$ are swapped.
\end{theorem}

In the result above, the larger the constant $\alpha_p$ is, the more restrictive the additivity violation is. 
A precise definition of additivity rates is given in Definition \ref{def:wad} and more detailed estimates on $\alpha$ are made in Theorem \ref{thm:wad}. 
Also, our model of random quantum channels together with the idea of complementarity is described in detail at the beginning of Section \ref{sec:channel}.

\subsection{Range of capacity}\label{sec:introC}
In \cite{ Holevo1998, SchumacherWestmoreland1997}, 
the Holevo capacity of quantum channels, denoted by $\chi(\cdot)$, 
was proven to be the capacity of transmitting classical information without entangled inputs.
Here, $\chi(\cdot)$ is defined for quantum channels $L$:
\begin{equation}
\chi(L) = \max_{\{p_i,X_i\}} \left[ H_1\left(\sum_i p_i L(X_i)\right ) -\sum_i p_iH_1( L(X_i) ) \right],
\end{equation}
where the $(p_i,X_i)$ are all possible ensembles with $(p_i)$ and $(X_i)$ being a probability distribution and quantum states, respectively. 
It is not difficult to see that
\begin{equation}\label{eq:H-max-H-min}
\chi(L) \leq H_1^{\max}(L) - H_1^{\min} (L).
\end{equation}
Here, $H_1^{\max}(L) = \max_{X}H_1(L(X))$; note that this quantity is trivially additive, $H_1^{\max}(L^{\otimes r}) = r H_1^{\max}(L)$.
The equality is saturated, for example, if the channel $L$ has covariant property \cite{Holevo2005}, and
it is also the case in our setting for a similar reason, which will be discussed later. 
Since the classical capacity, denoted by $\mathcal C_{cl}(\cdot)$ is obtained by regularizing Holevo capacity $\chi(\cdot)$  \cite{ Holevo1998, SchumacherWestmoreland1997},
we can show the following estimate.
\begin{theorem}
Consider a sequence of random quantum channels $L_n:\mathcal M_d(\mathbb C) \to \mathcal M_k(\mathbb C)$, defined via random embeddings of $\mathbb C^d$ into $\mathbb C^n \otimes \mathbb C^k$, where $k$ is a fixed parameter and $d \sim tnk$ for a fixed $t \in (0,1)$.
Then, almost surely in the regime $1 \ll k \ll n$, the classical capacity is asymptotically bounded as follows:
\begin{enumerate}[i)]
\item When $0<t <1/2$ is a constant, we have  
\be
 \limsup_{n\to\infty} \mathcal C_{cl}(L_n)  \leq \log k + \log 2 + \frac 12 \log t(1-t) + o(1).
\ee
\item 
When $t \asymp k^{-\tau}$  and $0< \tau \leq 2$, we have, for some constant $c>0$
\be
\limsup_{n\to\infty} \mathcal C_{cl}(L_n)   \leq  \left(1-\frac \tau 2\right) \log k + c .
\ee
\item When $t \asymp k^{-\tau}$  and $\tau > 2$, the classical capacity is almost surely bounded by a constant. 
\end{enumerate} 

\end{theorem}

\subsection{PPT property and additivity violation}
Another topic treated in this current paper is the \emph{positive partial transpose} property (PPT) for quantum channels;
a quantum channel $L$ is called \emph{PPT} iff the partial transposition of its Choi matrix is positive semi-definite. 
The importance of PPT channels stems from their recent use in the proofs of superactivation for the quantum capacity, see \cite{sya}.
Hence, it is interesting to investigate typical PPT/non-PPT property for random quantum channels.
Also, we show that there exist PPT channels which violate additivity of R\'enyi $p$ entropy. 
This result is interesting because all entanglement-breaking channels are proven to be additive \cite{King03, Shor02}.  
Note that the set of entanglement-breaking channels is contained by the set of PPT channels 
but for qubit channels, these sets are the same. 
The above two problems are investigated in Section \ref{sec:PPT}, and we obtain the following results. 

First, in Section \ref{sec:PPT-t} we have
\begin{theorem}
Consider a sequence $L_n$ of random quantum channels of parameters $k,t$, and let 
\begin{equation}
t_{PPT} = \frac{1}{2} \left( 1- \sqrt{1-\frac{1}{k^2}}\right).
\end{equation}
If $t \in (0,t_{PPT})$ then, almost surely as $n \to \infty$, the sequence $L_n$ has the PPT property, whereas if $t \in (t_{PPT}, 1)$, then, almost surely, the sequence $L_n$ does not have the PPT property. We say that the value $t_{PPT}$ is a \emph{threshold} for the PPT property of random quantum channels.
\end{theorem}

Second, we have
\begin{theorem}
 Consider a sequence of random quantum channels $L_n:\mathcal M_d(\mathbb C) \to \mathcal M_k(\mathbb C)$, defined via random embeddings of $\mathbb C^d$ into $\mathbb C^n \otimes \mathbb C^k$ and let $d\sim \frac n{4k}$. 
Suppose one of the following two procedures are made:
\begin{itemize}
\item fix $k \geq76$ and take large enough $p$ and $n$, or
\item fix $p \geq 30.95$ and take large enough $n$ and $k$
\end{itemize}
then, typically $L_n$ are PPT and violate additivity:
\be
H^{\min}_p (L_n \otimes \bar L_n) < 2H^{\min}_p (L_n).
\ee
\end{theorem} 
This theorem is divided into two theorems: Theorem \ref{thm:ppt1} and Theorem \ref{thm:ppt2} in Section \ref{sec:PPT-v}.

\subsection{Structure of the paper}

The paper is divided roughly into two parts: Sections \ref{sec:additivity-rate} and \ref{sec:additive-bounds} deal with the general theory of additivity rates and their lower bounds, while Sections \ref{sec:random-Choi}-\ref{sec:PPT} deal with random quantum channels.

More precisely, after recalling some basic definitions and results in Section \ref{sec:preliminaries}, we introduce in Section \ref{sec:additivity-rate} one of the main topics of this paper: additivity rates of quantum channels. Then, in Section \ref{sec:additive-bounds}, we introduce lower bounds for minimum output R\'enyi entropies, which are additive with respect to tensor products. These results can be used to lower bound the additivity rates. In Sections \ref{sec:random-Choi} and \ref{sec:other-bounds-random-channels} we study these bounds for random quantum channels. After recalling some known results about the minimum output entropies of random quantum channels in Section \ref{sec:MOE}, we give lower bounds for the additivity rates of random quantum channels in Section \ref{sec:additivity-rates-random}, limiting the possible violations of the additivity of the minimum output entropies of these channels. Based on previous results, we present in Section \ref{sec:cap} upper bounds for the classical capacitiy of (random) quantum channels, and in Section \ref{sec:PPT} examples of random channels which are PPT and violate the additivity of the minimum R\'enyi output entropies.

\section{Preliminaries}
\label{sec:preliminaries}
In this section, we go through basic definitions and knowledge, which are needed through this current paper. 
We give definitions on quantum states, channels and entropy in Section \ref{sec:channel}, 
and then make quick overviews on graphical Weingarten calculus and free probability in Section \ref{sec:graphical-Weingarten}  and Sec \ref{sec:fp},
respectively, as much as we need.

Let us start by introducing some notation. In this paper, the operator $\mathrm{Tr}$ denotes the usual, unnormalized trace. The reader may choose $\log$ to denote the logarithm in basis $2$ or $e$, depending on her/his background. Finally, we use the following asymptotic notations for sequences:
\begin{align}
x_n \sim y_n \quad &\iff \quad \lim_{n \to \infty} x_n / y_n =1\\
x_n  \asymp y_n \quad &\iff \quad 0 < \liminf_{n \to \infty} x_n / y_n \leq \limsup_{n \to \infty} x_n / y_n < \infty.
\end{align} 
\subsection{Quantum states and channels}
\label{sec:channel}
In this paper we will consider quantum channels $L:\mathcal M_d(\mathbb C) \to \mathcal M_k(\mathbb C)$, defined via the Stinespring dilation \cite{spring}
\begin{equation}\label{eq:quantum-channel-stinesrping}
L(X) = [\operatorname{Tr}_{\mathbb C^n} \otimes \operatorname{id}_{\mathbb C^k}](VXV^*),
\end{equation}
for an isometry $V: \mathbb C^d \to \mathbb C^n \otimes \mathbb C^k$. 
In this case, the dimensions of input, output and environment spaces are $d$, $k$ and $n$, respectively. 
If we swap the roles of $\mathbb C^k$ and $\mathbb C^n$, we get another channel $L^C$, called the \emph{complementary} channel \cite{Holevo2005a,KMNR2007}.
Outputs of this channel $L^C : \mathcal M_d(\mathbb C) \to \mathcal M_n(\mathbb C)$ share non-zero eigenvalues with those of the channel $L$ as long as inputs are pure states.
Hence, our results on entropy bounds are also shared by $L$ and $L^C$.
For such maps $L$ and $L^C$ we know that $L \otimes \mathrm{id}_{\mathbb C^m}$ and $L^C \otimes \mathrm{id}_{\mathbb C^m}$ are positive for any $m \geq 0$.
This property is called \emph{completely positivity}.
Later, we shall consider random quantum channels obtained by choosing the isometry $V$ randomly. 
The probability distribution of the random variable $V$ will be the uniform one on the set of isometries, obtained by truncating a Haar-distributed unitary matrix $U \in \mathcal U(kn)$: $V$ will consist of the first $d$ columns of a $kn \times kn$ random Haar unitary matrix. For now, let us introduce a key concept in this paper, the \emph{Choi matrix} of a quantum channel \cite{choi}. To a channel $L$, we associate its Choi matrix $C_L \in \mathcal M_k(\mathbb C) \otimes \mathcal M_d(\mathbb C)$, defined by 
\begin{equation}\label{eq:choi-matrix}
C_L = [L \otimes \mathrm{id}](E_d) = \sum_{i,j=1}^d L(e_ie_j^*) \otimes e_i e_j^*,
\end{equation}
where $E_d \in \mathcal M_{d^2}(\mathbb C)$ is the (unnormalized) \emph{maximally entangled state}
\begin{equation}\label{eq:maximally-entangled}
E_d = \sum_{i,j=1}^d e_ie_j^* \otimes e_i e_j^*.
\end{equation}
It is a classical result 
that a linear map $L$ is completely positive if and only if its Choi matrix $C_L$ is positive semidefinite. 

Finally, we shall denote 
by $\mathcal M_d^{1,+}(\mathbb C)$ the set of $d$-dimensional \emph{quantum states}
\begin{equation}
\mathcal M_d^{1,+}(\mathbb C) = \{X \in \mathcal M_d(\mathbb C) \, : \, \mathrm{Tr} X = 1 \text{ and } X\geq 0\}.
\end{equation}

Let us now introduce the entropic quantities we are interested in. The \emph{Shannon} entropy of a probability vector $x \in \mathbb R^k$, $x_i \geq 0$, $\sum_i x_i=1$ is defined by
\begin{equation}
H(x) = -\sum_{i=1}^k x_i \log x_i,
\end{equation}
where we put $0 \log 0 = 0$. This quantity is extended, via functional calculus, to quantum states, where it is known as the \emph{von Neumann entropy}. The \emph{R\'enyi entropies} are a one-parameter generalizations of these quantities. They are defined for any $p\in [0,\infty]$, as follows:
\begin{align}
H_p(x) = 
\begin{cases}
\log \#\{i \, : \, x_i \neq 0\}, &\quad \text{ if } p=0\\
H(x), &\quad \text{ if } p=1\\
\displaystyle \frac{1}{1-p}\log \sum_{i=1}^k x_i^p, &\quad \text{ if } p\neq 0,1, \infty\\
-\log \max\{x_i\}, &\quad \text{ if } p=\infty.
\end{cases}
\end{align}
The same quantities are defined for quantum states, and satisfy $H_p(X) \in [0, \log k]$. In what follows, we shall use the following well-known result \cite{BS-book}:

\begin{lemma}\label{lem:H-decreasing-in-p}
For a fixed probability vector $x$ (resp.~quantum state $X$), the function 
\begin{equation}
[0, \infty] \ni p \mapsto H_p(x)
\end{equation}
(resp.~$p \mapsto H_p(X)$) is non-increasing in $p$. In a similar manner, for a fixed quantum channel $L$, the function $p \mapsto H_p^{\min}(L)$ is decreasing in $p$.
\end{lemma}

Recall from the introduction that the \emph{minimum output $p$-R\'enyi entropy} of a quantum channel $L$ is defined by
\begin{equation}\label{eq:def-H-p-min}
H_p^{\min}(L) = \min_{X \in \mathcal M_d^{1,+}(\mathbb C)} H_p(L(X)),
\end{equation}
for an arbitrary entropy parameter $p \in [0, \infty]$. The $H_p^{\min}$ functionals are \emph{sub-additive}, in the sense that for any quantum channels $L, K$, we have
\begin{equation}\label{eq:H-p-min-sub-additive}
H_p^{\min}(L \otimes K) \leq H_p^{\min}(L) + H_p^{\min}(K).
\end{equation}

For a pair of quantum channels $(L,K)$ such that $L \otimes K$ has no pure outputs, define the \emph{relative violation of minimum output $p$-entropy} of the pair $(L,K)$ by
\begin{equation}\label{eq:def-relative-violation}
v_p(L,K) := \frac{H_p^{\min}(L)+H_p^{\min}(K)}{H_p^{\min}(L \otimes K)} \in [1, \infty].
\end{equation}
With this notation, we call the pair $(L,K)$ $p$-additive iff.~ $v_p(L,K) = 1$.

\subsection{The graphical Weingarten integration formula}
\label{sec:graphical-Weingarten}

The model of random quantum channels we are interested in involves random isometries, which can be seen as truncated random Haar unitary matrices. Since our approach to understanding statistics of such channels is the moment method, we shall compute integrals of polynomials in the entries of unitary matrices. The main result here is the Weingarten formula, which was introduced by Weingarten \cite{wei+} in the physics literature and rigorously developed by Collins \cite{col}, and Collins and {\'S}niady \cite{csn}.

\begin{theorem}
\label{thm:Wg}
 Let $n$ be a positive integer and
$i=(i_1,\ldots ,i_p)$, $i'=(i'_1,\ldots ,i'_p)$,
$j=(j_1,\ldots ,j_p)$, $j'=(j'_1,\ldots ,j'_p)$
be $p$-tuples of positive integers from $\{1, 2, \ldots, n\}$. Then, the following integral over the Haar measure of $\mathcal U_n$ can be evaluated as 
\begin{equation}\label{eq:Wg} 
\int_{\mathcal U_n}U_{i_1j_1} \cdots U_{i_pj_p}
\bar U_{i'_1j'_1} \cdots
\bar U_{i'_pj'_p}\ dU=
\sum_{\alpha, \beta\in \mathcal S_{p}}\delta_{i_1i'_{\alpha (1)}}\ldots
\delta_{i_p i'_{\alpha (p)}}\delta_{j_1j'_{\beta (1)}}\ldots
\delta_{j_p j'_{\beta (p)}} \Wg (n,\alpha^{-1}\beta),
\end{equation}
where the function $\Wg$ is called the \emph{Weingarten function} (see the next definition).
If $p\neq p'$ then
\begin{equation} \label{eq:Wg_diff} 
\int_{\mathcal U(n)}U_{i_{1}j_{1}} \cdots
U_{i_{p}j_{p}} \bar U_{i'_{1}j'_{1}} \cdots
\bar U_{i'_{p'}j'_{p'}} \ dU= 0.
\end{equation}
\end{theorem}
  
For a permutation $\sigma \in \mathcal S_p$, $\# \sigma$ denotes the number of cycles of $\sigma$, and  $|\sigma |$ is the length of $\sigma$, i.e. the minimal number of transpositions that multiply to $\sigma$. Note that the length function defines a distance on $\mathcal S_p$, via $d(\sigma, \tau) = |\sigma^{-1}\tau|$. Let us recall the definition of the unitary Weingarten function.
\begin{definition}\label{def:Wg}
The unitary Weingarten function 
$\Wg(n,\sigma)$
is a combinatorial function which depends on a dimension parameter $n$ and on a permutation $\sigma$
in the symmetric group $\mathcal S_p$. 
It is the inverse of the function $\sigma \mapsto n^{\#  \sigma}$ with respect to the following convolution operation:
\begin{equation}
\forall \sigma,\pi \in \mathcal S_p, \quad \sum_{\tau \in \mathcal S_p} \Wg(n,\sigma^{-1}\tau)n^{\#(\tau^{-1}\pi)} = \delta_{\sigma, \pi}.
\end{equation}
In the large $n$ limit ($p$ is being kept fixed), it has the following asymptotics 
\begin{equation}\label{eq:Wg-asympt}
\Wg(n,\sigma) = n^{-(p + |\sigma|)} (\Mob(\sigma) + O(n^{-2})),
\end{equation}
where the M\"{o}bius function $\mathrm{Mob}$ is multiplicative on the cycles of $\sigma$ and its value on a $r$-cycle is
\be(-1)^{r-1} \mathrm{Cat}_{r-1},\ee
where $\mathrm{Cat_r}$ are the Catalan numbers.
Note that we omit the dimension in the Weingarten function when there is no confusion and write $\Wg(\sigma)$ for $\Wg(n,\sigma)$.
Also, we use the notation $\Mob(\alpha,\beta)=\Mob(\alpha^{-1}\beta)$.
\end{definition}

When applying the above integration formula, especially in the cases where the degree of the polynomial to be integrated is high, one has to deal often with sums indexed by a large set of indices. Computing such sums is a tedious task, so we use here a graphical formulation of the Weingarten formula, introduced in \cite{cne10a}. Here we sketch the main ideas, referring the reader to original work \cite{cne10a} for the technical details. This method has been used recently in relation to channel capacities \cite{cne10b,cfn12,cfn13,fne}, entanglement theory \cite{ane,cnz} and condensed matter physics \cite{cgp}.

The Weingarten graphical calculus builds up on the \emph{tensor diagrams} introduced by theoretical physicists and adds to it the ability to perform averages over diagrams containing Haar unitary matrices. In the graphical formalism, tensors (vectors, linear forms, matrices, etc.) are represented by \emph{boxes}, see Figure \ref{fig:AptA}, left diagram. To each box, one attaches labels of different shapes, corresponding to vector spaces. The labels can be  filled (black) or empty (white) corresponding to spaces or their duals: a $(p,q)$-tensor will be represented by a box with $p$ black labels and $q$ white labels attached. The example in Figure \ref{fig:AptA}, left corresponds to a (square) matrix $A \in \mathcal M_n(\mathbb C) \otimes \mathcal M_k(\mathbb C)$. 

Besides boxes, our diagrams contain \emph{wires}, which connect the labels attached to boxes. Each wire corresponds to a tensor contraction between a vector space $V$ and its dual $V^*$ ($V \times V^* \to \mathbb C$). See Figure \ref{fig:AptA}  for the example of the partial trace. A \emph{diagram} is simply a collection of such boxes and wires and corresponds to an element in a tensor product space (which might be degenerate, i.e.~ the scalars $\mathbb C$).

\begin{figure}[htbp] 
\includegraphics[valign=t]{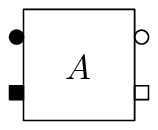} \qquad\qquad\qquad \includegraphics[valign=t]{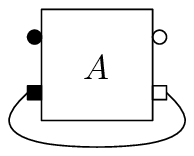}
\caption{Diagram for a matrix $A$ acting on a tensor product space and for its partial trace $[\mathrm{id} \otimes \mathrm{Tr}](A)$, obtained by contracting with a wire the labels corresponding to the second tensor factor.} 
\label{fig:AptA}
\end{figure}

Let us now describe briefly how one computes expectation values of such diagrams containing boxes $U$ corresponding to Haar-distributed random unitary matrices. The idea in \cite{cne10a} was to implement in the graphical formalism the Weingarten formula in Theorem \ref{thm:Wg}. Each pair of permutations $(\alpha,\beta)$ in  (\ref{eq:Wg}) will be used to eliminate $U$ and $\bar U$ boxes and wires will be added between the black, resp. white, labels of the box $U$ with index $i$ and the black, resp. white, labels of the box $\bar U$ with index $\alpha(i)$, resp. $\beta(i)$. In this way, for each pair of permutations, one obtains a new diagram $\mathcal D_{\alpha,\beta}$, called a \emph{removal} of the original diagram corresponding to $\alpha, \beta$. The graphical Weingarten formula is described in the following theorem \cite{cne10a}.

\begin{theorem}\label{thm:graphical-Wg}
If $\mathcal D$ is a diagram containing boxes $U, \bar U$ corresponding to a Haar-distributed random unitary matrix $U \in \mathcal U(n)$, the expectation value of $\mathcal D$ with respect to $U$ can be decomposed as a sum of removal diagrams $\mathcal D_{\alpha, \beta}$, weighted by Weingarten functions:
\begin{equation}\label{eq:graphical-Wg}
\mathbb E_U(\mathcal D)=\sum_{\alpha, \beta} \mathcal D_{\alpha, \beta} \Wg (n, \alpha^{-1}\beta).
\end{equation}
\end{theorem}

Since the above Weingarten formula is written as a sum over permutations we review next some additional properties of the symmetric group endowed with the distance $d(\alpha,\beta)= |\alpha^{-1}\beta|$. The function $|\cdot|$ has nice properties, for example, $|\alpha^{-1}|=|\alpha|$ and $|\alpha\beta|=|\beta\alpha|$; we refer the readers to \cite{nsp} for more details. 
For three permutations $\alpha, \beta,\gamma$ the triangle inequality holds:
\be
|\alpha^{-1}\gamma| \leq |\alpha^{-1}\beta| + |\beta^{-1}\gamma|.
\ee
When the bound above is saturated, we say that $\beta$ is on a \emph{geodesic} between $\alpha$ and $\gamma$, 
and write $\alpha - \beta - \gamma$. When $\gamma$ is the full cycle permutation, $\gamma = (p \, \cdots \, 3 \, 2\, 1)$, permutations lying on the geodesic  between $\mathrm{id}$ and $\gamma$ are simply called \emph{geodesics permutations}. In \cite[Proposition 23.23]{nsp}, it is shown that geodesic permutations are in bijection with \emph{non-crossing partitions}. Recall that a partition $\pi$ of $\{1, \ldots, p\}$ is said to be non-crossing if 
\begin{equation}
\forall \, i<j<k<l , \quad i \sim k \text{ and } j \sim l \implies i \sim j \sim k \sim l,
\end{equation}
where $\sim$ denotes the equivalence relation on $\{1, \ldots, p\}$ induced by $\pi$. We denote by $NC(p)$ the set of non-crossing partitions on $p$ elements. Moreover, the isomorphism between geodesic permutations and non-crossing partitions respects many combinatorial properties of the objects, such as the number of cycles (resp. blocks); see \cite[Section 23]{nsp} for more details.

\subsection{Some elements of free probability}
\label{sec:fp}

An excellent reference for the theory of free probability is \cite{nsp}; we recall now only some basic facts from this theory needed in the current paper. 

A  $C^*$ probability space is a unital $C^*$-algebra $\mathcal A$
equipped with a state $\tau$, which gives a norm;
$\|a\|_\tau = \lim_{p \to \infty} (\tau(a^p))^{1/p} $.
Such a  $C^*$ probability space is denoted by $(\mathcal A, \tau,\| \cdot \|_{\tau})$. 

The convergence of the eigenvalues of random matrices can be stated in the language of $C^*$ probability spaces as follows. Note that we define two types of convergence: the convergence in distribution (which is the convergence of all moments) and the strong convergence (which implies, in particular, the convergence of the extreme eigenvalues of the matrices). In this paper, we are interested in the operator norms of random matrices, and the usual convergence in distribution does not guarantee the convergence of the norms in the case when the size of the matrices grows; hence we shall make use of the strong convergence. 

\begin{definition}\label{def:strong-convergence}
Suppose we have $C^*$-probability spaces: 
$(\mathcal A, \tau,\| \cdot \|_{\tau})$ and $(\mathcal A_N, \tau_N, \| \cdot \|_{\tau_N})$ with $N \in \mathbb N$, 
where $\tau$ and $\tau_N$ are faithful traces. 
For $l$-tuple elements $a=(a_1,\ldots, a_l)$ in $\mathcal A$ and $a^{(N)} = (a_1^{(N)},\ldots, a_l^{(N)})$ in $\mathcal A^{(N)}$,
\begin{enumerate}[i)]
\item we say $a^{(N)}$ converges to $a$ in distribution if 
\be
\lim_{N \to \infty} \tau_N[P(a^{(N)},a^{(N) *})] = \tau [P(a,a^*)],
\ee
\item we say $a^{(N)}$ converges to $a$ strongly in distribution if in addition
\be
\lim_{N \to \infty}\| [P(a^{(N)},a^{(N) *})\|_{\tau_N} = \|P(a,a^*)\|_\tau .
\ee
\end{enumerate}
Here, $P$ is any polynomial in non-commuting $2l$ variables.
\end{definition}

The strong asymptotic freeness of random unitary matrices and deterministic matrices has been proven by Collins and Male: 
\begin{proposition}[\cite{collinsmale}] \label{prop:CM}
Suppose we have $C^*$-probability spaces: 
$(\mathcal A, \tau,\| \cdot \|_{\tau})$ and $(\mathcal M_N(\mathbb C), \tau_N, \| \cdot \|_{\tau_N})$ with $N \in \mathbb N$.
Here, $\tau$ is a faithful trace and $\tau_N$ is the usual normalized trace on the $N \times N$ matrix space $\mathcal M_N(\mathbb C)$.
Take 
\begin{itemize}
\item a $p$-tuple of  free Haar unitary elements $u= (u_1, \ldots, u_p)$ in $\mathcal A$, and
\item a $p$-tuple of i.i.d.~ Haar-distributed unitary matrices $U^{(N)}=(U_1^{(N)}, \ldots, U_p^{(N)})$ in $\mathcal M_N(\mathbb C)$.
\end{itemize}
Suppose we are given
\begin{itemize}
\item a $q$-tuple of elements $y=(y_1, \ldots, y_q)$ free from $u$ in $\mathcal A$, and
\item a $q$-tuple of matrices $Y^{(N)}=(Y_1^{(N)},\ldots, Y_q^{(N)})$ independent from $U^{(N)}$ in $\mathcal M_N(\mathbb C)$.
\end{itemize}
such that $Y^{(N)}$ converges to $y$ strongly in distribution.
Then, almost surely $(U^{(N)}, Y^{(N)})$ converges to $(u,y)$ strongly in distribution. 
\end{proposition}

The following useful statement was proven by Male:
\begin{proposition}[Proposition 7.3 in \cite{male}]\label{prop:M}
Suppose we have $C^*$-probability spaces: 
$(\mathcal A, \tau,\| \cdot \|_{\tau})$ and $(\mathcal A_N, \tau_N, \| \cdot \|_{\tau_N})$ with $N \in \mathbb N$, 
where $\tau$ and $\tau_N$ are faithful traces.
Take
\begin{itemize}
\item a $l$-tuple of self-adjoint  elements $z=(z_1,\ldots,z_l)$ in $\mathcal A$, and
\item a $l$-tuple of self-adjoint elements $z^{(N)}=(z_1^{(N)},\ldots,z_l^{(N)})$ in $\mathcal A_N$.
\end{itemize}
If we assume that $z^{(N)}$ converges to $z$ strongly in distribution,
then we have strong convergence in the following sense: 
for any polynomial $P$ in $l$ non-commuting variables with coefficients in $\mathcal M_k (\mathbb C)$,
\be
\lim_{N \to \infty} \| P(z^{(N)}) \|_{\tau_k \otimes \tau_N}=\| P(z) \|_{\tau_k \otimes \tau} 
\ee
\end{proposition}
Note that in the above result, one can drop the self-adjointness assumption by considering real and imaginary parts of the operators involved.

We prove next a simple lemma about push-forwards of free additive convolution powers of probability measures and we recall a well-known result about the free multiplicative convolution product of Bernoulli distributions $b_t  =(1-t) \delta_0 + t \delta_1$. Recall that given two free elements $a,b$ having distributions $\mu, \nu$, the distributions of $a+b$ and respectively $a^{-1/2}ba^{-1/2}$ are denoted by $\mu \boxplus \nu$, respectively $\mu \boxtimes \nu$, and they are called the free additive (resp.~ multiplicative) convolutions of $\mu$ and $\nu$ (for the latter, we require $a \geq 0$). We denote by $f_\#\mu$ the push-forward of a measure $\mu$ by a measurable function $f$: if the random variable $X$ has distribution $\mu$, then $f(X)$ has distribution $f_\#\mu$.

\begin{lemma}\label{lemma:push-forward}
Let $\mu$ be a compactly supported probability measure on $\mathbb R$ so that,
 for any $T\geq1$, $\mu^{\boxplus T}$ is well-defined.
Then, we have, for any  $a,b \in \mathbb R$
\be(  (x \mapsto ax+b)_\#   \mu)^{\boxplus T} =  (x \mapsto ax+Tb)_\#  (\mu^{\boxplus T}).\ee
\end{lemma} 
\begin{proof}
First, let $v_\mu$ be an element in the $C^*$-algebra which gives the probability measure $\mu$ so that
$av_\mu+b$ induces the probability measure $\mu_{ax+b}$. 
Then, first by using multi-linear property of cummulant, 
\be
\kappa_n (av_\mu) = \kappa_n (av_\mu, \ldots, av_\mu) = a^n \kappa_n (v_\mu).
\ee
Also, shift does not change cummulants $\kappa_n$ except for the case when $n=1$: 
$\kappa_1(av+b) = \kappa_1(av)+b$. 
Therefore, 
\be
\kappa_n((av_\mu + b)^{\boxplus T}) = T [a^n \kappa_n(v_\mu) + b\delta_{1,n}   ] = \kappa_n (av_\mu^{\boxplus T}) + Tb \delta_{1,n}
= \kappa_n (av_\mu^{\boxplus T} + Tb).
\ee
This completes the proof. 
\end{proof} 

\begin{proposition}\label{prop:bernoulli-boxtimes-boxplus}
The free multiplicative convolution of two Bernoulli distributions $b_s, b_t$ (with $s,t \in[0,1]$) is given by
\eq{ 
b_s \boxtimes b_t &= (1-\min(s,t))\delta_0 + \max(s+t-1,0) \delta_1  \notag\\
&\qquad +\frac{\sqrt{(\varphi^+(s,t)-x)(x-\varphi^-(s,t))}}{2\pi x(1-x)} \mathbf{1}_{[\varphi^-(s,t),\varphi^+(s,t)]}(x)\,dx,
}
where the bounds of the a.c.~ part of the support are given by
\begin{equation}\label{eq:phi-st-pm}
\varphi^\pm(s,t) = s+t - 2st \pm 2 \sqrt{st(1-s)(1-t)}.
\end{equation}
Equivalently, for any $T \geq 1$,
\eq{
 b_s^{\boxplus T} &= \max (0,1-Ts) \delta_0 + \max(0,1-T(1-s)) \delta_{T} \notag\\
&\qquad +\frac {T \sqrt{(\gamma^+(s,T)-x)(x-\gamma^-(s,T))}}{2\pi x (T-x)} \mathbf 1_{[\gamma^-(s,T),\gamma^+(s,T)]}(x)  \, dx,
}
where $\gamma^{\pm}(s,T) = (T-2)s +1 \pm 2 \sqrt{(T-1)s(1-s)}$.
Note that $ \varphi^{\pm}(s,t) =t\gamma^\pm(s,1/t)$.
\end{proposition}
\begin{proof}
The first claim is taken from \cite[Example 3.6.7]{vdn}, while the second one follows from the fact that (see \cite[Exercise 14.21]{nsp})
\be (1- T^{-1} )\delta_0 +  T^{-1}  b_s^{\boxplus T} = b_s \boxtimes \left( (1- T^{-1} ) \delta_0 +  T^{-1} \delta_T \right).\ee
\end{proof}

\section{Additivity rates for quantum channels}\label{sec:additivity-rate}

In this section, we discuss how the functional $H_p^{\min}(\cdot)$ behaves with respect to tensor products. The ideas and results developed here will be applied to random quantum channels later. Our inspiration comes from \cite{mon}, where a multiplicative version of the additivity rates was established. 

\subsection{Definition and basic properties}

First, as is explained in Section \ref{sec:additivity-violation},  because of non-additivity properties of quantum channels,
we do not know how $H^{\min}_p(L^{\otimes r})$ behaves as $r$ grows. 
So, we introduce a notion which quantifies $H^{\min}_p(L^{\otimes r})$ in terms of $H^{\min}_p(L)$:
\begin{definition}\label{def:wad}
For a quantum channel $L$ and an entropy parameter $p \in [0, \infty]$, define the \emph{$p$-additivity rate} of $L$ by
\begin{equation}\label{eq:def-wad}
\alpha_p(L) = \sup\left\{ a \in [0,1] \, :\, \liminf_{r \to \infty} \frac{1}{r}H_p^{\min}(L^{\otimes r}) \geq a H_p^{\min}(L)\right\}.
\end{equation}
\end{definition} 

Different characterizations as well as some basic properties of $p$-additivity rates can readily be obtained from basic properties of the entropy functionals. 

\begin{proposition}\label{prop:additivity-rates}
The $p$-additivity rate of a quantum channel $L$ can be characterized in the following equivalent ways:
\begin{align}
\alpha_p(L) &= \sup\left\{ a \in [0,1] \, :\, \forall r \geq 1, \quad \frac{1}{r}H_p^{\min}(L^{\otimes r}) \geq a H_p^{\min}(L)\right\} \\
& = 
\begin{cases}
\displaystyle \lim_{r \to \infty} \frac{H_p^{\min}(L^{\otimes r})}{r H_p^{\min}(L)} &,\quad \text{ if } H_p^{\min}(L) >0 \\
1 &,\quad \text{ if } H_p^{\min}(L) =0.
\end{cases}\\
& = 
\begin{cases}
\displaystyle \inf_{r \geq 1} \frac{H_p^{\min}(L^{\otimes r})}{r H_p^{\min}(L)} &,\quad \text{ if } H_p^{\min}(L) >0 \\
1 &,\quad \text{ if } H_p^{\min}(L) =0.
\end{cases}
\end{align}
\end{proposition} 
\begin{proof}
The statements follow from Fekete's sub-additive lemma \cite[Lemma 1.2.1]{ste} and \eqref{eq:H-p-min-sub-additive}: the $\liminf$ in \eqref{eq:def-wad} is actually a limit and it is equal to the infimum of the sequence $H_p^{\min}(L^{\otimes r}) / r$. The zero entropy case follows from the fact that if the channel $L$ has zero minimum output entropy for some $p$, then the same holds for all tensor powers $L^{\otimes r}$ because $0 \leq H_p^{\min}( L^{\otimes r}) \leq r H_p^{\min} (L) = 0$.
Moreover, such a channel is additive with any channel (see \cite[Lemma 1]{fuk}).
\end{proof} 

\begin{proposition} \label{prop:additivity-rates2}
The $p$-additivity rate functionals have the following set of properties:
\begin{enumerate}
\item The additivity relation $H^{\min}_p(L^{\otimes r}) = r H^{\min}_p(L)$ holds for all $r \geq 1$ if and only if $\alpha_p(L) = 1$. 
\item Monotonicity with respect to tensor powers:
\be\alpha_p(L^{\otimes s}) \geq \alpha_p(L),\ee
for all integer tensor powers $s \geq 1$. 
\item Convex-like behaviour with respect to tensor products:
\be\alpha_p(L \otimes K) \leq v_p(L,K)\left[t \alpha_p(L) + (1-t)\alpha_p(K) \right] \leq v_p(L,K) \max\left\{\alpha_p(L),\alpha_p(K)\right\},\ee
where $v_p(L,K)$ is the relative violation of the minimum $p$-output entropy \eqref{eq:def-relative-violation} and $t  = H_p^{\min}(L) / [H_p^{\min}(L) + H_p^{\min}(K)] \in [0,1]$; if $H_p^{\min}(L) = H_p^{\min}(K) =0$, just put $t=0$. 
\item Additivity violations yield upper bounds:
\begin{equation}\label{eq:additivity-violation-upper-bound-alpha}
\alpha_p(L) \leq \frac{1}{v_p(L,L)}, \quad \forall p \in [0,\infty].
\end{equation}
\end{enumerate}
\end{proposition}
\begin{proof}
The first property follows directly from the definition. For the second one, in the case when $H_p^{\min}(L) >0$, write
\begin{equation}
\alpha_p(L^{\otimes s})  = \lim_{r \to \infty} \frac{H_p^{\min}((L^{\otimes s})^{\otimes r})}{r H_p^{\min}(L^{\otimes s})} \geq \lim_{r \to \infty} \frac{H_p^{\min}(L^{\otimes sr})}{sr H_p^{\min}(L)} = \alpha_p(L).
\end{equation}
The last two statements follow from the definition of the relative violation $v_p$. 
\end{proof}

\begin{remark}
In Proposition \ref{prop:additivity-rates2}, we set an upper bound for $\alpha_p(L)$ by using the relative violation $v_p(L,L)$. 
However, in Section \ref{sec:add-rate}, we lower bound $\alpha_p(L\otimes \bar L)$ by using $v_p(L,\bar L)$,
where $\bar L$ is the complex conjugate of $L$.
\end{remark}

\subsection{Examples: the Werner-Holevo and the antisymmetric channels}

In Proposition \ref{prop:additivity-rates2}, we have seen that any $p$-additive channel $L$ has unit additivity rate $\alpha_p(L)=1$. We discuss next some examples of non-additive channels. Below, we shall denote with $A^\top$ the transposition of a matrix $A$.

\begin{example}\label{ex:WH}
The Werner-Holevo channel $W_d:\mathcal M_d(\mathbb C) \to \mathcal M_d(\mathbb C)$,
\be W_d(X) = \frac{1}{d-1} [ \mathrm{Tr}(X) I_d - X^\top]\ee
is the first known example of a quantum channel that violates the additivity of the minimum $p$-output R\'enyi entropy. In \cite{who}, it has been shown that $W_3$ violates the additivity for any value $p > 4.79$. From \cite{who} we have explicitly
\be v_p(W_d,W_d) \geq \frac{2 \log (d-1)}{\log\left\{ (d^2-1)[(1-2/d)/(d-1)^2]^p + [(2-2/d)/(d-1)^2]^p\right\}/(p-1)},\ee
from which we can infer the following upper bounds for additivity rates (see \eqref{eq:additivity-violation-upper-bound-alpha}):
\be \alpha_5(W_3) \leq \frac{1}{v_5(W_3,W_3)} \leq 0.989 \quad \text{and} \quad \alpha_\infty(W_3) \leq \frac{1}{v_\infty(W_3,W_3)} \leq \frac{\log 3}{\log 4}.\ee
\end{example}

\begin{example}\label{ex:GHP}
In \cite{ghp}, the authors construct explicit counterexamples to the additivity relation, for all values $p>2$, by considering the natural embedding of the anti-symmetric subspace $\Lambda^2(\mathbb C^d)$ into $\mathbb C^d \otimes \mathbb C^d$. This yields a channel $A_d:\mathcal M_{d(d-1)/2}(\mathbb C) \to \mathcal M_d(\mathbb C)$. From \cite{ghp}, one has
\be v_p(A_d,A_d) \geq \frac{2 \log 2}{\frac{p}{p-1}\log\left[2 \frac{d}{d-1} \right]}.\ee
From the above relation, using \eqref{eq:additivity-violation-upper-bound-alpha}, one gets, for example
\be \forall p > 2, \quad \alpha_p(A_d)  \leq \frac{1}{v_p(A_d,A_d)} = \frac{p}{2(p-1)} (1+\log_2[d/(d-1)]) \xrightarrow[d \to \infty]{} \frac{p}{2(p-1)} < 1.\ee
\end{example}

\section{Additive bounds for the R\'enyi entropies via (partial) traces and transpositions}\label{sec:additive-bounds}

In this section we introduce several \emph{additive} bounds for the R\'enyi entropies of quantum channels (we focus on $p=2,\infty$) that we obtain by considering the operator norm of the vectorized version of the isometry defining the channel, after applying one or several traces or transpositions. We perform an exhaustive study of this method, concluding that the method yields 5 non-trivial bounds, including the one studied by Montanaro \cite{mon}. The key point is that the bounds we are providing are additive with respect to tensor powers of channels, so they can be used to bound the additivity rates defined in the previous sections. Interesting bounds for the classical capacity of quantum channels can be obtained from these bounds. 

Recall that the $2$, resp.~ $\infty$-minimum output R\'enyi entropies of a quantum channel $L$ are 
\begin{align}
H_2^{\min}(L) &= \min_X \frac{ \log \mathrm{Tr}\left[L(X)^2 \right]}{1-2} = - \max_X \log \mathrm{Tr}\left[L(X)^2 \right]\\
H_\infty^{\min}(L) &= - \max_X \log \|L(X)\|,
\end{align}
where $X$ runs over all the input quantum states. In what follows, we shall write $\Theta(A) = A^\top$ for the transposition map, which is an involution on matrix algebras. Moreover, for bi-partite matrices $B \in \mathcal M_p(\mathbb C) \otimes \mathcal M_q(\mathbb C)$, we write $B^\Gamma$ for the \emph{partial transposition} of $B$ with respect to the second subsystem, 
\be B^\Gamma = [\mathrm{id}_p \otimes \Theta_q](B).\ee
Equivalently, the partial transposition operation can be defined on simple tensors by $(B_1 \otimes B_2)^\Gamma = B_1 \otimes B_2^\top$.

\subsection{Quantities arising from vectorized isometries}

The starting point of our study is the vectorization of the isometry $V : \mathbb C^d \to \mathbb C^n \otimes \mathbb C^k$ defining the channel $L$ as in \eqref{eq:quantum-channel-stinesrping}. To this isometry we associate its vectorization $v \in \mathbb C^n \otimes \mathbb C^k \otimes \mathbb C^d$ (which is a tripartite tensor) by the relation
\begin{equation}\label{eq:v-V}
v = \sum_{i=1}^n \sum_{j=1}^k \sum_{s=1}^d \langle e_i \otimes f_j, V g_s \rangle e_i \otimes f_j \otimes g_s,
\end{equation}
where $\{e_i\}$, $\{f_j\}$, $\{g_s\}$ are orthonormal bases of respectively $\mathbb C^n$, $\mathbb C^k$, $\mathbb C^d$. The Choi matrix $C_L$ of the channel $L$ (see \eqref{eq:choi-matrix}) is related to the third order tensor $v$ by the partial trace operation:
\be C_L = [\mathrm{Tr}_n \otimes \mathrm{id_k} \otimes \mathrm{id_d}](vv^*).\ee
For a graphical representation of the vectorization $v$ and its relation to the Choi matrix $C_L$, see Figure \ref{fig:v-V}. Note also that $\|v\|^2 = \mathrm{Tr}(V^*V) = d$.

\begin{figure}[htbp] 
\includegraphics{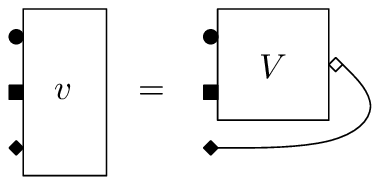} \quad \quad \quad
\includegraphics{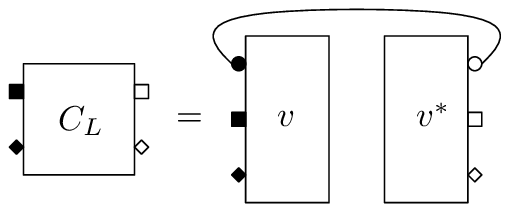}
\caption{The vectorization $v$ of the isometry $V$ defining a quantum channel is an order 3 tensor. The partial trace of the orthogonal projection of $v$ with respect to the first tensor factor gives the Choi matrix of the channel. In both figures, round shaped labels correspond to $\mathbb C^n$, squares  to $\mathbb C^k$, and diamonds to $\mathbb C^d$.} 
\label{fig:v-V}
\end{figure}

We shall now apply different operations to the orthogonal projection on $v$ and take the operator norm of the resulting matrix, in order to obtain a scalar quantity $B(v)$; in the next subsection, we show that some of these quantities are (additive) bounds for the $2$ or the $\infty$  minimum output R\'enyi entropies of quantum channels. We are interested in  three operations: the identity $\mathrm{id}$, the trace $\mathrm{Tr}$, and the transposition $\Theta$. These operators will act on the three tensor factors of $vv^*$, and we shall denote by $B_{QRS} = B_{QRS}(L)$ the quantity
\begin{equation}
B_{QRS} = \| [Q_n \otimes R_k \otimes S_d](vv^*)\|,
\end{equation}
where $Q,R,S \in \{\mathrm{id},\mathrm{Tr}, \Theta\}$. As an illustration, in the case of the Choi matrix, where we apply the trace map on the first factor and the identity on the other two, we have $B_{\mathrm{Tr}, \mathrm{id}, \mathrm{id}}(L) = \|C_L\|$.

We gather in Table \ref{tab:bounds} the $27=3^3$ possibilities we obtain by applying on each of the tensor factors the maps above. We obtain 5 different bounds ($B_{I}$, $B_{C}$, $B_{C\Gamma}$, $B_{Cc\Gamma}$, and $B_{M\Gamma}$) which appear several times in the list, as well as some trivial, constant bounds which do not depend on the channel.

\begin{table}[htdp]
\caption{Bounds obtained by applying traces and transpositions to the orthogonal projection on $v$. The defining occurrences of the bounds appear in red, while the other occurrences appear in blue; trivial, constant bounds are black.}
\begin{center}
\begin{tabular}{|r|c|c|c|l|}
\hline
N$^\circ$ & $\mathbb C^n$ & $\mathbb C^k$ & $\mathbb C^d$ & Norm  \\ \hline \hline 
1. & $\mathrm{id}$ & $\mathrm{id}$ & $\mathrm{id}$ & $=d$ \\ \hline
2. & $\mathrm{id}$ & $\mathrm{id}$ & $\mathrm{Tr}$ & $=1$ \\ \hline
3. & $\mathrm{id}$ & $\mathrm{id}$ & $\Theta$ & $=1$ \\ \hline
4. & $\mathrm{id}$ & $\mathrm{Tr}$ & $\mathrm{id}$ & \textcolor{blue}{$=B_{I}$} \\ \hline
5. & $\mathrm{id}$ & $\mathrm{Tr}$ & $\mathrm{Tr}$ & \textcolor{blue}{$=B_C$} \\ \hline
6. & $\mathrm{id}$ & $\mathrm{Tr}$ & $\Theta$ & \textcolor{red}{$=:B_{Cc\Gamma}$} \\ \hline
7. & $\mathrm{id}$ & $\Theta$ & $\mathrm{id}$ & \textcolor{blue}{$=B_{I}$} \\ \hline
8. & $\mathrm{id}$ & $\Theta$ & $\mathrm{Tr}$ & \textcolor{red}{$=:B_{M\Gamma}$} \\ \hline
9. & $\mathrm{id}$ & $\Theta$ & $\Theta$ & \textcolor{blue}{$=B_C$} \\ \hline
10. & $\mathrm{Tr}$ & $\mathrm{id}$ & $\mathrm{id}$ & \textcolor{red}{$=:B_C$} \\ \hline
11. & $\mathrm{Tr}$ & $\mathrm{id}$ & $\mathrm{Tr}$ & \textcolor{red}{$=:B_{I}$} \\ \hline
12. & $\mathrm{Tr}$ & $\mathrm{id}$ & $\Theta$ & \textcolor{red}{$=:B_{C\Gamma}$} \\ \hline
13. & $\mathrm{Tr}$ & $\mathrm{Tr}$ & $\mathrm{id}$ & $=1$ \\ \hline
14. & $\mathrm{Tr}$ & $\mathrm{Tr}$ & $\mathrm{Tr}$ & $=d$ \\ \hline
\end{tabular}
\qquad
\begin{tabular}{|r|c|c|c|l|}
\hline
N$^\circ$ & $\mathbb C^n$ & $\mathbb C^k$ & $\mathbb C^d$ & Norm  \\ \hline \hline 
15. & $\mathrm{Tr}$ & $\mathrm{Tr}$ & $\Theta$ & $=1$ \\ \hline
16. & $\mathrm{Tr}$ & $\Theta$ & $\mathrm{id}$ & \textcolor{blue}{$=B_{C\Gamma}$} \\ \hline
17. & $\mathrm{Tr}$ & $\Theta$ & $\mathrm{Tr}$ & \textcolor{blue}{$=B_{I}$} \\ \hline
18. & $\mathrm{Tr}$ & $\Theta$ & $\Theta$ & \textcolor{blue}{$=B_C$} \\ \hline
19. & $\Theta$ & $\mathrm{id}$ & $\mathrm{id}$ & \textcolor{blue}{$=B_C$} \\ \hline
20. & $\Theta$ & $\mathrm{id}$ & $\mathrm{Tr}$ & \textcolor{blue}{$=B_{M\Gamma}$} \\ \hline
21. & $\Theta$ & $\mathrm{id}$ & $\Theta$ & \textcolor{blue}{$=B_{I}$} \\ \hline
22. & $\Theta$ & $\mathrm{Tr}$ & $\mathrm{id}$ & \textcolor{blue}{$=B_{Cc\Gamma}$} \\ \hline
23. & $\Theta$ & $\mathrm{Tr}$ & $\mathrm{Tr}$ & \textcolor{blue}{$=B_C$} \\ \hline
24. & $\Theta$ & $\mathrm{Tr}$ & $\Theta$ & \textcolor{blue}{$=B_{I}$} \\ \hline
25. & $\Theta$ & $\Theta$ & $\mathrm{id}$ & $=1$ \\ \hline
26. & $\Theta$ & $\Theta$ & $\mathrm{Tr}$ & $=1$ \\ \hline
27. & $\Theta$ & $\Theta$ & $\Theta$ & $=d$ \\ \hline
\end{tabular}
\end{center}
\label{tab:bounds}
\end{table}

Let us first comment on the equalities in Table \ref{tab:bounds}. These follow from the following basic facts about the operator norm. First, we note that the operator norm is invariant under \emph{global} transposition. Moreover, consider a rank one projection $xx^*$ acting on a bipartite Hilbert space $\mathbb C^p \otimes \mathbb C^q$ (here, $p$ is the product of some, possible empty, subset of $\{n,k,d\}$ and $q=nkd/p$). Then, one has
\begin{equation}\label{eq:equality-id-Tr-Theta-projection}
\|[\mathrm{id}_p \otimes \mathrm{Tr}_q](xx^*)\| = \|[\mathrm{Tr}_p \otimes \mathrm{id}_q](xx^*)\| = \|[\mathrm{id}_p \otimes \Theta_q](xx^*)\| = \|[\Theta_p \otimes \mathrm{id}_q](xx^*)\| = \lambda_1,
\end{equation}
where $\lambda_1$ is the largest Schmidt coefficient of the vector $x \in \mathbb C^p \otimes \mathbb C^q$. The statement above is well known in the case of the partial traces. In the case of the partial transpositions, it is also straightforward, see \cite[Lemma III.3]{hil}. Let us use these simple facts to prove the equality cases in Table \ref{tab:bounds}. 

First, the constant value $d$ appears 3 times in the table, and the equality of the three quantities is precisely equation \eqref{eq:equality-id-Tr-Theta-projection} with the choice $p=1$, $q=nkd$. The same relation \eqref{eq:equality-id-Tr-Theta-projection} implies the equality of the quantities in the rows 2, 3, 13, and 25 with the choice $p=nk$, $q=d$; the common value is the operator norm of the orthogonal projection on the image of $V$ (see row 2), which is 1. Moreover, applying a global transposition to the quantities on rows 15 and 26, we obtain, respectively, the quantities on rows 13 and 2, which are 1; we have shown thus the equality of all the quantities which give 1 in Table \ref{tab:bounds}.

The choice $p=k$, $q=nd$ in \eqref{eq:equality-id-Tr-Theta-projection} yields the equality of the rows 4,7,11,21; the common value is the operator norm output of the identity $\|L(I)\|$, as seen from row 11. Applying a global transposition on the rows 17 and 24, we obtain the rows 11 and 4, giving the same value $B_I:=\|L(I)\|$.

Similarly, the choice $p=n$, $q=kd$ in \eqref{eq:equality-id-Tr-Theta-projection} yields the equality of the rows 5, 9, 10, 19; the common value is the operator norm of the Choi matrix $\|C_L\|$, as seen from row 10. Applying a global transposition on the rows 23 and 18, we obtain the rows 5 and 10, giving the same value $B_C:=\|C_L\|$.

The 6 remaining cases are the ones corresponding to the 6 permutations of the operators $\mathrm{id}$, $\mathrm{Tr}$, and $\Theta$ acting on the 3 legs of the tensor $v$. The fact that each of the quantities $B_{C\Gamma}$, $B_{Cc\Gamma}$, and $B_{M\Gamma}$ appears twice is a consequence of the invariance of the norm by global transposition. The quantity $B_{C\Gamma}$ corresponds to the partial transposition of the Choi matrix of $L$: $B_{C\Gamma} := \|C_L^\Gamma\|$. The quantity  $B_{Cc\Gamma}$ corresponds to the partial transposition of the Choi matrix of the complementary channel  $L^c$: $B_{Cc\Gamma} := \|C_{L^c}^\Gamma\|$. Finally,  $B_{M\Gamma}$ corresponds to Montanaro's bound \cite[Fact 1 and Proposition 4]{mon}: it is the norm of the partial transposition of the orthogonal projection on the image of $V$,  $B_{M\Gamma} := \|M_L^\Gamma\|$, where $M$ is the projection on the image of $V$, $M_L = VV^*$. 

Let us note the important fact that the last two bounds we considered,  $\|C_{L^c}^\Gamma\|$ and $\|M_L^\Gamma\|$, are not defined in terms of the channel $L$, but in terms of the isometry $V$, or its vectorized version $v$ (whereas the first two are defined in terms of the Choi matrix of $L$). In the next lemma, we show that these two quantities do not depend on the actual choice of the isometry $V$ defining the channel, but only on the channel itself, so the notations $B_{Cc\Gamma}(L)$ and $B_{M\Gamma}(L)$  are justified.

\begin{lemma}
Consider a fixed quantum channel $L:\mathcal M_d(\mathbb C) \to \mathcal M_k(\mathbb C)$ and let $V:\mathbb C^d \to \mathbb C^n \otimes \mathbb C^k$ be a Stinespring isometry for $L$, i.e.~ $L(X) = [\mathrm{Tr}_n \otimes \mathrm{id}_k](V X V^*)$. Then, the bounds $B_{Cc\Gamma} = \|C_{L^c}^\Gamma\|$ and $B_{M\Gamma} = \|M_L^\Gamma\|$, defined using $V$, do not depend on the choice of the isometry $V$.
\end{lemma}
\begin{proof}
Let $r$ be the rank of the Choi matrix $C_L$ of $L$ and consider a \emph{minimal} purification $v_0 \in \mathbb C^r \otimes (\mathbb C^k \otimes \mathbb C^d)$ of $C_L$. The vectorized version of the Stinespring isometry $V$ is another purification of $C_L$ (not necessarily minimal). Since any purification of a quantum mixed state is related to a minimal one by an isometry, there exist an isometric operator $W: \mathbb C^r \to \mathbb C^n$ such that $v = (W \otimes I_{kd})v_0$. We now have that 
\be vv^*  = (W \otimes I_{kd}) v_0v_0^*(W \otimes I_{kd})^*,\ee
so, after taking a partial trace, a partial transposition, the operator norm, and using $W^*W = I_r$, we can conclude.
\end{proof}
\medskip

To summarize, we have associated to a quantum channel $L$ the following quantities:
\begin{align}
\label{eq:B-C}B_C(L) &= \| C_L \| \\
\label{eq:B-C-Gamma}B_{C\Gamma}(L) &= \| C_L^\Gamma \| \\
\label{eq:B-C-c-Gamma}B_{Cc\Gamma}(L) &= \| C_{L^c}^\Gamma \| \\
\label{eq:B-M}B_{M\Gamma}(L) &= \| M_L^\Gamma \| \\
\label{eq:B-I}B_I(L) &= \| L(I) \|,
\end{align}
where $C_L$, resp.~ $C_{L^c}$ are the Choi matrices of the channel $L$, resp.~ of the complementary channel $V^c$ and $M_L=VV^*$ is the projection on the image of the Stinespring isometry $V$ defining the channel $L$.

\subsection{Additivity and R\'enyi entropy bounds}

We start our discussion with the proof of the fact the quantities $B_\cdot(L)$ introduced previously \eqref{eq:B-C}-\eqref{eq:B-I} are \emph{multiplicative} with respect to the tensor product operation. Later, in order to be consistent with the entropic quantities, we shall take logarithms of these quantities, making them \emph{additive}. 

\begin{lemma}\label{lem:mult}
Let $Q,R,S$ be arbitrary operations chosen from the set $\{\mathrm{id}, \mathrm{Tr}, \Theta\}$. Then, for any quantum channels $L_1, L_2$, the following multiplicativity relation holds:
\be B_{Q,R,S}(L_1 \otimes L_2) = B_{Q,R,S}(L_1) \cdot B_{Q,R,S}(L_2).\ee
\end{lemma}
\begin{proof}
Let $V_{1,2}$ the isometries defining the quantum channels $L_{1,2}$. It is immediate that $V_1 \otimes V_2$ is a  Stinespring isometry for $L_1 \otimes L_2$; the same holds for the vectorized versions $v_{1,2}$. Hence, 
\begin{align}
B_{Q,R,S}(L_1 \otimes L_2) &= \| [Q \otimes R \otimes S]\left( (v_1 \otimes v_2) (v_1 \otimes v_2)^* \right) \| \notag\\
&= \| [Q \otimes R \otimes S]\left( v_1v_1^* \otimes v_2v_2^*\right) \| \notag\\
&= \| [Q \otimes R \otimes S]( v_1v_1^*) \otimes [Q \otimes R \otimes S](v_2v_2^*) \| \notag\\
&= \| [Q \otimes R \otimes S]( v_1v_1^*)\| \cdot \| [Q \otimes R \otimes S](v_2v_2^*) \| \notag\\
&= B_{Q,R,S}(L_1) \cdot B_{Q,R,S}(L_2),
\end{align}
where we have used the multiplicativity of the maps $T \in \{\mathrm{id}, \mathrm{Tr}, \Theta\}$: $T(X \otimes Y) = T(X) \otimes T(Y)$.
\end{proof}

We continue with a useful linear algebra lemma, relating the (partially transposed) Choi matrix of a channel with that of the dual channel. Recall that the dual map $L^*$ of a quantum channel $L:\mathcal M_d(\mathbb C) \to \mathcal M_k(\mathbb C)$ is the unital, completely positive map $L^*:\mathcal M_k(\mathbb C) \to \mathcal M_d(\mathbb C)$ which satisfies the following duality relation with respect to the Hilbert-Schmidt scalar product
\be \forall \, X \in \mathcal M_k(\mathbb C), Y \in \mathcal M_d(\mathbb C), \quad \langle X, L(Y) \rangle = \langle L^*(X), Y \rangle.\ee

\begin{lemma}\label{lem:C-L-star}
Let $L:\mathcal M_d(\mathbb C) \to \mathcal M_k(\mathbb C)$ be a linear map and $L^*:\mathcal M_k(\mathbb C) \to \mathcal M_d(\mathbb C)$ be its dual with respect to the usual, Hilbert-Schmidt, scalar product. Then, 
\begin{align}
C_{L^*} &= F_{k,d} C_L^\top F_{k,d}^*\\
C_{L^*}^\Gamma &= F_{k,d} (C_L^\Gamma)^\top F_{k,d}^*,
\end{align}
where $F_{k,d}:\mathbb C^k \otimes \mathbb C^d \to \mathbb C^d \otimes \mathbb C^k$ is the flip operator, i.e. $F_{k,d}(a \otimes b) = b \otimes a$. In particular, the matrices $C_L$ and $C_{L^*}$ (resp. $C_L^\Gamma$ and $C_{L^*}^\Gamma$) have the same spectrum.
\end{lemma}
\begin{proof}
We are going to show the first equality, the proof for the partial transpositions being similar. A simple proof of the claim can be obtained using the graphical notation for tensors, see Figure \ref{fig:C-L-star}: the diagram on the right can be obtained by flipping horizontally ($\top$) and vertically ($F_{k,d}$) the input and output labels of the diagram on the left. 
\begin{figure}[htbp] 
\includegraphics{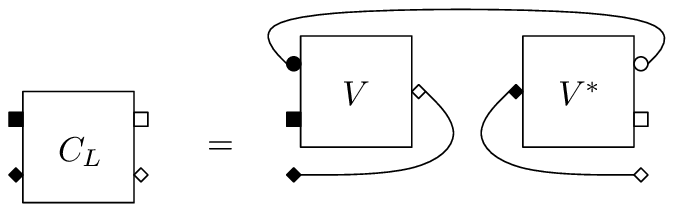} \quad\quad\quad
\includegraphics{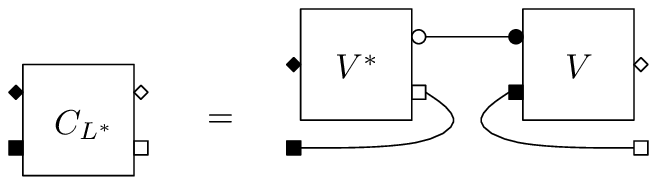}
\caption{Diagrams for the Choi matrix of a channel $L$ (left) and for the Choi matrix of the adjoint channel $L^*$ (right). The channel $L$ is given by a Stinespring isometry $V:\mathbb C^d \to \mathbb C^n \otimes \mathbb C^k$. The round, square, and resp.~diamond shapes denote the vector spaces $\mathbb C^n$, $\mathbb C^k$ and resp.~$\mathbb C^d$.} 
\label{fig:C-L-star}
\end{figure}

Let us now give a detailed, algebraic proof. Consider orthonormal bases $\{e_i\}_{i=1}^d$, $\{f_x\}_{x=1}^k$ of $\mathbb C^d$, resp.~$\mathbb C^k$, and compute
\begin{align}
\nonumber e_i^* \otimes f_x^*C_{L^*}  e_j \otimes f_y &= e_i^* L^*(f_xf_y^*) e_j \\
\nonumber &= \operatorname{Tr}[e_je_i^*  L^*(f_xf_y^*)]\\
\nonumber &= \operatorname{Tr}[L(e_je_i^*) f_xf_y^*]\\
\nonumber &= f_y^* L(e_je_i^*) f_x \\
\nonumber &= f_y^* \otimes e_j^* C_L f_x \otimes e_i\\
\nonumber &= f_x^* \otimes e_i^* C_L^\top f_y \otimes e_j\\
&= e_i^*  \otimes f_x^* F_{k,d} C_L^\top F_{k,d}^* e_j \otimes f_y,
\end{align}
proving the claim.
\end{proof}

We can now prove the main result of this section, the fact that the quantities in \eqref{eq:B-C}-\eqref{eq:B-I} are lower bounds for the minimum output $p$-R\'enyi entropy of the quantum channel $L$.

\begin{proposition}\label{prop:bound}
Let $L:\mathcal M_d(\mathbb C) \to \mathcal M_k(\mathbb C)$ be a quantum channel. Then, for all integers $r \geq 1$, 
\begin{align}
\label{eq:bound-tensor-B-C} H_2^{\min}(L^{\otimes r}) &\geq - r \log \|C_L\| \\
\label{eq:bound-tensor-B-C-Gamma} H_2^{\min}(L^{\otimes r}) &\geq - r \log \|C^\Gamma_L\| \\
\label{eq:bound-tensor-B-C-c-Gamma} H_2^{\min}(L^{\otimes r}) &\geq - r \log \|C^\Gamma_{L^c}\| \\
\label{eq:bound-tensor-B-M} H_\infty^{\min}(L^{\otimes r}) &\geq - r \log \|M_L^\Gamma\| \\
\label{eq:bound-tensor-B-I} H_\infty^{\min}(L^{\otimes r}) &\geq - r \log \|L(I)\| .
\end{align}
\end{proposition}
\begin{proof}
Note that, using the multiplicativity result of Lemma \ref{lem:mult}, we only need to show that the inequalities hold for $r=1$.
Let us start with the first two inequalities, involving the Choi matrix $C_L$. 
The proof begins with a straightforward ``linearization trick'' and proceeds by linear algebra manipulations. For a fixed quantum state $X \in \mathcal M_d^{1,+}(\mathbb C)$, we have
\begin{align}
\nonumber \operatorname{Tr}[L(X)^2] &= \operatorname{Tr}[L(X) \otimes L(X)  \, \cdot \, F_{k,k}]\\
\nonumber &= \operatorname{Tr}\left[ (X\otimes L(X)) \, \cdot \, [L^* \otimes \operatorname{id}](F_{k,k})\right]\\
\label{eq:proof-C-L-Gamma-star}&= \operatorname{Tr}\left[ (X \otimes L(X)) \, \cdot \, C_{L^*}^\Gamma\right]\\
\nonumber &\leq \lambda_{\max}(C_{L^*}^\Gamma) = \lambda_{\max}(C_{L}^\Gamma) \leq \|C_L^\Gamma\|,
\end{align}
where we have used the isospectral property proved in Lemma \ref{lem:C-L-star} and the fact that $X \otimes L(X)$ is positive semidefinite and has unit trace. Taking the supremum over all input states $X$ yields \eqref{eq:bound-tensor-B-C-Gamma}. To show \eqref{eq:bound-tensor-B-C}, apply in \eqref{eq:proof-C-L-Gamma-star} the transposition operation $\Theta$ on the second factor of the tensor product; this will remove the partial transposition on the matrix $C_{L^*}^\Gamma$, and the claim will follow. 

The inequality \eqref{eq:bound-tensor-B-C-c-Gamma} follows from \eqref{eq:bound-tensor-B-C-Gamma} and the fact that $H_2^{\min}(L) = H_2^{\min}(L^c)$. Finally, \eqref{eq:bound-tensor-B-M} has been proved in \cite[Fact 1 and Proposition 4]{mon}, while \eqref{eq:bound-tensor-B-I} follows trivially from the matrix inequality $L(X) \leq L(I)$ (see also \cite{fri}). 
\end{proof}

\begin{remark}
The bounds above are tight: the maximally depolarizing channel $\Delta(X) = \mathrm{Tr}(X) I/d$ saturates bounds \eqref{eq:bound-tensor-B-C} and \eqref{eq:bound-tensor-B-C-Gamma}, while the identity channel $\mathrm{id}(X) = X$ saturates bounds \eqref{eq:bound-tensor-B-C-c-Gamma}, \eqref{eq:bound-tensor-B-M}, and \eqref{eq:bound-tensor-B-I}.
\end{remark}

\subsection{Bounds for additivity rate}
The additive bounds derived above can be used to bound the $p$-additivity rates of channels, as we show in the following result. Note that our starting point is the value at $p=2$ (see Proposition \ref{prop:bound}) so our method will give interesting results for values of $p$ in the interval $[0,2]$. The extension of the inequalities for $p >2$ is less interesting, since the derived inequalities contain additional factors which depend on $p$. Indeed, for any quantum state $\rho$ and $p >2$, 
\begin{equation}\label{eq:ineq-H-p-2}
\frac{p}{2(p-1)} H_2(\rho) \leq H_p(\rho) \leq H_2(\rho),
\end{equation}
which is the best general bound for the $p$-R\'enyi entropy in terms of the $2$-R\'enyi entropy.
\begin{proposition}\label{prop:bound-alpha}
For a quantum channel $L$ having no pure outputs, the $p$-additvity rate \eqref{eq:def-wad} is lower bounded, for all $p \in [0,2]$, by the following quantity
\begin{equation}\label{eq:bound-alpha}
\alpha_p(L) \geq \hat \alpha_p(L) := \frac{- \log B}{H_p^{\min}(L)},
\end{equation}
where $B = \min\{\|C_L\|,\|C_L^\Gamma\|,\|C_{L^c}^\Gamma\|,\|M_L^\Gamma\|, \|L(I)\|\}$.
\end{proposition}
\begin{remark}
When $p>2$, the quantity $B$ above has to be replaced by 
\begin{equation}
\min\{\|C_L\|^{c_p}, \|C_L^\Gamma\|^{c_p}, \|C_{L^c}^\Gamma\|^{c_p},\|M_L^\Gamma\|, \|L(I)\| \},
\end{equation}
where $c_p = p/(2p-2)$ is a correction exponent which appears because one has to use in this case \eqref{eq:ineq-H-p-2}.
 \end{remark}

We investigate next how the lower bound for additivity rates behaves with respect to tensor products. 

\begin{proposition}\label{prop:alpha-Gamma-v-p}
Let $L$ and $K$ be two quantum channels with the property that their tensor product $L \otimes K$ has no pure output. Then, for all $p \in [0,2]$, the following inequality holds:
\begin{equation}\label{eq:bound-violation}
\hat \alpha_p(L \otimes K) \leq v_p(L,K) \left[ t \hat \alpha_p(L) + (1-t) \hat \alpha_p(K) \right],
\end{equation}
where $v_p(L,K)$ is the relative violation of the minimum $p$-output entropy \eqref{eq:def-relative-violation} and $t  = H_p^{\min}(L) / [H_p^{\min}(L) + H_p^{\min}(K)] \in [0,1]$ (note that $H_p^{\min}(L)$ and $H_p^{\min}(K)$ cannot be both null).
\end{proposition}
\begin{proof}
Let $B_L$, $B_K$, and $B_{L \otimes K}$ be the bounds associated to the channels $L$, $K$, and $L \otimes K$, as in  Proposition \ref{prop:bound-alpha}; they obviously satisfy $B_L B_K \leq B_{L \otimes K}$. Starting from the right-hand-side of the inequality to be proven, we have then
\begin{equation}
v_p(L,K) \left[ t \hat \alpha_p(L) + (1-t) \hat \alpha_p(K) \right] = \frac{- \log(B_L B_K)}{H_p^{\min}(L \otimes K)} \geq \frac{- \log B_{L \otimes K}}{H_p^{\min}(L \otimes K)} = \hat \alpha_p(L \otimes K).
\end{equation}
\end{proof}

\subsection{Examples: the Werner-Holevo and the antisymmetric channels}
In this section we compute the bounds derived earlier for the Werner-Holevo channel discussed in Example \ref{ex:WH} and for the antisymmetric channel from Examples \ref{ex:GHP}.

Recall that the Werner-Holevo channel $W_d:\mathcal M_d(\mathbb C) \to \mathcal M_d(\mathbb C)$ is defined by
\begin{equation}
W_d(X) = \frac{1}{d-1} [ \mathrm{Tr}(X) I_d - X^\top].
\end{equation}

First, note that $W_d(I) = I$, hence $\|W_d(I)\| = 1$. The Choi matrix and its partial transpose read
\begin{align}
C_{W_d} &=  \frac{1}{d-1} [ I_{d^2} - F_d] \\
C_{W_d}^\Gamma &=  \frac{1}{d-1} [ I_{d^2} - E_d],
\end{align}
where $E_d$ is the (un-normalized) maximally entangled state \eqref{eq:maximally-entangled} and $F_d$ denotes the flip operator $F_d(x \otimes y) = y \otimes x$, for $x, y \in \mathbb C^d$. The corresponding bounds are easily computed from the above relations: $\|C_{W_d} \| = 2/(d-1)$ and $\|C_{W_d}^\Gamma \| = 1$; in particular, note that the matrix $C_{W_d}^\Gamma$ is not positive semidefinite, i.e.~ the Werner-Holevo channels is not PPT. A minimal purification of the Choi matrix $C_{W_d}$ is given by the four-partite vector $v_{W_d} \in (\mathbb C^d)^{\otimes 4}$ (see Figure \ref{fig:minimal-purification-WH} for a graphical representation)
\begin{equation}\label{eq:def-v-W-d}
v_{W_d} = \frac{1}{\sqrt{2(d-1)}}v =  \frac{1}{\sqrt{2(d-1)}}\sum_{i,j=1}^d e_i \otimes e_j \otimes e_j \otimes e_i - e_i \otimes e_j \otimes e_i \otimes e_j.
\end{equation}

\begin{figure}[htbp] 
\includegraphics{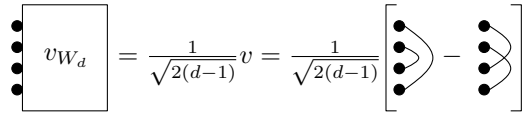} 
\caption{A minimal purification of the Choi matrix of a Werner-Holevo channel.} 
\label{fig:minimal-purification-WH}
\end{figure}

Using the above purification, the remaining two bounds are seen to be equal: $\|C_{W_d^c}^\Gamma\| = \|M_{W_d}^\Gamma\| = 2/(d-1)$. Note that for the vector $v_{W_d}$, the top two factors of the four-partite tensor product correspond to the ``environment dimension'', the third one corresponds to the output, while the fourth one corresponds to the input. In conclusion, the additivity rates of the Werner-Holevo channel are bounded, for all $p \in [0,2]$, as follows:
\begin{equation}
\alpha_p(W_d) \geq \frac{-\log(2/(d-1))}{\log (d-1)} = 1-\frac{\log 2}{\log(d-1)}.
\end{equation}

Let us now move to the case of the antisymmetric channel $A_d$ introduced by Grudka, Horodecki, and Pankowski in \cite{ghp}. If $P:\mathbb C^d \otimes \mathbb C^d \to \mathbb C^{{d \choose 2}}$ is the projection operator on the antisymmetric subspace, i.e.~ $P=Q^*$, where $Q: \Lambda^2(\mathbb C^d) \cong \mathbb C^{{d \choose 2}} \to \mathbb C^d \otimes \mathbb C^d$ is the canonical embedding of the antisymmetric space $\Lambda^2(\mathbb C^d)$ into $\mathbb C^d \otimes \mathbb C^d$, then 
\begin{equation}
v_{A_d} = \frac{1}{\sqrt 2} (I_d \otimes I_d \otimes P) v =  \frac{1}{\sqrt 2} (I_d \otimes I_d \otimes P)\sum_{i,j=1}^d e_i \otimes e_j \otimes e_j \otimes e_i - e_i \otimes e_j \otimes e_i \otimes e_j.
\end{equation}
Noticing the similarities between the formula above and \eqref{eq:def-v-W-d}, one can easily show that $\|C_{A_d}^\Gamma\| = \|C_{A_d^c}^\Gamma\| = 1$, while $\|A_d(I)\| = \|C_{A_d}\| = \|M_{A_d}^\Gamma\| = (d-1)/2$. We conclude that the lower bounds discussed in this paper are trivial for the antisymmetric channel and we leave the question of determining the additivity rates of this channel open.

\section{Partially transposed random Choi matrices and their norm}
\label{sec:random-Choi}

In this section and in the next one, we analyze the behavior of the bounds \eqref{eq:B-C}-\eqref{eq:B-I} for random quantum channels $L$. 
Here, we compute the bound \eqref{eq:B-C-Gamma}, while in the next section we compute the four others (when possible) and we compare them, in the spirit of Proposition \ref{prop:bound-alpha}.

This section contains one of the main results of our work, Theorem \ref{thm:strong-convergence}, which contains a formula for the asymptotic operator norm of the partially transposed Choi matrices of random quantum channels. More precisely, we consider a sequence of random quantum channels $L_n : \mathcal M_{d_n}(\mathbb C) \to \mathcal M_k(\mathbb C)$ where $d_n \sim tkn$, defined by
\begin{equation}\label{eq:random-quantum-channel}
L_n(X)=[\mathrm{Tr}_{\mathbb C^n} \otimes \mathrm{id}_k](V_n XV_n^*),
\end{equation}
where $V_n : \mathbb C^{d_n} \to \mathbb C^n \otimes \mathbb C^k$ is the random isometry.
We write $C_n :=C_{L_n}$ for the Choi matrix of the quantum channel $L_n$. The proof is split into three parts, delimited by subsections, and uses the moment method. We start by computing the asymptotic moments of the random matrices $C_n^\Gamma$, using the graphical Weingarten calculus introduced in Section \ref{sec:graphical-Weingarten}. Then, in the second subsection, we identify a probability measure having the exact moments computed in the first part; the measure is given in terms of the free additive convolution operation from free probability theory. Finally, we show that the random matrices $C_n^\Gamma$ converge strongly, which gives us the desired norm convergence. 

Let us first start with some general considerations about the random matrix we are studying. 
Associated to the channel $L$ is the \emph{partially transposed Choi matrix} of $L$, 
\be C_L^\Gamma = [\mathrm{id}_k \otimes \Theta_d] [L \otimes \mathrm{id_d}](E_d) = [L \otimes \mathrm{id_d}](F_d),\ee
which is the random matrix we are interested in (see Figure \ref{fig:C-Gamma-n} for a graphical representation). Before investigating the eigenvalue distribution of the random matrix $C_L^\Gamma$, let us first comment on its distribution as a matrix. Since, in this paper, Choi matrices are not normalized ($\mathrm{Tr} C_L^\Gamma = \mathrm{Tr} C_L = d$), the matrix $d^{-1}C_L$ is a (random) density matrix (mixed quantum state). In the literature, several ensembles of random density matrices have been considered: the induced measures \cite{zso}, the Bures measure \cite{hal}, or measures associated to graphs \cite{cnz}, just to name a few (see also \cite{zpnc} for a random matrix theory perspective). We would like to argue at this point that the distribution of the Choi matrix we consider is not related to the induced measures introduced in \cite{zso}. Indeed, it is easy to see that the distribution of $C_L$ involves more than one column of the unitary matrix $U$, while, in order to define the induced measures, one needs just one column of a unitary operator (or a random point on the unit sphere, or a normalized Gaussian vector). A rigorous argument for this fact is given in Section \ref{sec:PPT}, Remark \ref{rk:random-density-matrices-PPT}, using PPT thresholds.

\subsection{Exact moments}\label{sec:moment-cal}

In the next proposition, we compute the moments of the partially transposed Choi matrix $C_n^\Gamma$, for a fixed value of $n$. We make use of the graphical calculus formalism introduced in Section \ref{sec:graphical-Weingarten}.

\begin{proposition}\label{prop:moments-CnGamma}
For any integer dimensions $n,k,d$, the moments of the random matrix $C_n^\Gamma \in \mathcal M_{kd}^{sa}(\mathbb C)$ are given by
\begin{equation}\label{eq:moments-CnGamma}
\forall p \geq 1, \quad\mathbb E \frac{1}{kd}\operatorname{Tr}(C_n^\Gamma)^p = (kd)^{-1}\sum_{\alpha,\beta \in \mathcal S_p} n^{\#\alpha}k^{\#(\gamma^{-1}\alpha)}d^{\#(\gamma\beta)}\operatorname{Wg}_{nk}(\alpha^{-1}\beta).
\end{equation}
In the above equation, $\#(\cdot)$ denotes the number of cycles of a permutation, $\gamma = (p   \cdots 3 \, 2 \, 1) \in \mathcal S_p$ denotes the full cycle permutation and $\operatorname{Wg}$ is the Weingarten function \cite{col}, see Definition \ref{def:Wg}.
\end{proposition}
\begin{proof}
The proof is an application of the graphical Weingarten formula from Theorem \ref{thm:graphical-Wg}. We have depicted in Figure \ref{fig:C-Gamma-n} the diagram for the random matrix $C_n^\Gamma$, which we are investigating. Note that we have replaced the random isometry $V_n$ by a random, Haar-distributed, unitary matrix $U_n \in \mathcal U_{nk}$. Moreover, we choose, for the sake of simplicity, not to represent labels corresponding to the added dimensions $nk/d$, which do not play any role in the moment computations which follow, since the contractions of the corresponding wires multiply the result by the scalar $1$.
\begin{figure}[htbp] 
\includegraphics{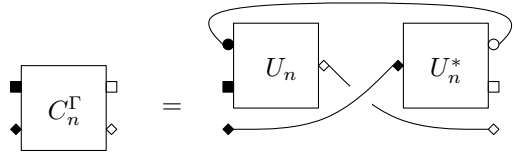} 
\caption{Diagram for the random matrix $C_n^\Gamma$. The box $U_n$ corresponds to a random Haar unitary matrix. We do not depict the labels corresponding to the environment of size $kn/d$, which does not play any role in the computations.} 
\label{fig:C-Gamma-n}
\end{figure}
We are interested in computing, for all integer $p \geq 1$, the moment $\mathbb E \operatorname{Tr}(C_n^\Gamma)^p$. The diagram corresponding to this real number is depicted in Figure \ref{fig:Tr-C-n}.  We now use formula \eqref{eq:graphical-Wg} to compute the expectation, with respect to the random unitary matrix $U_n$:
\be \mathbb E \operatorname{Tr}(C_n^\Gamma)^p  = \sum_{\alpha, \beta \in \mathcal S_p} \mathcal D_{\alpha,\beta} \operatorname{Wg}_{kn}(\alpha^{-1}\beta),\ee
where $\mathcal D_{\alpha,\beta}$ is the diagram obtained by erasing the $U_n$ and $\bar U_n$ boxes and connecting the black (resp. white) decorations of the $i$-th $U_n$ box with the corresponding black (resp. white) decorations of the $\alpha(i)$-th (resp. $\beta(i)$-th) $\bar U_n$ box. The resulting diagram $\mathcal D_{\alpha,\beta}$ is a collection of loops corresponding to different vector spaces, as follows (see Figure \ref{fig:D-alpha-beta}):
\begin{enumerate}
\item $\# \alpha$ loops of dimension $n$, corresponding to round-shaped labels. The round decorations are initially connected with the identity permutation and the graphical expansion connects them with the permutation $\alpha$. The resulting number of loops is $\#\alpha = \#(\mathrm{id}^{-1}\alpha)$;
\item $\#(\gamma^{-1} \alpha)$ loops of dimension $k$, corresponding to square-shaped labels. The square decorations are initially connected with the permutation $\gamma$ (that is, $i \mapsto i-1$) and the graphical expansion connects them with the permutation $\alpha$. The resulting number of loops is $\#(\gamma^{-1} \alpha)$;
\item $\#(\gamma\beta)$ loops of dimension $d$, corresponding to diamond-shaped labels. The diamond decorations are initially connected with the permutation $\gamma^{-1}$ (i.e. $i \mapsto i+1$) and the graphical expansion connects them with the permutation $\beta$. The resulting number of loops is $\#((\gamma^{-1})^{-1} \beta) = \#(\gamma\beta)$.
\end{enumerate}
\begin{figure}[htbp] 
\includegraphics{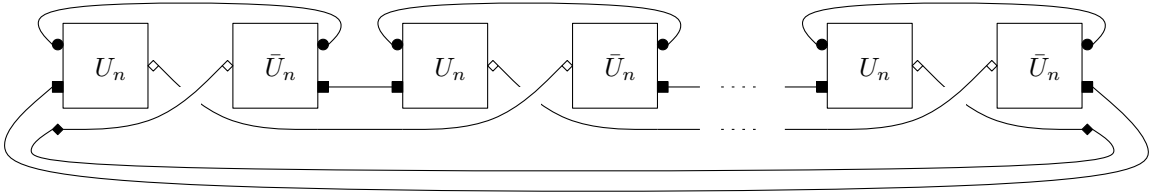} 
\caption{Diagram for the trace of the $p$-th power of the partially transposed Choi matrix $C_n^\Gamma$. The diagram contains $p$ copies of the matrix from Figure \ref{fig:C-Gamma-n}.} 
\label{fig:Tr-C-n}
\end{figure}
\begin{figure}[htbp] 
\includegraphics{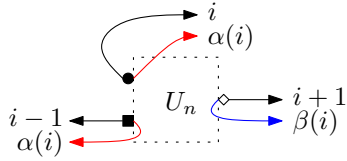} 
\caption{Diagram representing the $i$-th $U_n$ box in the graphical expansion. The box itself no longer exists and the diagram $\mathcal D_{\alpha, \beta}$ consists of a collection of loops.} 
\label{fig:D-alpha-beta}
\end{figure}
\end{proof}

\subsection{Limiting spectral distribution}
\label{sec:limiting-spectrum-C-Gamma}
We consider now the asymptotic behavior of the partially transposed Choi matrix $C_n^\Gamma$, in the following regime:

\medskip

\noindent\textbf{Asymptotic regime}
\begin{itemize}
\item $k$ is fixed
\item $n \to \infty$
\item $d \to \infty$ with $d \sim tkn$ for some fixed ratio $t \in (0,1)$.\end{itemize}

\medskip

In this regime (which we call the ``fixed $k$ regime''), the output dimension $k$ and the input/(total space) ratio $t$ are treated as parameters that we shall consider fixed in what follows. We compute next the asymptotic moments of the random matrix $C_n^\Gamma$ and we identify a probability measure having these exact moments. 
Recall that the dilation operator $D_t$ is defined by $D_t[\mu] = (x \mapsto tx)_\# \mu$.

\begin{theorem}\label{thm:limit-moment}
In the fixed $k$ asymptotic regime, the moments of the random matrix $C_n^\Gamma$ converge towards the moments of the following probability measure $\mu^{(C\Gamma)}_{k,t}$:
\begin{equation}\label{eq:def-mukt-1}
\mu^{(C\Gamma)}_{k,t} = D_t \left[ \left((1-s)\,\delta_{-1} +s \, \delta_{+1} \right)^{\boxplus 1/t}\right],
\end{equation}
where 
$D_\cdot$ is the dilation operator, $\boxplus$ is the additive free convolution and $s:=(k+1)/(2k)$. More explicitly, we can write
\begin{align}
 \label{eq:def-mukt-2} t \cdot \mu^{(C\Gamma)}_{k,t} &= \max(t-s,0)\delta_{-1} + \max(s+t-1,0) \delta_1 + \notag \\
 &\qquad  \frac{\sqrt{(2\varphi^+(s,t)-1-x)(x-2\varphi^-(s,t)+1)}}{2\pi (1-x)(x+1)} \mathbf{1}_{[2\varphi^-(s,t)-1,2\varphi^+(s,t)-1]}(x)dx,
\end{align}
where $\varphi^\pm(s,t)$ were defined in \eqref{eq:phi-st-pm}.
\end{theorem}
\begin{proof}
First, start from the exact moment formula \eqref{eq:moments-CnGamma} and isolate the powers of the parameter $n \to \infty$, using also the Weingarten function asymptotic \eqref{eq:Wg-asympt}. The exponent of $n$ in the fixed $k$ asymptotic regime is:
\begin{align}
\text{power of $n$ in \eqref{eq:moments-CnGamma}} &= -1+\#\alpha + \#(\gamma \beta)-p - |\alpha^{-1} \beta| \notag\\
& = p-1 -( |\alpha| + |\alpha^{-1} \beta| + |\beta^{-1} \gamma^{-1}|) \notag\\
&   \leq p-1 - |\gamma^{-1}| = 0,
\end{align}
where we have only made use of the geodesic inequality
\be  |\alpha| + |\alpha^{-1} \beta| + |\beta^{-1} \gamma^{-1}| \geq |\gamma^{-1}|,\ee
which is saturated for permutations $\alpha,\beta$ lying on the geodesic $\id - \alpha - \beta - \gamma^{-1}$.
Then,
\begin{align}
\lim_{n \to \infty} \mathbb E \frac{1}{kd}\operatorname{Tr}(C_n^\Gamma)^p 
&=(tk^2)^{-1} \sum_{\id - \alpha - \beta - \gamma^{-1}}  k^{\#(\gamma^{-1} \alpha)} (tk)^{\#(\gamma\beta)} k^{-p-|\alpha^{-1}\beta|} \Mob (\alpha^{-1}\beta) \nonumber \\
&= (tk^2)^{-1} \sum_{\id - \alpha - \beta - \gamma^{-1}}  k^{1+ e(\alpha)} (tk)^{1+|\beta|} k^{- \# \alpha - |\beta|} \Mob (\alpha^{-1}\beta) \nonumber \\
\label{eq:asymptotic-moments}&= \sum_{\id - \beta - \gamma^{-1}}   t^{|\beta|} \sum_{\id - \alpha - \beta} k^{ e(\alpha)- \# \alpha } \Mob (\alpha^{-1}\beta).
\end{align}
Above, $e(\cdot)$ denotes the number of cycles of even size of a permutation, and we use the fact \cite[Lemma 2.1]{bn12} that 
\be
1+e(\alpha) = \#(\gamma^{-1} \alpha)
\ee
for permutations $\alpha \in \mathcal S_p$ lying on the geodesic $\id - \alpha - \gamma^{-1}$.

Next, we shall identify a probability measure $\mu^{(C\Gamma)}_{k,t}$ having moments as in equation \eqref{eq:asymptotic-moments}. The main tool here will be the free  moment-cumulant formulas \cite[Lecture 11]{nsp}. 
Consider the following probability measure
\be
\nu_k := \frac {k-1}{2k}\delta_{-1} +\frac {k+1}{2k} \delta_{+1}.
\ee
Since the $\alpha$-moment of $\nu_k$ is ($\|c\|$ denotes the number of elements in a cycle $c$ of a permutation)
\begin{equation}
m_\alpha(\nu_k) =\prod_{c \in \alpha}m_{\|c\|}(\nu_k) = \prod_{c \in \alpha} \left( \frac{k-1}{2k}(-1)^{\|c\|} + \frac{k+1}{2k} \right) = k^{ e(\alpha)- \# \alpha},
\end{equation}
we have that
\be
\eqref{eq:asymptotic-moments} &=& \sum_{\id - \beta - \gamma^{-1}}   t^{|\beta|} \sum_{\id - \alpha - \beta} m_\alpha(\nu_k) \Mob (\alpha^{-1}\beta)\notag\\
&=&  t^p \sum_{\id - \beta - \gamma^{-1}}   t^{-\# \beta} \kappa_\beta (\nu_k) 
= t^p m_p \left(\nu_k^{\boxplus 1/t }\right),
\ee
which is precisely the $p$-th moment of the probability measure 
\be \mu^{(C\Gamma)}_{k,t}=D_t\left[ \nu_k^{\boxplus 1/t} \right]. \ee
This proves the first statement. 

We now show the second statement. By using Lemma \ref{lemma:push-forward}  we have (recall that $b_s = (1-s)\delta_0 + s \delta_1$ is the Bernoulli distribution)
\be
\mu^{(C\Gamma)}_{k,t} &=& \{x\mapsto tx\}_{\#}  \left( \{x \mapsto 2x-1\}_{\#} b_s \right)^{\boxplus 1/t}\notag \\
&=& \{x\mapsto tx\}_{\#} \{x \mapsto 2x-1/t\}_{\#} \left(  b_s \right)^{\boxplus 1/t}
= \{x \mapsto 2tx-1\}_{\#} \left(  b_s \right)^{\boxplus 1/t}
\ee
Therefore, by using Proposition \ref{prop:bernoulli-boxtimes-boxplus} with $T=1/t$ we have
\be
\mu^{(C\Gamma)}_{k,t} &=& \max (0,1-s/t) \delta_{-1} + \max(0,1-(1-s)/t) \delta_{1} \notag\\
&+&\frac {1/t \sqrt{(\gamma_+ - \frac{x+1}{2t})(\frac{x+1}{2t}-\gamma_-)}}{2\pi (\frac{x+1}{2t}) (\frac 1t - \frac{x+1}{2t})} 
\mathbf 1_{[2t\gamma_--1,2t\gamma_+-1]}   \,\frac{ dx}{2t} 
\ee
where $\varphi^{\pm}(s,t)= t \gamma_\pm$. This completes the proof. 
\end{proof}

\begin{remark}\label{rk:atoms} 
The atoms (possibly) appearing in equation \eqref{eq:def-mukt-2} can be interpreted as follows. First, note that the partially transposed Choi matrix $C_L^\Gamma$ is equal, up to a unitary conjugation, to the matrix $(P_n \otimes I_k)(I_n \otimes F_{k,k})(P_n \otimes I_k),$ where $P_n \in \mathcal M_{nk}(\mathbb C)$ is an orthogonal projection of rank $d \sim tnk$ and $F_{k,k}$ is the flip operator. Since the eigenvalues of the flip operator are $+1$, resp.~$-1$, with multiplicities $k(k+1)/2$, resp.~$k(k-1)/2$, by the interlacing theorem for eigenvalues of Hermitian matrices, we get that the matrix has (among others) eigenvalues 
\begin{align*}
+1& \quad \text{ with multiplicity at least } \quad nk^2 \cdot \max\left(t+\frac{k+1}{2k}-1,0 \right)\\
-1& \quad \text{ with multiplicity at least } \quad nk^2 \cdot \max\left(t+\frac{k-1}{2k}-1,0 \right)
\end{align*}
values which corresponds, at the limit, to the atoms at $\pm 1$ of the measure $\mu^{(C\Gamma)}_{k,t}$ from \eqref{eq:def-mukt-2}.
\end{remark}

\subsection{Strong convergence}

In the previous subsection, we showed that the random variable $C_n^\Gamma$ has the same asymptotic moments as the probability measure $\mu^{(C\Gamma)}_{k,t}$. This type of convergence is not sufficient for our purposes, since it does not deal with extremal eigenvalues. We will improve this result with the following proposition. In particular, we will obtain the convergence of the \emph{norm} of the random matrix $C_n^\Gamma$, which is, ultimately, the quantity we are interested in (see Proposition \ref{prop:bound}). For an arbitrary probability measure $\mu$, we denote by $\|\mu\|$ its $L^\infty$ norm, i.e. the $L^\infty$ norm of any random variable $X$ having distribution $\mu$.

We start with a technical lemma, showing the strong convergence (in the sense of Definition \ref{def:strong-convergence}) for a family of deterministic matrices of growing dimension. 

\begin{lemma}\label{lem:strong-convergence-deterministic-matrices}
Let $\{E_{ij}\}_{i,j=1}^k$ be the matrix units of $\mathcal M_k(\mathbb C)$ and define the orthogonal projection
\be \mathcal M_n(\mathbb C) \otimes \mathcal M_k(\mathbb C) \ni P_n = \mathrm{diag}(\underbrace{1, \ldots, 1}_{d_n \text{ times}}, \underbrace{0, \ldots, 0}_{nk - d_n \text{ times}}),\ee
where  $k \in \mathbb N$ and $t \in (0,1)$ are fixed and $d_n$ is a multiple of $k$ such that $d_n \sim tnk$ when $n \to \infty$. Then, the $k^2+1$ tuple of random variables $P_n, \{I_n \otimes E_{ij}\}_{i,j=1}^k$ converge strongly, as $n \to \infty$, towards respectively $p, \{e_{ij}\}_{i,j=1}^k \in \mathcal P_t \otimes \mathcal M_k(\mathbb C)$, where $p=\pi_t \otimes I_k$, $e_{ij} = 1 \otimes E_{ij}$, and $\mathcal P_t$ is the algebra generated by a projection $\pi_t$ of trace $t$ and the identity $1$.
\end{lemma}
\begin{proof}
The result follows from the block structure of the matrices $I_n \otimes E_{ij}$ and the fact that $P_n$ respects this block structure ($d_n$ is a multiple of $k$). We first show the convergence in distribution. In this proof, we denote by $\mathrm{tr}$ the \emph{normalized} trace, e.g.~ $\mathrm{tr}_k I_k = 1$.
For a fixed monomial $F \in \mathbb C \langle p, \{e_{ij}\} \rangle$ in $1+k^2$ non-commutative random variables, we have
\eq{ F(P_n, \{I_n \otimes E_{ij}\})& = F(I_k^{\oplus d_n/k} \oplus 0_k^{\oplus n-d_n/k}, \{E_{ij}^{\oplus n}\}) \notag\\
& = F(I_k, \{E_{ij}\})^{\oplus d_n/k} \oplus F(0_k, \{E_{ij}\})^{\oplus n-d_n/k},
}
and thus 
\begin{equation}\label{eq:trace-n-k}
\lim_{n \to \infty} [\mathrm{tr}_n \otimes \mathrm{tr}_k] F(P_n, \{E_{ij}\}) = t \, \mathrm{tr}_kF(I_k,  \{E_{ij}\}) + (1-t) \, \mathrm{tr}_kF(I_k,  \{E_{ij}\})\mathbf{1}_{p \notin F}.
\end{equation}
Note that since $F$ is a monomial, $p \notin F$ means that $F$ does not contain a factor $p$. 

On the other hand, we have
\be F(p,\{e_{ij}\}) = 
\begin{cases}
\pi_t \otimes F(I_k, \{E_{ij}\}), &\qquad \text{ if } p \in F,\\
1 \otimes F(I_k, \{E_{ij}\}), &\qquad \text{ if } p \notin F.
\end{cases}\ee
If $\tau$ is the trace of $\mathcal P_t$, then we have
\begin{align}
[\tau \otimes \mathrm{tr}_k]&F(p,\{e_{ij}\}) = t \, \mathrm{tr}_kF(I_k,  \{E_{ij}\}) \mathbf{1}_{p \in F} + \mathrm{tr}_kF(I_k,  \{E_{ij}\}) \mathbf{1}_{p \notin F} \notag\\
&= t \, \mathrm{tr}_kF(I_k,  \{E_{ij}\}) \mathbf{1}_{p \in F} + t \, \mathrm{tr}_kF(I_k,  \{E_{ij}\}) \mathbf{1}_{p \notin F} + (1-t) \, \mathrm{tr}_kF(I_k,  \{E_{ij}\}) \mathbf{1}_{p \notin F}\notag\\
&= t \, \mathrm{tr}_kF(I_k,  \{E_{ij}\}) + (1-t) \, \mathrm{tr}_kF(I_k,  \{E_{ij}\}) \mathbf{1}_{p \notin F}
\end{align}
which, together with \eqref{eq:trace-n-k}, allows to conclude the proof of the convergence in distribution for monomials, and, using linearity, for arbitrary non-commutative polynomials. 

Let us now show the norm convergence in Definition \ref{def:strong-convergence}, for a fixed polynomial $F$. Using again the block structure of the matrices $I_n \otimes E_{ij}$ and of $P_n$ (recall that $d_n$ is a multiple of $k$), we get that
\begin{align}
\|F(P_n, \{I_n \otimes E_{ij}\})\| &= \max ( \|F(I_k, \{E_{ij}\})\| , \|F(0_k, \{E_{ij}\})\| ) \notag\\
&= \max ( \|(G+H)(I_k,\{E_{ij}\})\| , \|G(I_k,\{E_{ij}\})\| ),
\end{align}
where $F=G+H$ is the decomposition of $F$ into monomials which do not contain the first variable ($G$) and monomials which do ($H$). Using the same decomposition, we have
\begin{align}
\|F(p, \{e_{ij}\})\| &= \|1 \otimes G(I_k,\{E_{ij}\}) + \pi_t \otimes  H(I_k,\{E_{ij}\})\| \notag \\
&= \| \pi_t \otimes  (G+H)(I_k,\{E_{ij}\}) + (1-\pi_t)  \otimes  G(I_k,\{E_{ij}\}) \|  \notag\\
&= \max ( \|(G+H)(I_k,\{E_{ij}\})\| , \|G(I_k,\{E_{ij}\})\| ),
\end{align}
finishing the proof.
\end{proof}

\begin{theorem}\label{thm:strong-convergence}
The random matrix $C_n^\Gamma$ converges \emph{strongly} towards an element having distribution $\mu^{(C\Gamma)}_{k,t}$ defined in \eqref{eq:def-mukt-1}-\eqref{eq:def-mukt-2}. Hence, we have the following norm convergence: almost surely,
\begin{equation}\label{eq:norm-mukt-1}
\lim_{n \to \infty} \|C_n^\Gamma\| = \|\mu^{(C\Gamma)}_{k,t}\| =
\begin{cases}
2 \varphi^+(s,t)-1, &\quad \text{ if } t+s<1\\
1, &\quad \text{ if } t+s \geq 1,
\end{cases}
\end{equation}
where $s=(k+1)/(2k)$ and $\varphi^+(s,t)$ was defined in \eqref{eq:phi-st-pm}. 
More explicitly, the above quantity is written as 
\begin{equation}\label{eq:norm-mukt-2}
2 \varphi^+(s,t)-1 = \frac{1-2t}{k} + 2\sqrt{\left(1-\frac 1 {k^2}\right) t(1-t)}.
\end{equation}
\end{theorem}
\begin{proof}
Firstly, let $\varepsilon_n = |\frac {d_n}{nk}-t|$ so that 
\be
(t-\varepsilon_n)nk \leq d_n\leq (t+\varepsilon_n) nk.
\ee
Then, we set $d_n^- = \lfloor n (t-\varepsilon_n) \rfloor k$ and $d_n^+=\lceil n(t+\varepsilon_n) \rceil k$ so that 
we can use Lemma \ref{lem:strong-convergence-deterministic-matrices} for $d_n^-$ and $d_n^+$.
Hence, the following proofs are applied to $d_n^-$ and $d_n^+$ but
on the other hand by using interlacing theorem, 
we can obtain the desired statement for the original sequence $d_n$.

Secondly, we prove the strong convergence. 
Recall that the channel $L$ has the following Stinespring representation
\begin{align}
L(X) &= [\mathrm{id}_k \otimes \mathrm{Tr}_n](VXV^*)\notag \\
&= [\mathrm{id}_k \otimes \mathrm{Tr}_n](UWXW^*U^*),
\end{align}
where $U \in \mathcal U(nk)$ is a unitary operator and $W: \mathbb C^d \to \mathbb C^k \otimes \mathbb C^n $ is the isometric embedding: $W=[I_d \, 0_{d \times (nk-d)}]^\top$. With this notation, we have
\be C_n^\Gamma = \sum_{i,j=1}^d [\mathrm{id}_k \otimes \mathrm{Tr}_n](UWE_{ij}W^*U^*) \otimes E_{ji} \in \mathcal M_k(\mathbb C) \otimes \mathcal M_d(\mathbb C).\ee
Define
\begin{align}
\nonumber \mathcal M_n(\mathbb C) \otimes \mathcal M_k(\mathbb C) \otimes \mathcal M_k(\mathbb C) \ni D_n &= (P_d U^*  \otimes I_k ) \left[ I_n \otimes  \sum_{i,j=1}^k E_{ij} \otimes E_{ji} \right] (UP_d \otimes I_k) \\
\label{eq:D-n} &= \sum_{i,j=1}^k  \left[ P_d U^*  \left( I_n \otimes  E_{ij} \right)  UP_d \right] \otimes E_{ji},
\end{align}
where $P_d = WW^*=I_d \oplus 0_{nk-d}$ is an orthogonal projection of rank $d$. It is easy to see that the matrices $C_n^\Gamma$ and $D_n$ have the same spectrum, up to null eigenvalues, so we will focus on showing the strong convergence property for $D_n$. Note that the expression in \eqref{eq:D-n} is a polynomial with $\mathcal M_k(\mathbb C)$ coefficients, in the variables $P_d$, $U$ and $I_n \otimes E_{ij}$. Using Lemma \ref{lem:strong-convergence-deterministic-matrices} and Proposition \ref{prop:CM}, we conclude that the tuple $(U,U^*, P_d, \{I_n \otimes E_{ij}\}_{i,j=1}^k)$ converges strongly, as $n \to \infty$, to a limit that we choose not to specify. 
We then apply Proposition \ref{prop:M} to obtain the strong convergence of $D_n$ to a limit element $x$. Since the distribution of the random matrix $C_n^\Gamma$ has been shown to converge to the measure $\mu^{(C\Gamma)}_{k,t}$, we conclude that, almost surely 
\be \lim_{n \to \infty} \left\| C_n^\Gamma \right\|_\infty =\|\mu^{(C\Gamma)}_{k,t}\|.\ee

Thirdly, notice that the norm $\|C_n^\Gamma\|$ is the maximum between the largest element of the support of the measure $\mu^{(C\Gamma)}_{k,t}$ and minus the smallest element of the support of $\mu^{(C\Gamma)}_{k,t}$. Let us show now that the latter quantity is almost smaller or equal that the former, finishing the proof. 
Indeed, in the case where $\mu^{(C\Gamma)}_{k,t}$ has an atom at $-1$, $t$ should be larger than $s$. But in that case, we also have $s+t-1>2s-1>0$, since $s=(k+1)/(2k)>1/2$, so $\mu^{(C\Gamma)}_{k,t}$ also has an atom at $1$. In the case when $t<1-s$, we have that $t\leq 1/2$ and thus $2\varphi^+(s,t) - 1 \geq -(2\varphi^-(s,t) - 1)$, showing that the maximum is also attained on the positive part of the support.
\end{proof}

\section{Other bounds for random quantum channels}
\label{sec:other-bounds-random-channels}

In this section, we compute the remaining four bounds from Section \ref{sec:additive-bounds}, in the case of random quantum channels. In the first three subsections, we study respectively the bounds \eqref{eq:B-C}, \eqref{eq:B-C-c-Gamma}, \eqref{eq:B-M}, while in the last one we compare these three bounds with \eqref{eq:B-C-Gamma}, which was computed in the previous section. Our conclusion is that the bound \eqref{eq:B-C-Gamma} coming from the partially transposed Choi matrix of the channel $L$ seems to give the sharpest estimate for minimum $p$-R\'enyi output entropies. Moreover, from a practical standpoint, we have a closed formula for the bound in \eqref{eq:B-C-Gamma} and we have shown strong convergence of the random matrices towards the corresponding probability distribution. 

The computations in this section are similar to the ones in Section \ref{sec:random-Choi}, so we refer the reader to that section for some of the details. We would like to mention already that in some cases, we are not able to obtain such precise estimates as in the previous section; in such situations, we state only some partial results, in the asymptotic limit where the output dimension $k$ is large.

We consider first the bound \eqref{eq:B-I} and we show that, in the case of random quantum channels, it is trivial. Hence, we shall not mention it in the remainder of the paper. Indeed, the (random) channels we are interested in have input dimension larger than output dimension: $L:\mathcal M_d(\mathbb C) \to \mathcal M_k(\mathbb C)$, with $d \to \infty$ and $k$ fixed. Using $d \geq k$, we have (see also \cite{fri})
\begin{equation}
\|L(I_d)\| \geq \frac{\mathrm{Tr} L(I_d)}{k} = \frac{\mathrm{Tr} I_d}{k} = d/k \geq 1,
\end{equation}
and thus the bound in \eqref{eq:B-I} reads $H_\infty^{\min}(L^{\otimes r}) \geq  0 \geq - r \log \|L(I_d)\|$, which is a trivial statement.

\subsection{Choi matrices}  
In this section, we discuss on the limiting eigenvalue distributions of the Choi matrix $C_L$ when taken randomly. As before, we consider a sequence of random quantum channels $L_n$ as in \eqref{eq:random-quantum-channel}, where the parameters scale as in Section \ref{sec:limiting-spectrum-C-Gamma}. We denote by $C_n \in \mathcal M_k(\mathbb C) \otimes \mathcal M_d(\mathbb C)$ the Choi matrix of $L_n$, which is a random quantum channel. 

\begin{proposition}\label{prop:random-Choi}
The random Choi matrix $C_n$ converges \emph{strongly} towards an element having distribution 
\begin{equation}\label{eq:mu-C}
\mu^{(C)}_{k,t} =  D_{kt} \left[ b_{k^{-2}}^{\boxplus 1/t} \right].
\end{equation}
Hence, we have the following norm convergence: almost surely,
\begin{equation}\label{eq:norm-mu-C}
\lim_{n \to \infty} \|C_n\| = \|\mu^{(C)}_{k,t}\| = 
\begin{cases}
k \varphi^+ (k^{-2}, t) & \text{if } t + k^{-2} < 1\\
k & \text{if }   t + k^{-2} \geq 1,
\end{cases}\end{equation}
where the function $\varphi^+$ was defined in \eqref{eq:phi-st-pm}. 
\end{proposition}
\begin{proof}
The starting point of the proof is the following moment formula for the random matrix $C_n$, obtained by graphical Weingarten calculus (see Theorem \ref{thm:graphical-Wg}). The powers of $n$, $d \sim tkn$ and $k$ in the formula below count loops and can be inferred from Figure \ref{fig:Choi-moment}. For any integer $p \geq 1$, we have (remember that $\gamma \in \mathcal S_p$ denotes the full cycle permutation $i \mapsto i-1$)
\begin{figure}[htbp] 
\includegraphics{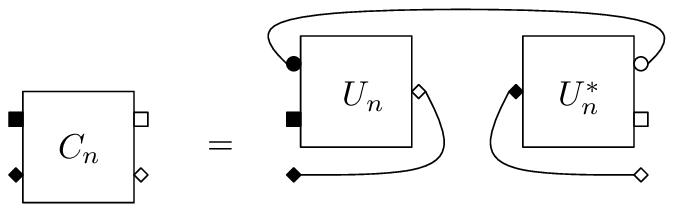} \qquad\qquad
\includegraphics{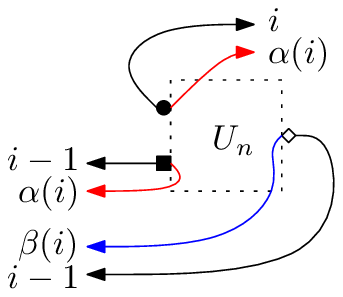} 
\caption{Diagrams for the Choi matrix $C_n$ and for the $i$-th $U_n$ box in the graphical expansion of $\trace \left [ C_n^p\right]$.} 
\label{fig:Choi-moment}
\end{figure}
\be\label{eq:moment-choi} 
(kd)^{-1} \Ex \trace \left [ C_n^p\right] 
=  (kd)^{-1} \sum_{\alpha,\beta \in \mathcal S_p}  n^{\#\alpha} k^{\#(\gamma^{-1}\alpha)} d^{\#(\gamma^{-1} \beta)} \Wg_{kn}(\alpha^{-1}\beta).
\ee

In the above equation, the surviving terms as $n\to\infty$ are the ones which contain the largest power of $n$ (we use below $d \sim tkn$ and the Weingarten function asymptotic from \eqref{eq:Wg-asympt}):
\be
\text{power of $n$ in }\eqref{eq:moment-choi} &=& -1+ \# \alpha +\#(\gamma^{-1} \beta) -p - |\alpha^{-1} \beta| \notag\\
 & \leq& p-1- (|\alpha| + |\alpha^{-1} \beta| + |\beta^{-1} \gamma|) \leq 0 .
\ee
 The above bound is saturated if and only if the triangle inequality $|\alpha| + |\alpha^{-1} \beta| + |\beta^{-1} \gamma| \geq |\gamma|$ is saturated, i.e.~ the permutations $\alpha, \beta$ are on the geodesic $\mathrm{id} - \alpha - \beta - \gamma$.
 This implies that 
\be
\lim_{n\to\infty} (kd)^{-1} \Ex \trace \left [ C_n^p\right] 
&=& (tk^2)^{-1}\sum_{\mathrm{id} - \alpha - \beta - \gamma} t^{p-|\gamma^{-1} \beta| }\,
k^{2p-|\gamma^{-1} \alpha|-|\gamma^{-1} \beta|-p-|\alpha^{-1}\beta|} \,\Mob(\alpha,\beta) \notag\\
&=&   k^{-p}\sum_{\mathrm{id} - \alpha - \beta - \gamma} t^{ |\beta| }\,
k^{2|\alpha|} \,\Mob(\alpha,\beta) \notag \\
&=&   (tk)^{p}\sum_{\mathrm{id} - \alpha - \beta - \gamma} t^{ - \#\beta }\,
k^{-2\#\alpha} \,\Mob(\alpha,\beta) .
\ee
Above, we have used some properties of geodesics, for example, $|\gamma^{-1}\beta| = |\gamma| - |\beta| = p-1-|\beta|$. 

Next, we fix $\beta \in \mathcal S_p$ and use the moment-cumulant formula \cite[Proposition 11.4 (2)]{nsp} in free probability:
\be
\sum_{\mathrm{id} - \alpha - \beta}  k^{-2\# \alpha}  \Mob (\alpha, \beta) = \sum_{\mathrm{id} - \alpha - \beta}  m_\alpha(b_{k^{-2}})  \Mob (\alpha, \beta) = \kappa_\beta(b_{k^{-2}})  ,
\ee
where $b_{k^{-2}}$ is the Bernoulli distribution $b_{k^{-2}} =(1-k^{-2}) \delta_0 + k^{-2}\delta_1$.  
Hence, using the multi-linearity of the free cumulant, we have
\begin{align}
\lim_{n\to\infty} (kd)^{-1} \Ex \trace \left [ C_n^p\right] 
&=(tk)^p \sum_{\mathrm{id}-\beta-\gamma}  t^{-\#\beta} \kappa_\beta(b_{k^{-2}}) \notag\\
&=(tk)^p \sum_{\mathrm{id}-\beta-\gamma}  \kappa_\beta(b_{k^{-2}}^{\boxplus 1/t}) \notag \\
&=(tk)^p m_p \left(b_{k^{-2}}^{\boxplus 1/t}\right) = m_p \left( D_{kt} \left[ b_{k^{-2}}^{\boxplus 1/t} \right] \right),
\end{align}
which shows that the random matrices $C_n$ converge in moments to the probability measure $\mu^{(C)}_{k,t}$ defined in \eqref{eq:mu-C}. The statement about the support of $\mu^{(C)}_{k,t}$ follows from Proposition \ref{prop:bernoulli-boxtimes-boxplus} with $T=1/t$ and $s=k^{-2}$. The proof of the strong convergence is identical to the one from Theorem \ref{thm:strong-convergence}, up to the occasional flipping of indices, due to the partial transposition (e.g.~ in \eqref{eq:D-n}, one should replace the matrix element $E_{ji}$ by $E_{ij}$); we leave the details to the reader. 
\end{proof}

Note that the limiting probability measure $\mu_{k,t}^{(C)}$ from \eqref{eq:mu-C} is supported on $[0,\infty)$; this is a consequence of the fact that the Choi matrices $C_n$ are positive semidefinite, since the quantum channels $L_n$ are completely positive. 

\subsection{Partially transposed Choi matrices of complementary channels}
\label{sec:limiting-spectrum-C-c-Gamma}

In this subsection, we discuss the limiting eigenvalue distributions of the partially transposed Choi matrix of the complementary channel $C_{L^c}^\Gamma$, our goal being to estimate the quantity \eqref{eq:B-C-c-Gamma} in the case where $L$ is a random quantum channel. More precisely, we consider a sequence of random quantum channels $L_n$ and we denote by $C^{c\Gamma}_n$ the corresponding partially transposed Choi matrices of the channels $L_n^c$. Unfortunately, we are not able to identify, at fixed $k$ and $t$, the limiting eigenvalue distribution of the random matrix $C_n^{c\Gamma}$ when $n \to \infty$; we have to settle in this case for another asymptotic regime, where $1 \ll k \ll n$. This regime is obtained by taking first the limit $n \to \infty$, followed by the limit $k \to \infty$. 

Before stating the result, we recall the notion of \emph{Kreweras complement} for non-crossing partitions. The Kreweras complement of $\alpha \in NC(p)$ is another non-crossing partition, denoted $\alpha^\mathrm{Kr} \in NC(p)$, defined in the following way \cite[Definition 9.21]{nsp}. First, expand the domain of partitions to $\{\bar 1, 1,  \bar2,2 \ldots, \bar p,  p\}$; let then  $\alpha^\mathrm{Kr} \in NC(\bar 1,\bar 2, \ldots, \bar p) \cong NC(p)$ be the \emph{largest} non-crossing partition such that $\alpha \cup \alpha^\mathrm{Kr}$ is still a non-crossing partition on $\{\bar 1, 1,  \bar2,2 \ldots, \bar p,  p\}$. Given a geodesic permutation $\mathrm{id}-\alpha-\gamma$, we define the geodesic permutation $\alpha^\mathrm{Kr} \in \mathcal S_p$ by identifying, as usual, geodesic permutations with non-crossing partitions; more precisely, we have $\alpha^\mathrm{Kr} = \alpha^{-1} \gamma$ (see \cite[Remark 23.24]{nsp}). Note that the above construction of Kreweras complement is slightly different from the one in \cite[Definition 9.21]{nsp}
in that $\bar i$ is left to $i$ for $i = 1,\ldots,p$ because  $\gamma = (p,p-1,\ldots,1)$ in our paper. 

\begin{lemma}\label{lem:twisting}
For $ \gamma_{2q} = (2q, 2q-1, \ldots, 1) \in \mathcal S_{2q}$, 
$\gamma_1= (2q-1,2q-3,\ldots,1)$ and $\gamma_2= (2q,2q-2,\ldots,2)$,
take $\tilde \alpha_i \in \mathcal S_{2q}$ two permutations on the geodesics $\mathrm{id} - \alpha_i- \gamma_i $ with $i=1,2$ (in particular, $\alpha_1$ acts only on odd numbers, while $\alpha_2$ acts only on even numbers). Then,
\be
\# \left(\gamma_{2q}^{-1} (\tilde \alpha_1 \oplus \tilde \alpha_2)\right) = \#(\alpha_1\gamma_q^{-1}  \alpha_2 ) = \# \left(\alpha_1\left(\alpha_2^\mathrm{Kr}\right)^{-1}  \right),
\ee
where the permutations $\alpha_{1,2} \in  \mathcal S_q$ are associated with $\tilde \alpha_{1,2}$ by deleting the trivial fixed points, via the correspondences $(2i-1) \to i$, resp.~ $2i \to i$, for $i=1,2,\ldots, q$.
\end{lemma}
\begin{proof}
We show $\# \left(\gamma_{2q}^{-1} (\tilde \alpha_1 \oplus \tilde \alpha_2)\right) = \#(\alpha_1\gamma_q^{-1}  \alpha_2 )$, the other equality following from the definition of the Kreweras complement for geodesic permutations. First, note that the respective actions on $\tilde \alpha_{1,2}$ are 
$\tilde \alpha_1(2i) = 2i$, $\tilde \alpha_1(2i-1) = 2 \alpha_1(i)-1$, resp.~ $\tilde \alpha_2(2i) = 2\alpha_2(i)$, $\tilde \alpha_2(2i-1) = 2 i-1$, for $i=1,2,\ldots, q$. An arbitrary element $i \in [q]$ is mapped by the permutation $\alpha_1\gamma_q^{-1}  \alpha_2$ to $\alpha_1(\alpha_2(i)+1)$. Let us compute the image of the corresponding element $2i \in [2q]$ through the permutation $\gamma_{2q}^{-1} (\tilde \alpha_1 \oplus \tilde \alpha_2)$:
\be 2i \xrightarrow{\tilde \alpha_1 \oplus \tilde \alpha_2} 2\alpha_2(i) \xrightarrow{\gamma_{2q}^{-1}} 2\alpha_2(i) +1,\ee
which is an odd number. Another application of $\gamma_{2q}^{-1} (\tilde \alpha_1 \oplus \tilde \alpha_2)$ yields
\be 2\alpha_2(i) +1 = 2(\alpha_2(i) +1)-1 \xrightarrow{\tilde \alpha_1 \oplus \tilde \alpha_2} 2 \alpha_1(\alpha_2(i)+1) - 1 \xrightarrow{\gamma_{2q}^{-1}} 2 \alpha_1(\alpha_2(i)+1),\ee
which establishes a bijection between the cycles of $\gamma_{2q}^{-1} (\tilde \alpha_1 \oplus \tilde \alpha_2)$ and $\alpha_1\gamma_q^{-1}  \alpha_2$, finishing the proof.
\end{proof}

\begin{remark}\label{rem:meanders}
In the proof of Lemma \ref{lem:twisting}, $\gamma_{2q}^{-1} (\tilde \alpha_1 \oplus \tilde \alpha_2)$ is 
the number of loops in meanders \cite{DGG} constructed by the two non-crossing partitions $\alpha_{1,2}$, see \cite{FSn,fbe}.
\end{remark}

To state the next result, we introduce the ``symmetric square root'' operator $R$ acting on probability measures supported on the positive real line:
\begin{equation}\label{eq:ssr}
R[\mu] := \frac 1 2 (\sqrt x)_\# \mu + \frac 1 2(-\sqrt x)_\# \mu.
\end{equation} 

\begin{proposition}\label{prop:random-Choi-c-Gamma}
In the asymptotic regime $1 \ll k \ll n$, the random matrix $C_n^{c\Gamma}$ converges in moments to the probability measure
\begin{equation}\label{eq:mu-C-c-Gamma}
\mu^{(Cc\Gamma)}_{t} =  (R \circ D_{t}) \left[ b_{t}^{\boxplus 1/t} \right],
\end{equation}
for which we have
\begin{equation}\label{eq:norm-mu-C-c-Gamma}
 \|\mu^{(Cc\Gamma)}_{t}\| = 
\begin{cases} 
2\sqrt{t(1-t)}& \text{if } t \leq 1/2\\
1&\text{if } t>1/2.
\end{cases} 
\end{equation}
\end{proposition}
\begin{proof}
As usual, the first step in the proof is a moment formula, valid at any fixed dimensions $n,k,d$. Using the graphical Weingarten formula (see Figure \ref{fig:Choi-c-Gamma-moment}), we have
\begin{figure}[htbp] 
\includegraphics{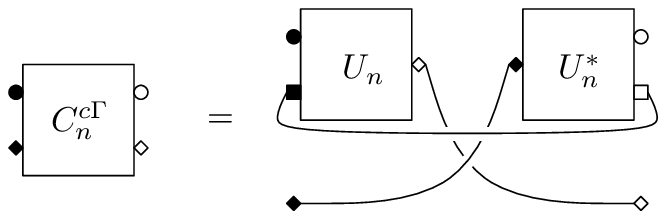} \qquad\qquad
\includegraphics{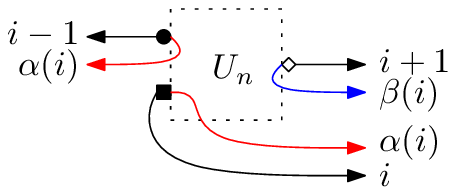} 
\caption{Diagrams for the partially transposed Choi matrix $C_n^{c\Gamma}$ of the complementary channel and for the $i$-th $U_n$ box in the graphical expansion of $\trace \left [ (C_n^{c\Gamma})^p\right]$.} 
\label{fig:Choi-c-Gamma-moment}
\end{figure}
\begin{equation}\label{eq:moments-C-c-Gamma}
(nd)^{-1} \Ex\trace \left[ \left(C_n^{c\Gamma}\right)^p \right]  
= (nd) ^{-1}\sum_{\alpha,\beta \in \mathcal S_p} n^{\#(\gamma^{-1} \alpha)} k^{\#\alpha} d^{\#(\gamma\beta)} \Wg_{kn}(\alpha^{-1}\beta).
\end{equation}
We compute the asymptotic moments of $C_n^{c\Gamma}$ by taking two successive limits, first in $n$ and then in $k$. The power of $n$ in the moment formula above reads
\begin{align}
\text{power of $n$ in \eqref{eq:moments-C-c-Gamma}} &=   -2 + p-|\gamma^{-1} \alpha| + p- |\gamma\beta| -p - |\alpha^{-1}\beta| \notag \\
&=  p-2 - (|\gamma \beta| + |\beta^{-1}\alpha|+ |\alpha^{-1}\gamma| ) \notag \\
&\leq p-2 - |\gamma^2| \leq 0,
\end{align}
where we have used the triangle inequality and the fact that $\gamma^2$ has either one or two cycles, depending on whether $p$ is respectively odd or even. The above bound is saturated if and only if $p$ is an even number and the permutations $\alpha, \beta$ lie on the geodesic $\gamma^{-1}  - \beta- \alpha - \gamma$. Hence, for even integers $p=2q$, $q \geq 1$, we have  
\be
\lim_{n\to\infty}(nd)^{-1} \Ex\trace \left[ \left(C_n^{c\Gamma}\right)^{p} \right]  
&=&  (tk)^{-1}\sum_{\gamma^{-1}  - \beta- \alpha - \gamma} k^{2p- |\alpha| -|\gamma \beta|-p-|\alpha^{-1}\beta|} \, t^{p-|\gamma\beta|} \Mob(\alpha,\beta)  \notag\\
&=&  (tk)^{-1}\sum_{\gamma^{-1} - \beta- \alpha - \gamma} k^{p- |\alpha| - |\gamma\alpha|} \, t^{p-|\gamma\beta|} \Mob(\alpha,\beta)  \notag\\
&=&  (tk)^{-1}\sum_{\gamma^{-1} - \beta- \alpha - \gamma} k^{\#\alpha - |\gamma\alpha|} \, t^{\#(\gamma\beta)} \Mob(\alpha,\beta) . 
\label{eq:ptc-moment1}
\ee
Then, we shift our permutations and work on the geodesic: $\mathrm{id} - \gamma\beta -  \gamma\alpha - \gamma^2$.
Importantly, $\gamma^2$ decomposes as $\gamma^2 = \gamma_1 \oplus \gamma_2$ with $\gamma_1 = (2q-1, 2q-3,\ldots,1)$ and $\gamma_2 = (2q, 2q-2, \ldots, 2)$. Permutations on the geodesic   $\mathrm{id} - \gamma\beta -  \gamma\alpha - \gamma^2$ admit the same decomposition, so we write $\gamma\alpha = \alpha_1 \oplus \alpha_2$ and $\gamma \beta  = \beta_1 \oplus \beta_2$. Using the moment-cumulant formula, we obtain
\begin{align}
\eqref{eq:ptc-moment1}
&=(tk)^{-1}\sum_{\substack{\mathrm{id} - \beta_1 - \alpha_1 - \gamma_1\\\mathrm{id} - \beta_2 - \alpha_2 - \gamma_2}} 
k^{\#(\gamma^{-1} (\alpha_1 \oplus \alpha_2)) - |\alpha_1 \oplus \alpha_2|} \, t^{ \#(\beta_1 \oplus \beta_2)} \,
\Mob(\alpha_1,\beta_1)\,\Mob(\alpha_2,\beta_2)\notag\\
&=(tk)^{-1}\sum_{\substack{\mathrm{id} - \alpha_1 - \gamma_1\\\mathrm{id} - \alpha_2 - \gamma_2}} 
k^{\#(\gamma^{-1} (\alpha_1 \oplus \alpha_2)) - |\alpha_1| - |\alpha_2|} \left[\sum_{\mathrm{id} - \beta_1 - \alpha_1} t^{\#\beta_1} \Mob(\alpha_1,\beta_1) \right]  \notag \\
& \qquad \qquad \times \left[\sum_{\mathrm{id} - \beta_2 - \alpha_2} t^{\#\beta_2} \Mob(\alpha_2,\beta_2) \right]  \notag\\
&=(tk)^{-1}\sum_{\substack{\mathrm{id} -  \alpha_1 - \gamma_1\\\mathrm{id} - \alpha_2 - \gamma_2}} 
k^{\#(\gamma^{-1} (\alpha_1 \oplus \alpha_2)) - |\alpha_1| -|\alpha_2|} \, \kappa_{\alpha_1} (b_t)  \, \kappa_{\alpha_2}(b_t). 
\label{eq:ptc-moment2}
\end{align}
Unfortunately, we are not able to identify a probability measure having these moments. One of the reasons for this is the relation between the sum above and the combinatorics of meanders \cite{DGG}, see Remark \ref{rem:meanders}. We are thus taking the second limit, $k \to \infty$. To do so, we calculate the power of $k$ in \eqref{eq:ptc-moment2} by using Lemma \ref{lem:twisting}. In what follows, we abuse notation by writing, as in Lemma \ref{lem:twisting},  $\alpha_{1,2}$ for the permutations in $\mathcal S_q$ obtained by restricting the previous $\alpha_{1,2} \in \mathcal S_{2q}$ on odd, resp.~ even numbers; note that by doing this, the quantities $|\alpha_{1,2}|$ remain invariant. We get:
\be
\text{power of $k$ in \eqref{eq:ptc-moment2}} &=& -1+ \#( \alpha_1(\alpha_2^\mathrm{Kr})^{-1})  - |\alpha_1| - |\alpha_2| \notag\\
 &=& |\alpha_2^\mathrm{Kr}|  - ( |\alpha_1| + |\alpha_1^{-1} \alpha_2^\mathrm{Kr} | ) \leq 0, 
\ee
where the bound is saturated if and only if $\alpha_1, \alpha_2 \in \mathcal S_q$ are on the geodesic $\mathrm{id} - \alpha_2 - \alpha_1^\mathrm{Kr}$.
Therefore, using repeatedly the moment cumulant formula, we obtain 
\begin{align}
\lim_{k\to\infty} \eqref{eq:ptc-moment2}
&= t^{-1} \sum_{\mathrm{id} - \alpha_1 - \alpha_2^\mathrm{Kr} - \gamma_q} \kappa_{\alpha_1}(b_t) \, \kappa_{\alpha_2} (b_t) 
= t^{-1} \sum_{\mathrm{id} -  \alpha_2^\mathrm{Kr} - \gamma_q}  m_{\alpha_2^\mathrm{Kr}} (b_t)  \,\kappa_{\alpha_2}(b_t) \notag\\
&=  t^{q} \sum_{\mathrm{id} -  \alpha_2^\mathrm{Kr} - \gamma_q} t^{- \#(\alpha_2)}  \,\kappa_{\alpha_2}(b_t)
=  t^{q} \sum_{\mathrm{id} -  \alpha_2^\mathrm{Kr}  - \gamma_q}\kappa_{\alpha_2}(b_t^{\boxplus 1/t}) 
= m_q \left(D_t \left[ b_t^{\boxplus 1/t} \right] \right).  
\end{align}
Finally, we have, for all integers $q \geq 1$,
\be \lim_{k\to\infty}\lim_{n\to\infty} (nd)^{-1} \Ex\trace \left[ \left(C_{n}^{c\Gamma}\right)^{2q} \right]  
= m_q\left(D_t \left[b_t^{\boxplus 1/t}\right]\right), \ee
and
\be \lim_{k\to\infty}\lim_{n\to\infty} (nd)^{-1} \Ex\trace \left[ \left(C_{n}^{c\Gamma}\right)^{2q+1} \right]  
= 0.\ee
To identify the distribution $\mu_{t}^{(Cc\Gamma)}$ as in \eqref{eq:mu-C-c-Gamma},
we have to modify $D_t \left[b_t^{\boxplus 1/t}\right]$ by $(x\to \sqrt x)_{\#}$ to have right powers for the moments and 
make it symmetric as the odd moments vanish. This process is done by \eqref{eq:ssr}.
To specify the support of $D_t \left[b_t^{\boxplus 1/t}\right] $, 
we use Proposition \ref{prop:bernoulli-boxtimes-boxplus} with $T=1/t$ and $s=t$, and since $\varphi^+(t,t)  = 4t(1-t)$, we obtain the conclusion about the support of $\mu_{t}^{(Cc\Gamma)} $, equation \eqref{eq:norm-mu-C-c-Gamma}.
\end{proof}

\subsection{Partially transposed random projections}
In this subsection, we turn to the study of the bound \eqref{eq:B-M} for random quantum channels. Using the fact that $M=VV^*$ is a (random) projection, we are actually interested in the operator norm of the partial transposition of a random projection -- this was Montanaro's point of view in \cite{mon}, where this question was studied, in a different asymptotic regime. For a sequence of random quantum channels $L_n$, we compute the limiting eigenvalue distribution $\mu_{k,t}^{(M\Gamma)}$ of the partially transposed random projection $M_n^\Gamma$, which is expressed in a free probabilistic framework. Unfortunately, we are not able to compute the support of $\mu_{k,t}^{(M\Gamma)}$, so we have to settle with some partial information: as in the previous section, we consider the limit $\mu_{t}^{(M\Gamma)} = \lim_{k \to \infty} \mu_{k,t}^{(M\Gamma)}$, and we compute the support of this simpler measure. Below, we write $\mu \boxminus \nu := \mu \boxplus D_{-1}[\nu]$.

Recall that the \emph{semicircular} probability distribution of mean $m$ and standard deviation $\sigma$ is given by
\begin{equation}
\mathrm{SC}_{m,\sigma} = \frac{\sqrt{4\sigma^2-(x-m)^2}}{2\pi \sigma^2} \mathbf{1}_{[m-2\sigma ,m+ 2\sigma]}(x) dx.
\end{equation}

\begin{proposition}\label{prop:random-M-Gamma}
The partially transposed random projection $M_n^\Gamma$ converges, in moments, towards the probability measure 
\begin{equation}\label{eq:mu-M}
\mu^{(M\Gamma)}_{k,t} =  D_{1/k} \left[ b_t^{\boxplus \frac {k(k+1)}2}  \boxminus b_t^{\boxplus \frac {k(k-1)}2}\right].
\end{equation}
The probability measures $\mu_{k,t}^{(M\Gamma)}$ converge, in distribution, as $k \to \infty$, towards
\begin{equation}\label{eq:mu-M-t}
\lim_{k \to \infty} \mu_{k,t}^{(M\Gamma)} = \mu_{t}^{(M\Gamma)} = \mathrm{SC}_{t,\sqrt{t(1-t)}},
\end{equation}
a semi-circular distribution with mean $t$ and varaince $t(1-t)$. In particular, we have
\begin{equation}\label{eq:norm-mu-M-t}
\| \mu_{t}^{(M\Gamma)} \| = t+2\sqrt{t(1-t)}.
\end{equation}
\end{proposition}
\begin{proof}
We start with a moment formula, obtained via graphical Weingarten calculus (see Figure \ref{fig:M-Gamma-moment} for an explanation of the exponents appearing below): for any integer $p \geq 1$, 
\begin{figure}[htbp] 
\includegraphics{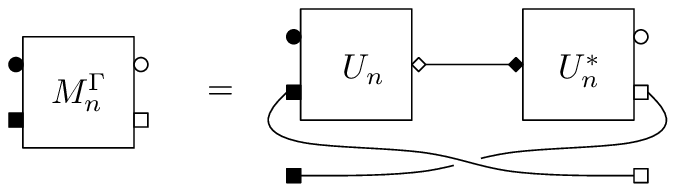} \qquad\qquad
\includegraphics{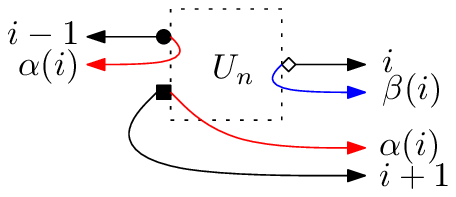} 
\caption{Diagrams for the partially transposed projection $M_n^\Gamma$ and for the $i$-th $U_n$ box in the graphical expansion of its $p$-th moment.} 
\label{fig:M-Gamma-moment}
\end{figure}
\begin{equation}\label{eq:moment-M-Gamma}
\frac{1}{nk} \Ex \trace \left[ (M_n^\Gamma)^p \right] = \frac{1}{nk} \sum_{\alpha,\beta \in \mathcal S_p} n^{\#(\gamma^{-1}\alpha)} k^{\#(\gamma\alpha)} d^{\#\beta}  \Wg (\alpha^{-1}\beta). 
\end{equation}
To obtain the surviving terms as $n \to \infty$, we maximize the power of $n$ for each term in the equation above (recall that $d \sim tkn$)
\begin{align}
\text{power of $n$ in \eqref{eq:moment-M-Gamma}} &= -1 + \#(\gamma^{-1}\alpha) + \#\beta -p-|\alpha^{-1}\beta|  \notag\\
&= p-1 - ( |\beta| + |\beta^{-1}\alpha|+|\alpha^{-1}\gamma|) \notag \\
&\leq p-1 - |\gamma| = 0.
\end{align}
The bound above is saturated if and only if the permutations $\alpha,\beta \in \mathcal S_p$ lie on the geodesic $\mathrm{id} - \beta - \alpha - \gamma$. This implies that
\eq{
\lim_{n\to\infty}(nk)^{-1} \Ex \trace \left[ \left(M_n^\Gamma \right)^p \right] 
&= k^{-1} \sum_{\mathrm{id} - \beta - \alpha - \gamma} k^{\#(\gamma\alpha)} (tk)^{\#\beta} k^{-p-|\alpha^{-1}\beta|}\Mob(\alpha,\beta) \notag\\
&= k^{-p} \sum_{\mathrm{id} - \beta - \alpha - \gamma} k^{e(\alpha) + \#\alpha } t^{\#(\beta)}\Mob(\alpha,\beta)\notag\\
&= k^{-p} \sum_{\mathrm{id}  - \alpha - \gamma} k^{e(\alpha) + \#\alpha } \kappa_\alpha(b_t),
\label{eq:moment-ptp1}
}
where we have made use again of the fact that, for a geodesic permutation $\mathrm{id} - \alpha - \gamma$, we have $\#(\gamma\alpha) = 1+ e(\alpha)$, where $e(\alpha)$ denotes the number of cycles of $\alpha$ having even length (see \cite[Lemma 2.1]{bn12} for a proof). Importantly, the general term in the sum above is a function which is multiplicative on the cycles of $\alpha$:
\eq{
k^{e(\alpha) + \#(\alpha) } \kappa_\alpha(b_t)  = \prod_{c \in \alpha} f(\|c\|) 
}
where $\|c\|$ is the length of a cycle $c$, and $f$ is defined by 
\be
f(r) =  \kappa_r(b_t) \times \begin{cases}
k  & \text{if $r$ is odd} \\
k^2 &\text{if $r$ is even}.
\end{cases}
\ee
On the other hand, given two probability measures $\nu_1$ and $\nu_2$, we have
\be
\kappa_r (\nu_1 \boxminus \nu_2) =\kappa_r (\nu_1 ) + (-1)^r \kappa_r (\nu_2) 
\ee
and we notice in fact that 
\be
f(r) = \kappa_r \left(b_t^{\boxplus \frac {k(k+1)}2}\right) + (-1)^r \kappa_r \left(b_t^{\boxplus \frac {k(k-1)}2}\right).  
\ee
Using the moment-cumulant formula, we prove our first claim:
\be \lim_{n\to\infty}(nk)^{-1} \Ex \trace \left[ \left(M_n^\Gamma \right)^p \right] = m_p\left( D_{1/k} \left[ b_t^{\boxplus \frac {k(k+1)}2}  \boxminus b_t^{\boxplus \frac {k(k-1)}2}\right] \right).\ee
Since we are not able to analytically describe the support of the measure $\mu_{k,t}^{(M\Gamma)}$ above, we take limit $k\to\infty$ in \eqref{eq:moment-ptp1}. For each geodesic permutation $\mathrm{id} - \alpha - \gamma$, we claim that
\[
\text{power of $k$ in \eqref{eq:moment-ptp1}}  = -p+e(\alpha) +\#(\alpha) \leq 0.
\]
Indeed, assume that $\alpha$ has $q$ fixed points, $0 \leq q \leq p$. Then, $e(\alpha) \leq (p-q)/2$ and $\#\alpha \leq q + (p-q)/2$, which proves the inequality. Permutations $\alpha$ saturating the bound must have \emph{exactly} $q$ fixed points and $(p-q)/2$ cycles of even length, which implies that the non-trivial cycles have length 2. 
We denote the set of non-crossing partitions having only blocks of length 1 and 2 by  $NC_{1,2}(p)$ (see also \cite{aub} for another instance where this set was related to non-centered semicircular distributions). 
Then,
\be \lim_{k\to\infty}\eqref{eq:moment-ptp1} = \sum_{\alpha \in NC_{1,2}(p)} \kappa_\alpha(b_t) =  \sum_{\alpha \in NC(p)} \kappa_\alpha \left( \mathrm{SC}_{t,\sqrt{t(1-t)}} \right)  = m_p \left( \mathrm{SC}_{t,\sqrt{t(1-t)}} \right), \ee
because the first two free cumulants are respectively $\kappa_{1} (b_t) = \kappa_1 \left( \mathrm{SC}_{t,\sqrt{t(1-t)}} \right) =t$ and $\kappa_2(b_t) = \kappa_2 \left( \mathrm{SC}_{t,\sqrt{t(1-t)}} \right) =t-t^2$.
 \end{proof}
 
\begin{remark}
The convergence in distribution in the result above was also found in \cite{anv}, using operator-valued free probabilistic methods.
\end{remark}
 
\begin{remark}
Equation \eqref{eq:mu-M-t} can also be proved using the free Central Limit Theorem \cite[Theorem 8.10]{nsp}. Indeed, reorder the terms in $\mu_{k,t}^{(M\Gamma)}$ to write
\be\mu_{k,t}^{(M\Gamma)} = D_{1/\sqrt 2}D_{1/\sqrt{k(k-1)/2}} \left[\left(b_t  \boxminus b_t\right)^{\boxplus \frac {k(k-1)}2} \right] \boxplus D_{1/k}[b_t^{\boxplus k} ].\ee
The first part above is responsible for the centered semicircular part of \eqref{eq:mu-M-t}, while the second term is responsible for the shift $t$. 
\end{remark}

\subsection{Comparing the bounds}
\label{sec:compare}

In the previous three subsections and in Section \ref{sec:random-Choi}, we have computed the asymptotic limits of the bounds \eqref{eq:B-C}-\eqref{eq:B-M} in the case of large dimensional random quantum channels (recall that the bound \eqref{eq:B-I} is always trivial). Since our ultimate goal is to use these quantities and Proposition \ref{prop:bound} to lower bound minimum output R\'enyi entropies of quantum channels, we analyze in this subsection which of the four quantities yields the tightest bounds. Note that the quantity $\|M_L^\Gamma\|$ is different from the other three, since it provides a lower bound for the $p=\infty$ minimum output R\'enyi entropy of $L$; however, using the inequality $H_2 \geq H_\infty$ (see Lemma \ref{lem:H-decreasing-in-p}), we shall consider it here as a lower bound for $H_2^{\min}(L)$.

Let us start with the two bounds arising from the Choi matrix of the channel $L$. 

\begin{proposition}\label{prop:comparing-Choi}
For a sequence of random quantum channels $L_n$, the following inequality holds almost surely:
\be \lim_{n \to \infty} \|C_n^{\Gamma}\| = \|\mu^{(C\Gamma)}_{k,t}\| \leq \|\mu^{(C)}_{k,t}\| = \lim_{n \to \infty} \|C_n\|.\ee
\end{proposition}
\begin{proof}
First, note that at fixed $k$, the two norms in question, \eqref{eq:norm-mukt-1} and \eqref{eq:norm-mu-C}, are increasing functions of $t$. Then, note that the value of $t$ where the norm of the Choi matrix becomes constant $(=1)$ is smaller: $t+k^{-2} \leq t+(k+1)/(2k)$. Finally, the following remarkable identity holds: 
\begin{equation}
\forall \, t < 1- k^{-2}, \quad \|\mu^{(C)}_{k,t}\| - \|\mu^{(C\Gamma)}_{k,t}\| = k \varphi^+(k^{-2},t) - 2 \varphi^+ \left(\frac{k+1}{2k},t \right) + 1 = kt \geq 0.
\end{equation}
The proof follows now from the three facts above, see also Figure \ref{fig:bound-Choi-vs-Choi-Gamma} for the case $k=2$.
\end{proof}
\begin{figure}[htbp] 
\includegraphics{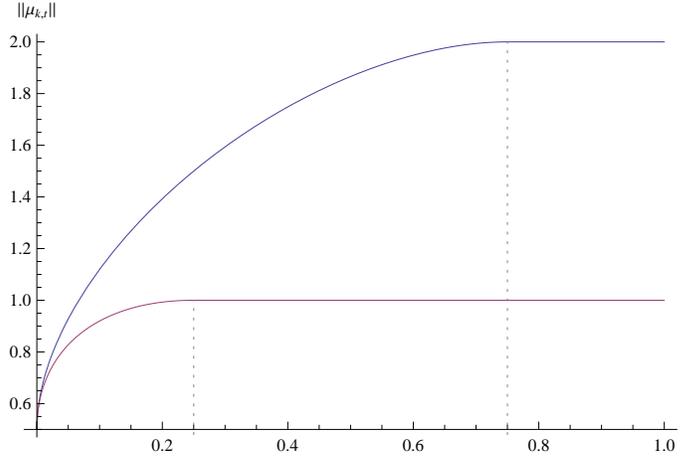} 
\caption{For large random quantum channels, the norm of the Choi matrix (blue) is larger than the norm of the partially transposed Choi matrix (red), here for $k=2$.} 
\label{fig:bound-Choi-vs-Choi-Gamma}
\end{figure}

In the previous two subsections, we were unfortunately not able to compute in full generality the asymptotic operator norm for the other two bounds, $\|C_{L^c}^\Gamma\|$ and $\| M_L^\Gamma \|$. We have to settle thus for a partial result, concerning their asymptotic behaviour in the case $1 \ll k \ll n$ (this regime corresponds to first taking the limit $n \to \infty$, followed by the limit $k \to \infty$). 

\begin{proposition}\label{prop:comparing-asymptotic}
For a sequence of random quantum channels $L_n$, in the asymptotical regime $1 \ll k \ll n$, the following inequalities hold:
\be \|\mu^{(C\Gamma)}_{t}\|  = \|\mu^{(Cc\Gamma)}_{t}\| \leq \|\mu^{(M\Gamma)}_{t}\|,
\label{eq:inequalities-measures}\ee
where 
\be \mu^{(C\Gamma)}_{t} = \lim_{k \to \infty} \mu^{(C\Gamma)}_{k,t} = 
 D_t \left[ \left(\frac{1}{2} \,\delta_{-1} +\frac{1}{2} \, \delta_{+1} \right)^{\boxplus 1/t}\right].
\label{eq:lim-mucg}\ee
In the case of the Choi matrices,  
\be \lim_{k \to \infty} \|\mu^{(C)}_{k,t}\| = + \infty.
\label{eq:lim-muC}\ee 
\end{proposition}
\begin{proof}
First,  \eqref{eq:def-mukt-1} shows the equality in \eqref{eq:lim-mucg}. 
Then, taking $k \to \infty $ in Theorem \ref{thm:strong-convergence}, 
$\| \mu^{(C\Gamma)}_{t} \|$ turns out to be the same as $\|\mu^{(Cc\Gamma)}_{t}\| $,
which is given in \eqref{eq:norm-mu-C-c-Gamma}. 
Next, \eqref{eq:norm-mu-M-t} results in the inequality in \eqref{eq:inequalities-measures}.
Finally, \eqref{eq:lim-muC} is obtained by  \eqref{eq:norm-mu-C}. 
\end{proof}
\begin{remark}
Note that, although the probability measures $\mu_t^{(C\Gamma)}$ and $\mu_t^{(Cc\Gamma)}$ are both symmetric and have the same support upper bound, they are different. Indeed, we have
\begin{equation}
\mathrm{Var}[\mu_t^{(C\Gamma)}] = t^2 \kappa_2 \left[ \left(\frac{1}{2} \,\delta_{-1} +\frac{1}{2} \, \delta_{+1} \right)^{\boxplus 1/t}\right] = t\mathrm{Var}\left[\frac{1}{2} \,\delta_{-1} +\frac{1}{2} \, \delta_{+1} \right] = t,
\end{equation}
while
\begin{equation}
\mathrm{Var}[\mu_t^{(Cc\Gamma)}] = t^2 \kappa_1\left[ b_t^{\boxplus 1/t}\right] =t^2.
\end{equation}
\end{remark}

Finally, we present in Figure \ref{fig:histograms} some numerical result in the case $k=2$, $t=0.1$, for all the bounds. In the case of $C_L$, $C_L^\Gamma$, and $M_L^\Gamma$, the $10$ random isometries with $n=2000$ were used to produce the eigenvalue plot, while for $C_{L^c}^\Gamma$ we used $10$ random isometries with $n=100$ (note that we removed from the graphs some Dirac masses at zero corresponding to rank-deficient matrices). The approximations we deduce for the norms of the matrices are presented in Table \ref{tab:norms-bounds}. In the case we consider ($k=2$), it seems that the bound corresponding to $C_{L^c}^\Gamma$ is the tightest. However, this might be due simply to the rather small value of $n$ compared to the other cases, so we do not wish to make any conjectures at this time.  
\begin{figure}[htbp] 
\includegraphics[width=0.4\textwidth]{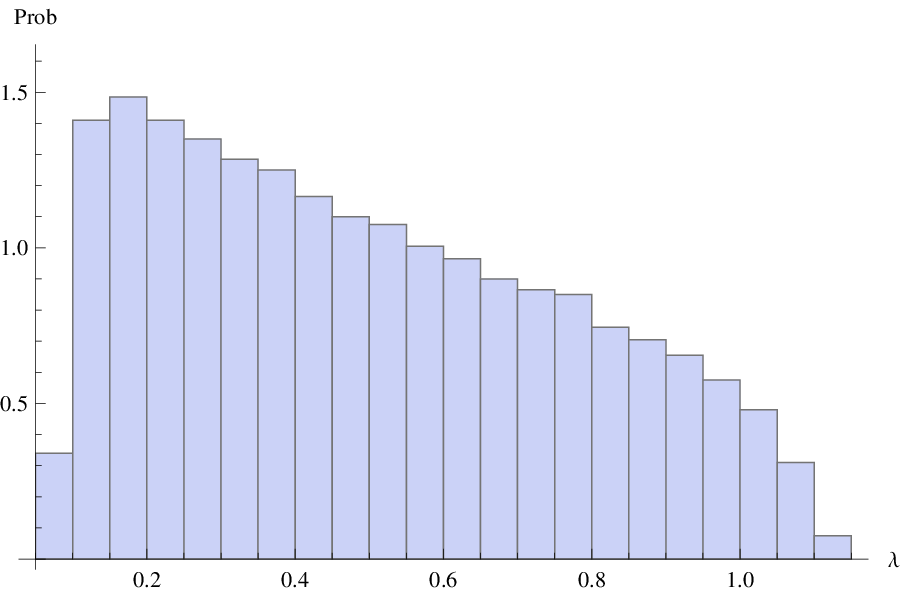} \quad 
\includegraphics[width=0.4\textwidth]{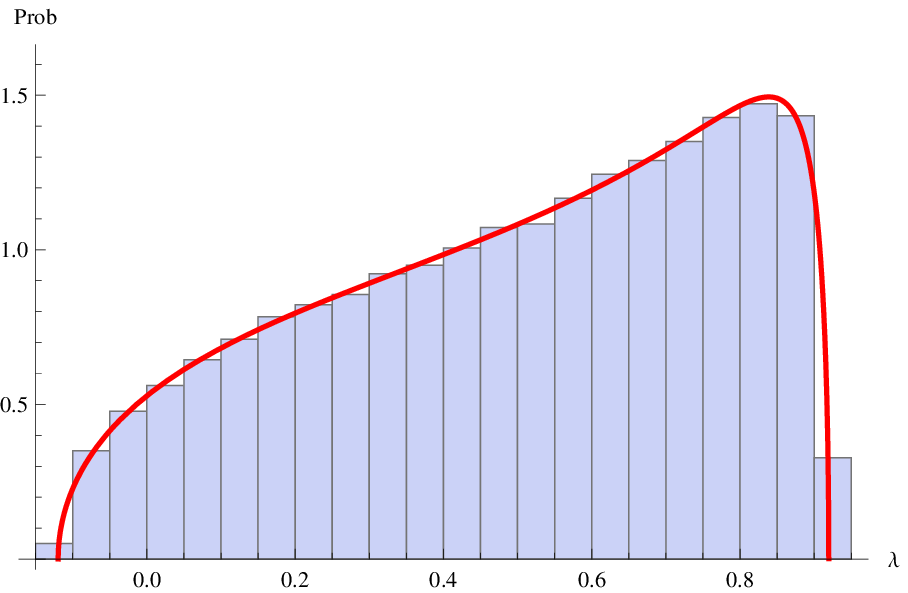} \\
\includegraphics[width=0.4\textwidth]{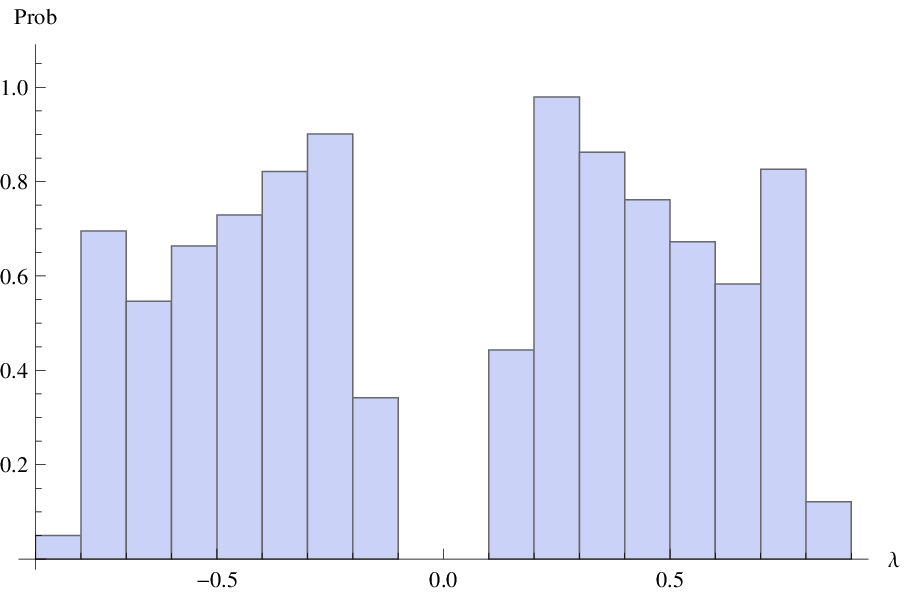} \quad 
\includegraphics[width=0.4\textwidth]{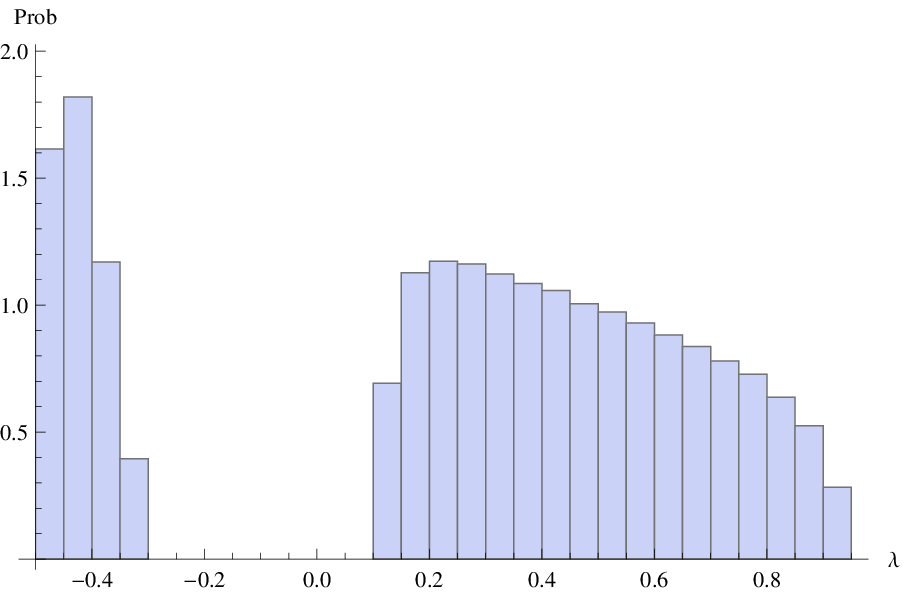} \\
\caption{Histograms for the eigenvalue distributions of, respectively,  $C_L$, $C_L^\Gamma$, and $C_{L^c}^\Gamma$, and $M_L^\Gamma$. For each histogram, we have chosen $10$ random matrices with $k=2$, $t=0.1$ and $n=2000$ (except for $C_{L^c}^\Gamma$, where $n=100$).}
\label{fig:histograms}
\end{figure}

\begin{table}[htdp]
\renewcommand{\arraystretch}{1.3}
\caption{Numerical results for the four bounds, in the case $k=2$ and $t=0.1$.}
\begin{center}
\begin{tabular}{|c|c|c|}
\hline
Bound &Numerical estimate & Theoretical value  \\ \hline \hline 
$\|C_L\|$ & $1.12360$ & $1.11962$ \\ \hline
$\|C_L^\Gamma\|$ & $0.91990$ & $0.91961$ \\ \hline
$\|C_{L^c}^\Gamma\|$ & $0.85758$ & -- \\ \hline
$\|M_L^\Gamma\|$ & $0.94548$ & -- \\ \hline
\end{tabular}
\end{center}
\label{tab:norms-bounds}
\end{table}

\bigskip

We would like to conclude this section with a discussion on the optimality of the four bounds, in the case of random quantum channels. From a practical point of view, note that we have explicit formulas, at fixed $k$ and $t$, only for the two bounds corresponding to Choi matrices, $\|C_L\|$ and $\|C_L^\Gamma\|$; among the two, Proposition \ref{prop:comparing-Choi} shows that the bound for the partial transposition is always tighter. In the asymptotical regime where  $t$ is fixed and $k \to \infty$, in Proposition \ref{prop:comparing-asymptotic} we show that the bound $\|C_L^\Gamma\|$ is tighter that $\|M_L^\Gamma\|$ (which is the bound used by Montanaro in \cite{mon}), while the two bounds corresponding to partially transposed Choi matrices, $\|C_L^\Gamma\|$ and $\|C_{L^c}^\Gamma\|$, perform equally well. Numerical simulations for $k=2$ and $t=0.1$ seem to suggest that $\|C_{L^c}^\Gamma\|$ performs better in this particular case; however, we do not consider this numerical data conclusive, because of the small value of the parameter $n$ that was used to obtained them, due to machine memory limitations.

For these reasons, we choose to work in the next sections with the bound $\|C_L^\Gamma\|$, corresponding to the partial transposition of the Choi matrix, see \eqref{eq:norm-mukt-1}-\eqref{eq:norm-mukt-2}.

\section{Minimum output entropies for a single random quantum channel}
\label{sec:MOE}

In this section we recall some upper bounds for minimum output entropies of random quantum channels we shall use in what follows. The following fact is a collection of results from \cite{bcn13,CollinsNechita2011}:

\begin{theorem}\label{thm:bcn}
For all $p \in [0,\infty]$ and for almost all sequences of random quantum channels $(L_n)_n$, we have
\begin{equation}\label{eq:upper-bound-h}
\limsup_{n \to \infty}H^{\min}_p (L_n) \leq H_p(x_{k,t})=:h_{p,k,t},
\end{equation}
where 
\be 
x_{k,t} = \left(y, \underbrace{\frac{1-y}{k-1}, \ldots, \frac{1-y}{k-1}}_{k-1 \text{ times}} \right)
\label{eq:single-opt}
\ee
with 
\begin{align}
y &= \max \operatorname{supp} (b_t \boxtimes b_{1/k}) = \max(1,\varphi^+(t,1/k)) \notag\\
&= \min\left[1, t + \frac{1}{k} - 2 \frac{t}{k}+ 2\sqrt{t(1-t)\frac{1}{k}\left(1-\frac{1}{k}\right)} \right].
\end{align}
This statement also holds for the sequence of complementary channels $(L^C_n)_n$.

Moreover, for $p \in [1,\infty]$, the above inequality is an equality, and $\limsup$ can be replaced by $\lim$ in \eqref{eq:upper-bound-h}.
\end{theorem}
\begin{proof}
In \cite[Theorem 4.1]{CollinsNechita2011} it is shown that the largest eigenvalue of an output of a random quantum channel is at most $y$, and that the value $y$ is almost surely attained. Given this partial information, the upper bound \eqref{eq:upper-bound-h} follows from the concavity of $p$-R\'enyi entropies for $p \in [0, \infty]$ (to maximize entropy, the smaller eigenvalues should be identical). The second part of the statement follows from the finer analysis in \cite[Theorem 5.2]{bcn12}, where it is shown that the eigenvalue vector $x_{k,t}$ above is the one which achieves the minimum entropy among outputs of the random quantum channel, in the case where $p \geq 1$. 
\end{proof}
Let us now study the asymptotics of the above upper bound, in the regime $k \to \infty$. Let us remind the reader that these results concern quantities for which the limit $n \to \infty$ has already been taken; in other words, we are considering the asymptotical regime $1 \ll k \ll n$. 
\begin{corollary}\label{cor:bcn}
In the setting of Theorem \ref{thm:bcn}, for fixed $p\in [0,\infty]$, asymptotically as $k \to \infty$, we have:
\begin{enumerate}[i)]
\item 
When 
$0<t<1$ is a fixed constant, 
\be h_{p,k,t} = o(1) + 
\begin{cases}
\displaystyle  \frac p {1-p} \log t, &\quad \text{ if } p >1\\[5pt]
   (1-t) \log k - t \log t - (1-t) \log(1-t) , &\quad \text{ if } p =1\\[5pt]
 \displaystyle  \log k + \frac{p}{1-p} \log(1-t) , &\quad \text{ if } 0 \leq p < 1.
\end{cases}\ee

\item When $t \asymp  k^{-\tau}$ for some constant $\tau >0$,
\be h_{p,k,t} = o(1) +  \begin{cases} \displaystyle  \frac {\tau p}{p-1} \log k, &\quad \text{ if }  \displaystyle  0<\tau \leq 1- \frac 1 p\\ 
 \log k, &\quad \text{ if } \displaystyle 1- \frac 1 p< \tau. 
\end{cases}\ee
\end{enumerate} 
\end{corollary} 
\begin{proof}
First, note that since we are in the large $k$ regime, $y=\varphi^+(t,1/k) < 1$. 
In the case where $t$ is fixed, we have $y=t + O(k^{-1/2})$ so that, depending on the value of $p$, the main contribution to the quantity $\|x_{k,t}\|_p^p$ is given either by $y$ (when $p>1$), $(1-y)/(k-1)$ (when $p<1$), or by the whole vector $x_{k,t}$ (when $p=1$). 
In the second case, where $t$ scales like $k^{-\tau}$, we have that $y \asymp   k^{-\tau} + k^{-1} $, which implies that 
\be
\|x_{k,t}\|^p_p \asymp  k^{-p \tau} + k^{-p+1}.
\ee
We conclude by taking logarithms of the expressions above.
\end{proof}

\begin{remark}
In Corollary \ref{cor:bcn} the case when $\tau=1$ can be derived from \cite{ASW11}.
Also, the case $\tau \geq 1$ can be treated via $\max_{X \in \mathcal M_d^{1,+} } \|L(X)-I/k \|_2 \asymp \sqrt{\frac t k}+ t$ \cite{fuk13}
which means that all the output states are highly mixed. 
\end{remark}

\section{Additivity rates of random quantum channels}
\label{sec:additivity-rates-random}

This section contains one of the main results of our work, lower bounds for the additivity rates of random quantum channels. We shall use Proposition \ref{prop:bound} and the estimates from Sections \ref{sec:random-Choi}. Indeed, in the following sections we will only consider the bound given by the operator norm of the partially transposed Choi matrix of a quantum channel, for the reasons discussed in Section \ref{sec:compare}.

\subsection{Minimum output R\'enyi entropy}

The following result is a direct consequence of the bound in Proposition \ref{prop:bound}, the strong convergence result in Theorem \ref{thm:strong-convergence} and Lemma \ref{lem:H-decreasing-in-p}.

\begin{theorem}\label{thm:bound-product}
Fix a positive integer $r$ and a R\'enyi entropy parameter $p \in [0,2]$. Then, almost surely, as $n \to \infty$,
\be \lim_{n \to \infty} H^{\min}_p (L_n^{\otimes r}) \geq -r \log \| \mu_{k,t}^{(C\Gamma)}\|,\ee
where $\| \mu_{k,t}^{(C\Gamma)}\|$ was given in \eqref{eq:norm-mukt-1}-\eqref{eq:norm-mukt-2}.
This statement also holds for the sequence of complementary channels $(L^C_n)_n$.
\end{theorem}

\begin{remark}
In the case $p >2$, the inequality $\|x\|_p \leq \|x\|_2$ must be used, and thus a correction factor appears in the bound:
\be \lim_{n \to \infty} H^{\min}_p (L_n^{\otimes r}) \geq -r \log \| \mu_{k,t}^{(C\Gamma)}\| \cdot \frac{p}{2(p-1)}.\ee
\end{remark}

Next,  we investigate how the quantity $\|\mu_{k,t}^{(C\Gamma)}\|$ behaves when $k \to \infty$ (recall that we write $x_n \asymp y_n$ when $\lim_{n \to \infty} x_n / y_n \in (0,\infty)$).
\begin{corollary}\label{cor:strong-convergence}
In the setting of Theorem \ref{thm:strong-convergence} we have the following asymptotic behaviors as  $k \to \infty $.
\begin{enumerate}[i)]
\item When $t \geq 1/2$ is a constant, we have, for any $k$,
\be \|\mu_{k,t}^{(C\Gamma)}\|  = 1.\ee
\item When $0<t <1/2$ is a constant, we have  
\be \|\mu_{k,t}^{(C\Gamma)}\|  = 2 \sqrt{t(1-t)} + o(1).\ee
\item 
When $t \asymp k^{-\tau}$  and $0< \tau \leq 2$, we have
\be
\|\mu_{k,t}^{(C\Gamma)}\|  \asymp k^{-\frac {\tau}2 } .
\ee
\item When $t \asymp k^{-\tau}$  and $\tau > 2$, we have
\be
\|\mu_{k,t}^{(C\Gamma)}\|  \asymp k^{-1}.
\ee

\end{enumerate} 
\end{corollary}

Finally, Theorems \ref{thm:bound-product}, \ref{thm:bcn}, and Corollaries \ref{cor:strong-convergence}, \ref{cor:bcn} immediately give one of the main results of this paper. The constants $\alpha^\Gamma_{p,k,t}$ below are almost sure limits of the lower bounds $\alpha_p^\Gamma(L_n)$ defined in Proposition \ref{prop:bound-alpha}.
\begin{theorem}\label{thm:wad}
For any $p \in [0,2]$, almost surely as $n \to \infty$, the $p$-additivity rates of random quantum channels $L_n$ are lower bounded by the constants
\begin{equation}\label{eq:wad-constant}
\alpha_p(L_n) \geq \alpha^\Gamma_{p,k,t} := \frac{- \log \|\mu_{k,t}^{(C\Gamma)}\|}{h_{p,k,t}}.
\end{equation}
where $\|\mu_{k,t}^{(C\Gamma)}\|$ and $h_{p,k,t}$ are given in \eqref{eq:norm-mukt-1}-\eqref{eq:norm-mukt-2} and \eqref{eq:upper-bound-h}. 
For example, in the case of the von Neumann entropy ($p=1$), we obtain
\begin{equation}
\alpha_1(L_n) \geq \frac{-\log  \left[\frac{1-2t}{k} + 2\sqrt{\left(1-\frac 1 {k^2}\right) t(1-t)}\right]}{-y \log y - (1-y) \log \frac{1-y}{k-1}}\mathbf{1}_{t < (k-1)/(2k)},
\end{equation}
where $y = \varphi_+(t, 1/k)$.

Again, these statements hold for the sequence of complementary channels $(L^C_n)_n$.
\end{theorem}

\begin{corollary}\label{cor:wad}
The additivity rate lower bounds $\alpha^\Gamma_{p,k,t}$ obtained in the theorem above have the  following behaviour as $k \to \infty$:
\begin{enumerate}[I)]
\item When $t \geq 1/2$ is a constant, then $\alpha^\Gamma_{p,k,t} = 0$ for all $p$ and, actually, for all $k$.
\item When $0<t<1/2$ is a constant,
\be \alpha^\Gamma_{p,k,t} = o(1) + \frac{p-1}{2p} \left [ 1 + \frac {2 \log 2 + \log (1-t)}{\log t} \right] \cdot \mathbf{1}_{(1,2]}.
\ee
\item When $t \asymp k^{-\tau}$ with $\tau >0$,
\be \alpha^\Gamma_{p,k,t} = o(1) +  \begin{cases}
 \frac{p-1}{2p} &\quad \text{ if }  0< \tau \leq 1-1/p\\
 \tau  /2 &\quad \text{ if } 1-1/p \leq \tau \leq 2\\
1 &\quad \text{ if }  \tau \geq 2.
\end{cases} \ee
\end{enumerate}
\end{corollary}
The above result can be summarized using the ``phase diagram'' in Figure \ref{fig:phase-diagram-alpha}, in which the asymptotical behavior ($k \to \infty$) of the lower bound $\alpha^\Gamma_{p,k,t}$ is presented as a function of $p$ and $\tau = \log(1/t) / \log k$.

\begin{figure}[htbp]
\begin{center}
\includegraphics{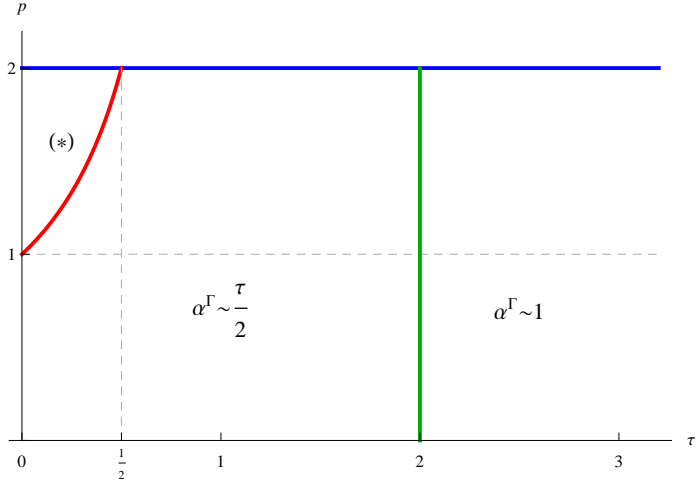}
\caption{Phase diagram for the additivity rate lower bound $\alpha^\Gamma_{p,k,t}$ as a function of $p$ and the scaling parameter $\tau$, for large $k$. In the region denoted by $(*)$, $\alpha^\Gamma_{p,k,t}$ behaves like $(p-1)/(2p)$, while the red curve in the plot is defined by $p=1/(1-\tau)$.}
\label{fig:phase-diagram-alpha}
\label{default}
\end{center}
\end{figure}

\begin{remark}
The above  lower bound  is non-trivial (i.e.~ $\alpha^\Gamma_{p,k,t} >0$) iff.~ $\|\mu_{k,t}^{(C\Gamma)}\| < 1$, which is equivalent to the condition $t < 1-s = (k-1)/(2k) < 1/2$.
\end{remark}

\subsection{Additivity rates versus weak multiplicativity exponents} \label{sec:Montanaro}
In this section, we compare the additivity rates from Theorem \ref{thm:wad} with Montanaro's results from \cite{mon}. First, let us comment on some major differences between the two approaches. First and foremost, the asymptotic regimes for random quantum channels are different: whereas we consider sequences of random quantum channels with \emph{fixed} output dimension, in \cite{mon} the author assumes that both the input and the output dimensions of the channels grow to infinity; for this reason, it is impossible to compare the results from a strictly mathematical perspective. Second, Montanaro's approach is to bound the additivity rate of a channel in terms of the matrix $M^\Gamma$ defined in \eqref{eq:B-M}; in this work, we introduce three additional bounds and we argue that in the fixed output dimension regime, the bound corresponding to the partial transposition of the Choi matrix outperforms the one from \cite{mon}. The new quantities we introduce are interesting also for the reason that they bound the $2$ R\'enyi entropy, whereas the quantity from \cite{mon} bounds the $\infty$-R\'enyi entropy. Finally, we not only compute the limiting operator norms of the relevant random matrices, but in most cases we compute also the limiting eigenvalue distributions (see Section \ref{sec:PPT} for an application). 

In spite of the fact that the current paper and \cite{mon} consider different asymptotic regimes, we would like to perform a heuristic benchmark of the additivity rates obtained. The following result is an adaptation of \cite[Theorem 3]{mon} for the $p=2$ R\'enyi entropy.
\begin{proposition}
We follow the notations in Section \ref{sec:channel}.
Suppose $k\leq n$, $\min \{d,k\} \geq 2 (\log_2 n)^{3/2}$ and $d=o(kn)$. 
Then, typically the additivity rate is bounded:
$\frac 1 r H^{\min}_2 (L^{\otimes r}) \geq \beta H^{\min}_2 (L) $ where
\be
\beta \sim \begin{cases} 1/4 & \text{if } d \geq n/k\\ 1/2 & \text{if } d \leq n/k\end{cases}
\ee
\end{proposition} 
In this proposition, if we take $n = 2^{(k/2)^{2/3}}$ this brings it to our setting  $1 \ll k\ll n$; note however that we are taking the two limits separately (first $n \to \infty$ and then $k \to \infty$) whereas in \cite{mon} the author considers the more general situation where both $k$ and $n$ grow at the same time (but at different speeds). 
By solving the equation (we have $t=k^{-\tau}$ in our mind)
\be
d= \frac n k = k^{-\tau} kn 
\ee
we get $\tau =2$. Hence,  $d \geq n/k$ and $d \leq n/k$ can be compared with our regimes $0< \tau \leq2$ and $\tau \geq 2$, respectively. Our method yields an improvement by a factor of two in this case. 

In the case of the von Neumann entropy ($p=1$), Montanaro obtains in \cite[Section 1.3]{mon} additivity rates of $1$ for $\tau =1$ and $1/2$ for $\tau \geq 2$, which are precisely the values we obtain in Corollary \ref{cor:wad}. This can be explained by the fact that in the case where $t \sim k^{-\tau}$, the norm $\|M^\Gamma\|$ from equation \eqref{eq:norm-mu-M-t} behaves like $2\sqrt t$, which is also the value from \cite[Theorem 8]{mon}. Using the monotonicity of the R\'enyi entropies as functions of $p$, one could improve the results in \cite{mon} for other values of $p$, as in the case $p=1$.

\subsection{Additivity rates for tensor products of conjugate random channels}
\label{sec:add-rate}
Besides the constructive counterexamples already discussed in the current work (see Examples \ref{ex:WH}, \ref{ex:GHP}), the most successful technique to construct quantum channels that violate the additivity relation consists of taking \emph{random conjugate channels}
(recall that for a channel $L$ as in \eqref{eq:quantum-channel-stinesrping}, its conjugate channel is obtained by replacing $V$ by $\bar V$).
In this subsection, we investigate how the additivity violations of the minimum output $p$-entropy allow to improve the additivity rate of the product channel $\alpha_p(L \otimes \bar L)$ with respect to the additivity rate of one channel $\alpha_p(L)$. These considerations are related to a conjecture by Hastings \cite{has} that the quantum channels $L \otimes \bar L$ should be generically additive. Although there is some evidence supporting this claim \cite{fne}, the conjecture is open to this date. We would like to bring further evidence to support this claim, by showing that, in some cases, the additivity rate of $L \otimes \bar L$ is larger than that of $L$, for random quantum channels $L$. 

The idea of considering conjugate channels when trying to expose additivity violations is due to Hayden and Winter \cite[Lemma 3.3]{hwi}, and it revolves around using the maximally entangled state \eqref{eq:maximally-entangled} as a test input for the tensor product between a channel and its conjugate:

\begin{lemma}\label{lem:Hayden-Winter}
Consider a quantum channel $L:\mathcal M_d(\mathbb C) \to \mathcal M_k(\mathbb C)$ which is defined via an isometry $V:\mathbb C^d \to \mathbb C^n \otimes \mathbb C^k$. Then, the output state $[L \otimes \bar L](E_d)$ cannot be too mixed, in the sense that 
\be \| [L \otimes \bar L](E_d) \| \geq \frac{d}{nk}.\ee 
The same bound holds for the complementary settings: $[L^C \otimes \overline{L^C}](E_d)$.
\end{lemma}

Since, obviously, the channels $L$ and $\bar L$ have the same additivity rates, the additivity rates of their product can be lower bounded using Proposition \ref{prop:alpha-Gamma-v-p} as follows:
\be \alpha^\Gamma_p(L \otimes \bar L) \geq v_p(L, \bar L) \alpha^\Gamma_p(L).\ee
It is intuitive now that the larger the entropy violation quotient $v_p$ is, the larger the additivity rate of the product channel will be, when compared to that of a single channel. Moreover, in order to get explicit lower bounds on the additivity rates $\alpha_p(L \otimes \bar L)$, we can lower bound the relative violation of additivity $v_p$ as follows:
\be v_p(L,\bar L) \geq \frac{2 H_p^{\min}(L)}{H_p([L \otimes \bar L](E_d))}.\ee
It follows that the quality of the lower bound $\alpha^\Gamma_p(L \otimes \bar L)$ can be in principle improved by a factor as large as $2$ when comparing to the single channel bound $\alpha^\Gamma_p(L)$.

We study next this phenomenon in the case of random quantum channels $L_n$ described in Section \ref{sec:random-Choi}. For our model of random quantum channels, the Hayden-Winter bound in Lemma \ref{lem:Hayden-Winter} was refined in \cite[Theorem 6.5]{cne10a}:

\begin{lemma}
Consider a sequence of random quantum channels $L_n$ as in Section \ref{sec:random-Choi}. Then, almost surely as $n \to \infty$, the eigenvalues of the random quantum state $[L \otimes \bar L](E_d)$ converge to the following deterministic vector:
\begin{equation}\label{eq:gamma-k-t}
\gamma_{k,t} = \left( t+\frac{1-t}{k^2}, \frac{1-t}{k^2}, \ldots, \frac{1-t}{k^2} \right) \in \mathcal M_{k^2}^{1,+}(\mathbb C).
\end{equation}
\end{lemma}

Using this lemma, we are able to quantify the improvement in the lower bound of the product channel which is provided by additivity violations. The following proposition is a direct consequence of the above discussion and the asymptotic relations already used in this section. 

\begin{proposition}\label{prop:alpha-conjugate-channels}
Consider a sequence $L_n$ of random quantum channels as in Section \ref{sec:random-Choi}.
Almost surely as $n \to \infty$, the $p$-additivity rates of random quantum channels $L_n \otimes \bar L_n$ are lower bounded by the constants 
\be \alpha_p(L_n \otimes \bar L_n) \geq  v_{p,k,t}\alpha^\Gamma_{p,k,t},\ee
where $\alpha^\Gamma_{p,k,t}$ are the single channel bounds from Theorem \ref{thm:wad}, and 
\be v_{p,k,t} = \frac{2h_{p,k,t}}{H_p(\gamma_{k,t})},\ee
where the vector $\gamma_{k,t}$ was defined in \eqref{eq:gamma-k-t}.
The same bound holds for the complementary settings: $L^C \otimes \overline{L^C}$.

In particular, the parameter $v_{p,k,t}$ behaves like $1+o(1)$ when $k \to \infty$, except in the following cases:
\begin{enumerate}[I)]
\item When $0<t<1/2$ is a constant and $p>1$, we have $v_{p,k,t} =2+ o(1)$.
\item When $t \asymp k^{-\tau}$ with $p>1$ and $0< \tau < 1-1/p$, we have $v_{p,k,t} =2 + o(1)$.
\item When $t \asymp k^{-\tau}$ with $p>1$ and $1-1/p \leq \tau \leq 2-2/p$, we have 
\be v_{p,k,t} = \frac{2p-2}{\tau p}+ o(1).\ee
\end{enumerate}
\end{proposition}

\section{Classical capacity for random quantum channels}\label{sec:cap}

In this section, we discuss how the additivity rate bounds derive earlier yield interesting upper bounds for the classical capacity of (random) quantum channels. In general, it is difficult to calculate the classical capacity for a given quantum channel \cite{BEHY2011}, so upper bounds are important in this situation. 

Let us start with an important relation between the additivity rate for the von Neumann entropy $\alpha_1$ and the classical capacity $\mathcal C_{cl}$ of a quantum channel. 
\begin{proposition} \label{prop:cap-bound}
For any quantum channel $L$, 
\begin{equation}
\mathcal C_{cl} (L) \leq H_1^{\max} (L) - \alpha_1(L) H_1^{\min}(L).
\end{equation}
In particular, 
\begin{equation}
	\mathcal C_{cl}(L) \leq H_1^{\max} (L) + \log B,
\end{equation}
where $B = \min\{\|C_L\|,\|C_L^\Gamma\|,\|C_{L^c}^\Gamma\|,\|M_L^\Gamma\| \}$.
\end{proposition}
\begin{proof}
The first inequality follows from \eqref{eq:H-max-H-min} and Proposition \ref{prop:additivity-rates}, whereas the second one follows from Proposition \ref{prop:bound}.
\end{proof}

Let us now analyze in detail the corresponding bounds for random quantum channels by using estimates developed in the previous sections. 
In our previous work \cite{cfn3}, we have investigated the Holevo quantity $\chi(\cdot)$ for random quantum channels, but we were not able to analyze $\mathcal C_{cl}(\cdot)$
because we did not have the techniques to treat output entropy of tensor powers of quantum channels. 
However, now we use Proposition \ref{prop:cap-bound} 
and get bounds for the classical capacity, as follows.

\begin{theorem}
For random quantum channels defined by \eqref{eq:quantum-channel-stinesrping},
we have, almost surely
\begin{equation}
  \limsup_{n\to\infty} \mathcal C_{cl} (L_n) \leq \log k -\log \| \mu_{k,t} \|.
\end{equation}
In particular, the capacity admits the following asymptotic bounds as  $k \to \infty $.
\begin{enumerate}[i)]
\item When $0<t <1/2$ is a constant, we have  
\be
t \log k - h(t) \leq \liminf_{n\to\infty} \mathcal C_{cl}(L_n) \leq \limsup_{n\to\infty} \mathcal C_{cl}(L_n)  \leq \log k + \log 2 + \frac 12 \log t(1-t) + o(1),
\ee
where $h(t) = -t\log t - (1-t) \log(1-t)$ is the binary entropy. 
\item 
When $t \asymp k^{-\tau}$  and $0< \tau \leq 2$, we have
\be
\limsup_{n\to\infty} \mathcal C_{cl}(L_n)   \leq  \left(1-\frac \tau 2\right) \log k + c 
\ee
for some constant $c>0$. 
\item When $t \asymp k^{-\tau}$  and $\tau > 2$, we have
\be
\limsup_{n\to\infty} \mathcal C_{cl}(L_n)   \leq   c 
\ee
for some constant $c>0$. 
\end{enumerate} 
\end{theorem}
\begin{proof}
For the upper bound, one can use Theorem \ref{thm:bound-product} and Corollary \ref{cor:strong-convergence} 
together with  \eqref{eq:H-max-H-min}. 
To show the lower bound, we claim that for our random channel $L_n$ almost surely as $n\to\infty$
\be \lim_{n\to\infty} \chi(L_n) = \log k - \lim_{n \to \infty} H_1^{\min} (L_n) = \log k - h_{1,k,t}.\ee
Indeed, the almost-sure limit image of $L_n$ is unitarily invariant (see \cite{bcn12,cfn3} for details). 
Hence, we can rotate and average an optimal output to get the maximally mixed state, whose entropy is $\log k$.
Then, Corollary \ref{cor:bcn} gives the lower bounds. 
\end{proof}

Note that the constants $c$ appearing in the result above could have been explicitly computed using Theorem \ref{thm:bound-product} and Corollary \ref{cor:strong-convergence}.

\section{PPT properties for random quantum channels}
\label{sec:PPT}
In this section, we investigate the sequence of random quantum channels $L_n$ defined in  \eqref{eq:quantum-channel-stinesrping} and find the threshold for PPT/non-PPT property. 
Also, we show existence of PPT channels which violate additivity of R\'enyi $p$ entropy. 

Recall that a quantum channel $L$ is said to have the \emph{PPT property} if its Choi matrix is PPT, i.e.~ $C_L^\Gamma \geq 0$. This is equivalent to the fact that, for any bi-partite input state $x$, the output $[\mathrm{id} \otimes L](xx^*)$ is a PPT quantum state. It follows that the class of PPT channels contains as a (strict, for large enough dimensions) subclass the set of entanglement breaking channels (for which the Choi matrix is separable). 

\subsection{Thresholds for PPT property}\label{sec:PPT-t}
\begin{theorem}\label{thm:PPT}
Let $L_n$ be a sequence of random quantum channels of parameters $k,t$ as in  \eqref{eq:quantum-channel-stinesrping}, 
and let $s=(k+1)/(2k)$ and 
\begin{equation}
t_{PPT} = \frac 1 2 - \sqrt{s-s^2} = \frac{1}{2} \left( 1- \sqrt{1-\frac{1}{k^2}}\right).
\end{equation}
If $t \in (0,t_{PPT})$ then,
almost surely, the sequence $\lambda_{\min}(C_n^\Gamma)$ converges to a positive limit (the channels being asymptotic PPT), whereas if $t \in (t_{PPT}, 1]$, then, almost surely, the sequence $\lambda_{\min}(C_n^\Gamma)$ converges to a negative limit (the channels being asymptotic non-PPT). In other words, the value $t_{PPT}$ is a \emph{threshold} for PPT channels: a random quantum channel is PPT if and only if its relative dimension $t$ of the input space is smaller than $t_{PPT}$.
\end{theorem}
\begin{proof}
We use the strong convergence proved in Theorem \ref{thm:strong-convergence}. Since the convergence is strong, the extremal eigenvalues of the partially transposed Choi matrix $C_n^\Gamma$ converge to the edges of the support of the limiting measure $\mu_{k,t}$ defined in equations \eqref{eq:def-mukt-1}-\eqref{eq:def-mukt-2}. Since we are interested in the positivity of the support, we only look at the smallest eigenvalue. We have that, almost surely, 
\be\lim_{n \to \infty} \lambda_{\min}(C_n^\Gamma) = 
\begin{cases}
2 \varphi^-(s,t)-1, &\quad \text{ if } t<s\\
-1, &\quad \text{ if } t \geq s.
\end{cases}\ee
So, for the limiting quantity to be strictly positive, both conditions $t<s$ and $\varphi^-(s,t) > 1/2$ need to be satisfied. In order to conclude, we need to show that these conditions are equivalent to the ones in the statement.

For a fixed value of $k$ (and thus $s$), the function $t \mapsto \varphi^-(s,t)$ is convex on $[0,1]$ and the equation $\varphi^-(s,t)=1/2$ has solutions $t_\pm=1/2 \pm \sqrt{s-s^2}$. Obviously, $\varphi^- (s,t) < 1/2 \iff t \in [0,t_-) \cup (t_+,1]$. A direct computation shows that $t_+>s$ iff $s>1/2+\sqrt 2 / 4$, which is equivalent to $k>\sqrt 2$, which is true for all integer $k>1$. Hence, $s \in  (t_-, t_+)$ for the relevant values of $k$, and the conclusion follows. 
\end{proof} 

\begin{remark}
In the above result, the threshold value $t_{PPT}$ behaves as $k^{-2}/4 + o(k^{-2})$ as $k \to \infty$. In that range of parameters ($t \asymp k^{-2}$ or smaller), we have also shown in Theorem \ref{thm:wad} that the channels $L_n$ are almost additive, i.e. their  additivity rate is $1 + o(1)$, for any $p \in [0,2]$.
\end{remark}

\begin{remark}\label{rk:random-density-matrices-PPT}
The above value of the threshold $t_{PPT}$ proves also the claim that the distribution of the random matrix $C_L^\Gamma$ is not that of a (rescaled) induced random density matrix \cite{zso}. Indeed, for random induced density matrices, the threshold for the PPT property has been computed in \cite[Theorem 6.2]{bn12}: $t_{PPT, induced} = 1/[4k(k-1)]$. In general,  the value for random Choi matrices is smaller than the value for the induced ensemble, $t_{PPT} < t_{PPT, induced}$, proving that the two probability distributions are different. However,  let us note that the two thresholds have the same asymptotic behaviour as $k \to \infty$. 
\end{remark}

\begin{remark}
Let us make one final remark concerning the case of the complementary channels. Indeed, it has been shown in Section \ref{sec:limiting-spectrum-C-c-Gamma} that the limiting spectral distribution of the Choi matrix of the complementary channel $L_n^c$ is \emph{symmetric} in the regime $1 \ll k \ll n$. Hence, in that regime, the complementary channel cannot be PPT. This illustrates the fact that, in general, a quantum channel and its complementary do not share the PPT property.
\end{remark}

\subsection{PPT channels violating additivity} \label{sec:PPT-v}
In this subsection, we show existence of PPT channels which violate the additivity of R\'enyi $p$ entropy with $p$ large. 
There are two theorems presented below. 
Theorem \ref{thm:ppt1} looks for the smallest possible dimension $k$, while
Theorem \ref{thm:ppt2} does for the minimum number of $p$. 
\begin{theorem}\label{thm:ppt1}
Take $k \geq 76$ and set $d= \frac n {4k}$, then for large enough $n$ and $p$, 
with high probability, $L_n$ are PPT and we have additivity violation:
\be
H^{\min}_p (L_n \otimes \bar L_n) < 2 H^{\min}_p (L_n).
\ee
\end{theorem}
\begin{proof}
First,
\be
t_{PPT} = \frac 12 \left [ 1-\sqrt{1- \frac 1 {k^2}} \right] > \frac 1 {4k^2} 
\ee
So, set $t=1/(4k^2)$ so that typical random channels are $PPT$ by Theorem \ref{thm:PPT}. 

Next, by using the result in \cite{cne10a}, almost surely
\eq{
\lim_{n \to \infty} \max_X \left\| (L_n \otimes \bar L_n) (B_n) \right\|_\infty = t + \frac {1-t}{k^2} = \frac {5} {4k^2} - \frac 1 {4k^4}
\label{eq:PPTtensor}
}
where $B_n$ are Bell states on $\mathbb C^n \otimes \mathbb C^n$. 
On the other hand, by using the result in \cite{cne10a}, almost surely
\eq{
\lim_{n \to \infty} \max_X \left\| L_n (X) \right\|_\infty 
&= t+ \frac 1 k - 2 \frac t k + 2 \sqrt{t(1-t)\frac 1 k \left(1-\frac 1k\right)} \\
&\leq \frac 1 {k} + \frac 1 {k\sqrt{k}} + \frac 1 {4k^2}  - \frac 1 {2k^3}  
\label{eq:PPTsingle}
}
This implies that, almost surely,
\eq{
\lim_{n \to \infty} H^{\min}_\infty (L_n \otimes \bar L_n) <  2  \lim_{n \to \infty}  H^{\min} _\infty (L_n) 
}
for large enough $k$.
Indeed, for $76 \geq k$, we have $\eqref{eq:PPTtensor} > \eqref{eq:PPTsingle}^2 $.
Note that since $H^{\min}_\infty (\cdot) = \lim_{p\to\infty} H^{\min (\cdot)}$ we can extend the above violation of additivity to large $p$.
The proof is complete by taking the intersection of two large-probability events. 
\end{proof} 

\begin{theorem}\label{thm:ppt2}
Set $d= \frac n {4k}$. Then, for large enough $k$ and $n$, with high probability
$L_n$ are PPT and we have additivity violation for all $p \geq 30.95$:
\be
H^{\min}_p (L_n \otimes \bar L_n) < 2H^{\min}_p (L_n).
\ee
\end{theorem} 
\begin{proof}
First, by using the second statement of Theorem \ref{thm:bcn},
almost surely
\eq{
\lim_{n\to\infty} H^{\min}_p (L_n) = h_{p,k,t}.
}
which was obtained by the output distribution $x_{k,t}$ in \eqref{eq:single-opt}.
For this distribution with $t=1/(4k^2)$, we have
\eq{
\left (\| x_{k,t}\|_p^p\right)^2  = 
 \frac{k^2}{k^{2p}} + (p^2-p) \frac 1{k^{2p}} +  O\left(\frac 1 {k^{2p+1/2}}\right)
 \label{eq:single-approx}
}
where the approximation is given by Wolfram Mathematica. 
Similarly,  for $\gamma_{k,t}$ in \eqref{eq:gamma-k-t}, 
\eq{
 \| \gamma_{k,t} \|_p^p & = \frac{1}{(k^2-1)^p} \left[ k^2 -1 - \frac {5p}{4}\right] 
+ \frac 1 {k^{2p}} \left(\frac 54\right)^p + O\left(\frac 1 {k^{2p+2}}\right) \notag \\
&=\frac 1 {k^{2p}}  \left[ k^2 -1-\frac{p}4 +  \left(\frac 54\right)^p\right]  + O\left(\frac 1 {k^{2p+1}}\right)
\label{eq:tensor-approx}}

A sufficient condition for additivity violation is that $\eqref{eq:single-approx} < \eqref{eq:tensor-approx} $, which is equivalent.
in the regime $k \to \infty$, to 
\eq{
p^2 -\frac{3p}4 +1 < \left( \frac 54 \right)^p 
}
Again, solving the equation gives $p=30.9441$ via Wolfram Mathematica.
As in the proof of Theorem \ref{thm:ppt1}, $L_n$ are typically PPT as $n\to\infty$ when $d=\frac n{4k}$.
Therefore, taking the intersection of two high-probability events completes our proof. 
\end{proof}

\bigskip

\noindent \textit{Acknowledgements.}
We would like to thank Michael Wolf for useful discussions, and the referee for a careful reading of our manuscript. M.F.~ is grateful for financial support via the CHIST-ERA/BMBF project CQC and from the John Templeton Foundation (ID$\sharp$48322).  The opinions expressed in this publication are those of the authors and do not necessarily reflect the views of the John Templeton Foundation. I.N.'s research has been supported by a von Humboldt fellowship and by the ANR projects {OSQPI} {2011 BS01 008 01},  {RMTQIT}  {ANR-12-IS01-0001-01}, and {StoQ} {ANR-14-CE25-0003-01}.


\begin{thebibliography}{99}

\bibitem{AHW2000}
Amosov, G.G., Holevo, A.S. and Werner, R.F.
{\it On some additivity problems of quantum information theory.}
 Probl. Inform. Transm. 36(4), 25 (2000).
 
\bibitem{anv}
Arizmendi, O., Nechita, I., and Vargas-Obieta, C.
{\it Block modified random matrices.}
In preparation.

\bibitem{aub} 
Aubrun, G.
{\it Partial transposition of random states and non-centered semicircular distributions.}
Random Matrices: Theory Appl., 01, 1250001 (2012).

\bibitem{ane}
Aubrun, G. and Nechita, I.
{\it Realigning random states.}
J. Math. Phys. 53, 102210 (2012).

\bibitem{ASW10}
Aubrun, G., Szarek, S., and Werner, E.
{\it Nonadditivity of R\'enyi entropy and Dvoretzky's theorem.}
 J. Math. Phys., 51(2):022102, 7, (2010).

\bibitem{ASW11}
Aubrun, G., Szarek, S., and Werner, E.
{\it Hastings's additivity counterexample via Dvoretzky's theorem.} 
Comm. Math. Phys., 305(1):85--97, (2011).

\bibitem{bn12}
Banica, T. and Nechita, I., 
{\it Asymptotic Eigenvalue Distributions of Block-Transposed Wishart Matrices.}
J. Theo. Prob., 0894-9840, 1-15 (2012).

\bibitem{BS-book}
Beck, C. and Schl\"ogl, F. 
{\it Thermodynamics of Chaotic Systems.} 
Cambridge: Cambridge University Press (1993).

\bibitem{bcn12}
Belinschi, S.T., Collins, B. and Nechita, I.
{\it Laws of large numbers for eigenvectors and eigenvalues associated to random subspaces in a tensor product.}
Inventiones Mathematicae, vol. 190, no. 3, 2012, pp. 647-697.

\bibitem{bcn13}
Belinschi, S.T., Collins, B. and Nechita, I.
{\it Almost one bit violation for the additivity of the minimum output entropy.}
Preprint arXiv:1305.1567.

\bibitem{BEHY2011}
Brand\~ao, F.G.S.L., Eisert, J., Horodecki, M. and Yang, D. 
{\it Entangled inputs cannot make imperfect quantum channels perfect.}
Phys. Rev. Lett. 106, 230502 (2011). 

\bibitem{choi}
Choi, M. D.
{\it Completely positive linear maps on complex matrices.}
Lin. Alg. Appl. Vol. 10, Iss. 3, 285–290 (1975). 

\bibitem{col}
Collins, B.
{\it Moments and Cumulants of Polynomial random variables on unitary groups, 
the Itzykson-Zuber integral and free probability.}
Int. Math. Res. Not., (17):953--982, 2003. 

\bibitem{cfn12}
Collins, B., Fukuda M., and Nechita, I.
{\it Towards a state minimizing the output entropy of a tensor product of random quantum channels.}
J. Math. Phys. 53, 032203 (2012)

\bibitem{cfn13}
Collins, B., Fukuda M., and Nechita, I.
{\it Low entropy output states for products of random unitary channels.}
Random Matrices: Theory Appl. 02, 1250018 (2013).

\bibitem{cfn3}
Collins, B.,  Fukuda, M. and  Nechita, I.
{\it On the convergence of output sets of quantum channels.}
arXiv:1311.7571 [math-ph], to appear in J. Operator Theory. 

\bibitem{cgp}
Collins, B., Gonzalez-Guillen, C., and Perez-Garcia, D.
{\it Matrix Product States, Random Matrix Theory and the Principle of Maximum Entropy.}
Comm. Math. Phys. 320 (2013), no. 3, 663?677. 

\bibitem{collinsmale}
Collins, B. and Male, C., 
{\it The strong asymptotic freeness of Haar and deterministic matrices.}
arXiv:1105.4345 [math.OA], accepted for publication to Annales Scientifique de l'ENS.

\bibitem{cne10a}
Collins, B. and Nechita, I.
{\it Random quantum channels I: Graphical calculus and the Bell state phenomenon.} 
Comm. Math. Phys. 297 (2010), no. 2, 345--370.

\bibitem{CollinsNechita2011}
Collins, B.. and Nechita, I.
{\it Random quantum channels {II}: entanglement of random subspaces,
  {R}\'enyi entropy estimates and additivity problems.}
 Adv. Math., 226(2):1181--1201, 2011.

\bibitem{cne10b}
Collins, B. and Nechita, I.
{\it Eigenvalue and Entropy Statistics for Products of Conjugate Random Quantum Channels.}
Entropy, 12(6), 1612-1631.

\bibitem{cnz}
Collins, B., Nechita, I., and {\.Z}yczkowski, K.
{\it Random graph states, maximal flow and Fuss-Catalan distributions.}
J. Phys. A: Math. Theor. 43 (2010), 275303.

\bibitem{csn}
Collins, B. and {\'S}niady, P. 
{\it Integration with respect to the Haar measure on unitary, orthogonal and symplectic group.}
Comm. Math. Phys. 264 (2006), no. 3, 773--795.

\bibitem{chl}
Cubitt, T., Harrow, A. W., Leung, D., Montanaro, A. and Winter, A. 
{\it Counterexamples to additivity of minimum output $p$-R\'enyi entropy for $p$ close to 0.}
Commun. Math. Phys. 284, 281--290 (2008).

\bibitem{DGG}
Di Francesco, P., Golinelli, O. and Guitter, E. 
{\it Meander, folding, and arch statistics.}
 Math. Comput. Modell. 26(8--10), 97--147 (1997). 
 
\bibitem{fbe}
Franz, R., and Earnshaw, B.
{\it A constructive enumeration of meanders.}
Annals of Combinatorics 6.1, 7-17 (2002).

 \bibitem{fri}
Friedland, S.
{\it Additive invariants on quantum channels and regularized minimum entropy.}
Operator Theory: Advances and Applications, Vol. 203, 237-245 (2010).
 
\bibitem{fuk}
Fukuda, M.
{\it Extending additivity from symmetric to asymmetric channels.}
 J. Phys. A: Math. Gen. 38 L753 (2005).

\bibitem{fuk13}
Fukuda, M. 
{\it Revisiting additivity violation of quantum channels.}
Comm. Maths. Phys. 332 (2),  pp 713-728 (2014).

\bibitem{fne}
Fukuda, M. and Nechita, I.
{\it Asymptotically well-behaved input states do not violate additivity for conjugate pairs of random quantum channels.}
Comm. Math. Phys. Vol. 328, No. 3 (2014), 995-1021.

\bibitem{FSn}
Fukuda, M. and \'Sniady, P.
{\it Partial transpose of random quantum states: exact formulas and meanders.}
J. Math. Phys. 54 (2013), no. 4, 042202.

\bibitem{ghp}
Grudka, A., Horodecki, M., and Pankowski, L. 
{\it Constructive counterexamples to the additivity of the minimum output R\'enyi entropy of quantum channels for all $p>2$.} 
J. Phys. A: Math. Gen. 43.42 (2010): 425304.

\bibitem{has}
Hastings, M. B.
{\it Superadditivity of communication capacity using entangled inputs.}
Nature Physics 5, 255 (2009).

\bibitem{hil}
Hildebrand, R.
{\it Positive partial transpose from spectra.}
Phys. Rev. A 76, 052325 (2007).

\bibitem{hwi}
Hayden, P. and Winter, A. 
{\it Counterexamples to the maximal p-norm multiplicativity conjecture for all $p>1$.} 
Comm. Math. Phys. 284 (2008), no. 1, 263--280.

\bibitem{Holevo1998}
Holevo, A.S. 
{\it The capacity of the quantum channel with general signal states.}
 IEEE Trans. Inform. Theory, 44(1):269--273, 1998.

\bibitem{Holevo2005}
Holevo, A.S.
{\it Additivity conjecture and covariant channels.}
International Journal of Quantum Information 2005 03:01, 41-47.

\bibitem{Holevo2005a}
 Holevo, A.S.
{\it On complementary channels and the additivity problem.}
 Prob. Th. and Appl., 51:133--143, 2005.

\bibitem{Holevo2006}
Holevo, A.S. 
{\it The additivity problem in quantum information theory.} 
 International Congress of Mathematicians. Vol. III, 999-1018, Eur. Math. Soc., Z{\"u}rich, 2006. 

\bibitem{hal}
Hall, M.J.W.
{it Random quantum correlations and density operator distributions.}
Phys. Lett. A 242, no. 3, 123--129 (1998).

\bibitem{hjo}
Horn, R. and Johnson, C.
{\it Matrix Analysis.}
Second Edition, Cambridge University Press, 2012.

\bibitem{King03}
King, C. 
{\it Maximal $p$-norms of entanglement breaking channels.}
Quantum Inf. Comput. 3, no. 2, 186–190 (2003). 

\bibitem{KMNR2007}
King, C., Matsumoto, K.,  Nathanson, M.,  and Ruskai, M.~B. 
{\it Properties of conjugate channels with applications to additivity and  multiplicativity.} 
Markov Process. Related Fields, 13(2):391--423, 2007.

\bibitem{KR2001}
King, C. and Ruskai, M.B.
{\it Minimal entropy of states emerging from noisy quantum channels.}
IEEE Trans. Inf. Th. 47(1), 192–209 (2001).

\bibitem{male}
Male, C.
{\it The norm of polynomials in large random and deterministic matrices.}
Prob. Theo. Rela. Fiel., 154: 477-532 (2012).

\bibitem{mon}
Montanaro, A.
{\it Weak multiplicativity for random quantum channels.}
Comm. Math. Phys. 319, 535-555 (2013).

\bibitem{nsp}
Nica, A. and Speicher, R.
{\it Lectures on the combinatorics of free probability.}
Cambridge Univ. Press (2006).

\bibitem{SchumacherWestmoreland1997}
Schumacher, B., and Westmoreland, M.
{\it Sending classical information via noisy quantum channels.}
Phys. Rev. A, 56(1):131--138, 1997.

\bibitem{Shor}
Shor, P.
{\it Equivalence of additivity questions in quantum information theory.}
 Comm. Math. Phys. 246(3), 453–472 (2004). 

\bibitem{Shor02}
Shor, P.
{\it Additivity of the classical capacity of entanglement-breaking quantum channels.}
J. Math. Phys. 43, no. 9, 4334–4340 (2002).

\bibitem{sya}
Smith, G., and Yard, J.
{\it Quantum communication with zero-capacity channels.}
Science, 321:1812--1815 (2008).

\bibitem{ste}
Steele, M.
{\it Probability Theory and Combinatorial Optimization.}
SIAM, 1997.

\bibitem{spring}
Stinespring, F. W.
{\it Positive functions on $C^*$-algebras.} 
 Proc. Am. Math. Soc. 6, 211–216 (1955)

\bibitem{vdn}
Voiculescu, D., Dykema, K. and Nica, A.
{\it Free random variables.}
American Mathematical Society, 1992.

\bibitem{who}
Werner, R.F., and Holevo, A.S.
{\it Counterexample to an additivity conjecture for output purity of quantum channels.}
J. Math. Phys. 43, 4353--4357 (2002).

\bibitem{wei}
Wolf, M. M. and Eisert, J. 
{\it Classical information capacity of a class of quantum channels.} 
New Journal of Physics 7.1 (2005): 93.

\bibitem{wei+}
Weingarten, D.
{\it Asymptotic behavior of group integrals in the limit of infinite rank.}
J. Math. Phys., 19(5):999--1001 (1978).

\bibitem{zpnc}
\.Zyczkowski, K., Penson, K.A., Nechita, I., and Collins, B.
{it Generating random density matrices.}
J. Math. Phys. 52, 062201 (2011).

\bibitem{zso}
\.Zyczkowski, K. and  Sommers, H.-J.
{\it Induced measures in the space of mixed quantum states.}
J. Phys. A, 34 (2001), 7111--7125.

\end{thebibliography}
\end{document}